%% file: main.tex
\documentclass[aos]{imsart}

\RequirePackage{amsthm,amsmath,amsfonts,amssymb}
\RequirePackage[authoryear]{natbib}
\RequirePackage[colorlinks,citecolor=blue,urlcolor=blue]{hyperref}
\RequirePackage{graphicx}

\usepackage{enumerate}
\usepackage{cleveref}
\usepackage{cancel}
\usepackage{mathrsfs}
\usepackage{subcaption}

\startlocaldefs
\theoremstyle{plain}
\newtheorem{theorem}{Theorem}
\newtheorem{lemma}{Lemma}
\newtheorem{proposition}{Proposition}
\newtheorem{corollary}{Corollary}
\newtheorem*{assumption*}{assumption}
\theoremstyle{remark}
\newtheorem{remark}{Remark}
\newtheorem{definition}{Definition}
\newtheorem{assumption}{Assumption}

\DeclareMathOperator*{\argmin}{arg\,min}

\def\P{P}
\def\Q{Q}

\def\Hb{G}
\def\Hc{F}

\def\R{\mathbb{R}}
\def\Re{\overline{\mathbb{R}}}

\def\Ms{\mathcal{M}}
\def\Mf{\mathcal{M}^{+}}
\def\Mp{\mathcal{P}}

\def\Dom{\operatorname{Dom}}

\def\mone{x}
\def\mtwo{z}
\def\monet{x_t}
\def\mtwot{z_t}

\def\mthree{v}

\def\mfour{y}

\def\t{\theta}
\def\tr{\theta}
\def\T{\Theta}
\def\dt{\mathrm{d}\t}
\def\d{d}

\def\rll{\delta_0}
\def\KL{\operatorname{KL}}

\def\journal@name{}

\endlocaldefs

\begin{document}

\begin{frontmatter}
\title{Hamiltonian Dynamics of Bayesian Inference\\Formalised by Arc Hamiltonian Systems}
\runtitle{Hamiltonian Dynamics of Bayesian Inference}

\begin{aug}
\author{\fnms{Takuo}~\snm{Matsubara}}
\address{School of Mathematics, The University of Edinburgh, \href{mailto:takuo.matsubara@ed.ac.uk}{takuo.matsubara@ed.ac.uk}}
\end{aug}

\begin{abstract}
	This paper advances theoretical understanding of infinite-dimensional geometrical properties associated with Bayesian inference.
    First, we introduce a novel class of infinite-dimensional Hamiltonian systems for saddle Hamiltonian functions whose domains are metric spaces.
    A flow of this system is generated by a Hamiltonian arc field, an analogue of Hamiltonian vector fields formulated based on (i) the first variation of Hamiltonian functions and (ii) the notion of arc fields that extends vector fields to metric spaces.
	We establish that this system obeys the conservation of energy.
    We derive a condition for the existence of the flow, which reduces to local Lipschitz continuity of the first variation under sufficient regularity.
    Second, we present a system of a Hamiltonian function, called the minimum free energy, whose domain is a metric space of negative log-likelihoods and probability measures.
    The difference of the posterior and the prior of Bayesian inference is characterised as the first variation of the minimum free energy.
	Our result shows that a transition from the prior to the posterior defines an arc field on a space of probability measures, which forms a Hamiltonian arc field together with another corresponding arc field on a space of negative log-likelihoods.
	This reveals the underlying invariance of the free energy behind the arc field.
\end{abstract}

%

\end{frontmatter}


\input{section_01}

\input{section_02}

\input{section_03}

\input{section_04}

\input{section_05}



\begin{supplement}
\stitle{Supplementary Material for `Hamiltonian Dynamics of Bayesian Inference Formalised by Arc Hamiltonian Systems'}
\sdescription{The supplementary material contains proofs for all the theoretical results in the main text.}
\end{supplement}


\bibliographystyle{imsart-nameyear} 
\bibliography{bibliography}       


\newpage
\input{supplement}

\end{document}

%% file: section_01.tex

\section{Introduction} \label{sec:introduction}

In Bayesian inference, statisticians elicit a prior distribution $P$ over a parameter space $\Theta$ of a model that reflects their belief in an appropriate model to describe a phenomenon of interest.
Observing data realised from the phenomenon, they update the prior distribution $P$ to the posterior distribution $P_*$ by Bayes'~rule.
Denote the negative log-likelihood by $f: \Theta \to \R$, making the associated model and data implicit\footnote{If the model $p(x \mid \theta)$ and data $\{ x_i \}_{i=1}^{n}$ are explicit, we have $f(\theta) = - \sum_{i=1}^{n} \log p(x_i \mid \theta)$. Since we assume no specific form of the model and input, we denote the negative log-likelihood simply by $f$.}.
Most generally, the posterior $P_*$ is a measure absolutely continuous to $P$ that has the Radon–Nikodym derivative
\begin{align}
	\frac{d P_*}{d P}(\t) = \frac{ \exp(- f(\t)) }{ \int_\Theta \exp(- f(\t)) d P(\t) } . \label{eq:posterior}
\end{align}
If the prior $P$ admits a probability density $p$, the posterior $P_*$ can be expressed as a probability density $p_*(\theta) \propto \exp( - f(\theta) ) p(\theta)$ as usual.
The measure $P_*$ of the form \eqref{eq:posterior} is also known as a \emph{Gibbs measure} for a \emph{potential} $f$ with a reference measure $P$ \citep{Georgii2011}.

The posterior, or the Gibbs measure, has a variational characterisation whose origin may be traced back to \cite{Gibbs1902}.
It is a unique solution of the following variational problem:
\begin{align}
	P_* =  \argmin_{Q \in \Mp} \mathcal{F}(Q) \quad \text{for} \quad \mathcal{F}(Q) = \int_{ \Theta} f(\tr) \d Q(\tr) + \KL(Q, P) \label{eq:fp_free_energy}
\end{align}
where $\mathcal{P}$ denotes a set of all probability measures on $\Theta$ and $\KL(Q, P)$ denotes the relative entropy of $Q$ from $P$.
The functional $\mathcal{F}$ is called the \emph{free energy}\footnote{More precisely, \cite{Jordan1997} called $\mathcal{F}$ the `generalised' free energy in a sense that a reference measure $P$ is used, where the (Helmholtz) free energy is recovered when $P$ is the Lebesgue measure.} in light of the thermodynamical significance.
This characterisation has been long known as the \emph{Gibbs variational principle} \citep{Ellis2006} mainly in the context of statistical mechanics \citep{Ruelle1967,Ruelle1969} and large deviation theory \citep{Dembo2009,Dupuis2011}.
Later, it was repeatedly rediscovered and has played a pivotal role in several recent advances of Bayesian inference, such as variational Bayes methods \citep{Attias1999,Blei2017}, PAC-Bayes analysis \citep{Catoni2007,Guedj2019}, and generalised Bayesian inference \citep{Bissiri2016}.
This paper establishes (i) a further variational characterisation of the posterior $P_*$ and (ii) its Hamiltonian dynamical interpretation.

Hamiltonian systems are a fundamental description of dynamical phenomena that obey the law of conservation of energy \citep{Goldstein2002}.
A flow of a Hamiltonian system is generated by the Hamiltonian vector field determined by the gradient of a Hamiltonian function.
Under sufficient regularity, the value of the Hamiltonian function remains invariant along the flow \citep{Abraham1987}.
While Hamiltonian systems are well-established for finite-dimensional topological spaces \citep{Arnold1978}, this work deals with saddle (i.e.~concave-convex) Hamiltonian functions whose domains are infinite-dimensional metric spaces.
A chief challenge in formulating Hamiltonian systems in such a setting is to have (a) an appropriate notion that acts as a gradient in an infinite-dimensional domain and (b) the well-defined evolution of a system in a metric topology.
To date, Hamiltonian systems have been extended to Banach spaces---or symplectic manifolds built on smooth Banach manifolds \citep{Lang1985} more generally---by works of \cite{Chernoff1974,Marsden1983}.
A formulation based on a norm topology is however not suitable for our main application involving dynamics of measures, as illustrated in \Cref{sec:background}.
This motivates the development of a novel class of Hamiltonian systems, termed \emph{arc Hamiltonian systems}, using (a) the \emph{first variation} of Hamiltonian functions and (b) the notion of \emph{arc fields} \citep{Calcaterra2000} that extends vector fields to metric spaces.
A flow of an arc Hamiltonian system is generated by the \emph{Hamiltonian arc field} determined by the first variation of a Hamiltonian function.
We provide a condition for the existence of the flow over $[0, \infty)$, and prove that this system obeys the conservation of energy.

We use arc Hamiltonian systems to study an infinite-dimensional geometrical property associated with Bayesian inference.
Denote $f^c(\t) := f(\t) - \int_\T f(\tr) \d P(\tr)$ for simpler presentation.
We shall see that the difference of the posterior $P_*$ and the prior $P$ is characterised as the first variation of the \emph{minimum free energy} $H(f, P)$ with respect to $f$:
\begin{align}
	P_* - P = \frac{\partial}{\partial f} H(f, P) \quad \text{for} \quad H(f, P) = - \log\left( \int_{ \Theta } \exp\left( - f^c(\tr) \right) \d P(\tr) \right) \label{eq:Gibbs_formula}
\end{align}
where the first variation here, formally defined in \Cref{sec:methodology}, is similar to one often considered in the context of gradient flows \citep{Santambrogio2015}.
We shall show that the first variation $(\partial / \partial f) H(f, P)$, coupled with the negative of the other first variation $- (\partial / \partial P) H(f, P)$ with respect to $P$, defines an arc Hamiltonian field on a metric space $(M, d)$ of negative log-likelihoods $f$ and probability measures~$P$.
It generates a flow $t \mapsto (f_t, P_t)$ of the system that follows the first-order approximation similar to Hamilton's equations in vector spaces:
\begin{align}
	(f_{t+s}, P_{t+s}) \approx (f_t, P_t) + s \left( - \frac{\partial}{\partial P} H(f_t, P_t), \frac{\partial}{\partial f} H(f_t, P_t) \right) . \label{eq:first_order_approximation}
\end{align}
Our formulation capitalises on a property that the right-hand side stays within the domain $M$ of the energy $H(f, P)$ for all $s \in [0, 1]$, due to the form of the first variations.
In such cases, the first-order approximation \eqref{eq:first_order_approximation} can determines a flow in a metric topology.

The first-order approximation \eqref{eq:first_order_approximation} of $P_{t+s}$ driven by the first variation $(\partial / \partial f) H(f_t, P_t) = P_* - P_t$ is a transition from $P_t$ to $P_*$ over $s \in [0, 1]$ in the mixture geodesic $(1 - s) P_t + s P_*$, where the posterior $P_*$ is defined for $(f_t, P_t)$.
By the belief updating mechanism of Bayesian inference, we can assign this transition $(1 - s) P + s P_*$ over $s \in [0, 1]$ to every probability measure $P$.
Our result confers the geometrical significance of such an assignment at every $P$ as a Hamiltonian arc field, by introducing another transition $f - s (\partial / \partial P) H(f, P)$ driven by the other first variation $(\partial / \partial P) H(f, P)$.
In particular, this reveals an underlying invariance of the minimum free energy $H(f, P)$ behind the arc field induced by the belief updating mechanism.
See \Cref{fig:illustration} for illustration of this dynamics.

Certainly, there can exist various geodesics to define a transition from $P$ to $P_*$.
Finally, it is intriguing to highlight a contrast with a Wasserstein gradient flow of the free energy $\mathcal{F}$ \citep{Jordan1997}. 
The Wasserstein dynamics describes the evolution of a probability measure $Q_t$ over time $t$ that minimises the free energy $\mathcal{F}(Q_t) \to \mathcal{F}(P_*)$ for a fixed potential $f$ with reference measure $P$.
We have the same Wasserstein dynamics for a class of free energies $\mathcal{F}$ of potentials $f - C$ up to a constant $C$\footnote{This is because the minimiser $P_*$ is invariant to any constant term of the potential $f$ due to the normalisation that cancels out the constant $C$ as follows: $\frac{d P_*}{d P}(\t) = \frac{ \exp(- f(\t) + C ) }{ \int_\Theta \exp(- f(\t) + C) d P(\t) } = \cancel{\frac{\exp(C)}{\exp(C)}} \frac{ \exp(- f(\t) ) }{ \int_\Theta \exp(- f(\t)) d P(\t) }$.}.
We suppose that the free energy $\mathcal{F}$ is defined with a `mean-zero' potential $f^c(\t) =  f(\t) - C$ with $C = \int_\T f(\tr) \d P(\tr)$.
The minimum free energy $H(f, P)$ equals to the free energy $\mathcal{F}(P_*)$ at the minimiser $P_*$:
\begin{align*}
	H(f, P) = - \log\left( \int_{ \Theta } \exp( - f^c(\tr) ) \d P(\tr) \right) = \int_{ \Theta} f^c(\tr) \d P_*(\tr) + \KL(P_*, P) = \mathcal{F}(P_*) 
\end{align*}
where this equality is known as the Gibbs or Donsker-Varadhan variational formula \citep{Dupuis2011}.
In contrast to the Wasserstein dynamics that minimises the free energy $\mathcal{F}(Q_t) \to \mathcal{F}(P_*)$, our result reveals the existence of the dynamics of the potential $f_t$ and reference measure $P_t$ that keeps the minimised free energy $H(f_t, P_t) = \mathcal{F}(P_*)$ constant.

The rest of the paper is structured as follows.
\Cref{sec:background} introduces basic notations and useful preliminaries for the development of arc Hamiltonian systems.
In \Cref{sec:methodology}, we establish the framework of arc Hamiltonian systems, providing the theory on the existence and the energy-conserving invariance of their flow.
The general framework of arc Hamiltonian systems may be of independent interest in many fields, such as convex analysis, information theory, and statistical physics.
In \Cref{sec:application}, we present the system of the minimum free energy formalised by arc Hamiltonian systems.
We draw our conclusion in \Cref{sec:conclusion}. 

\begin{figure}
	\centering
	\includegraphics[width=\textwidth]{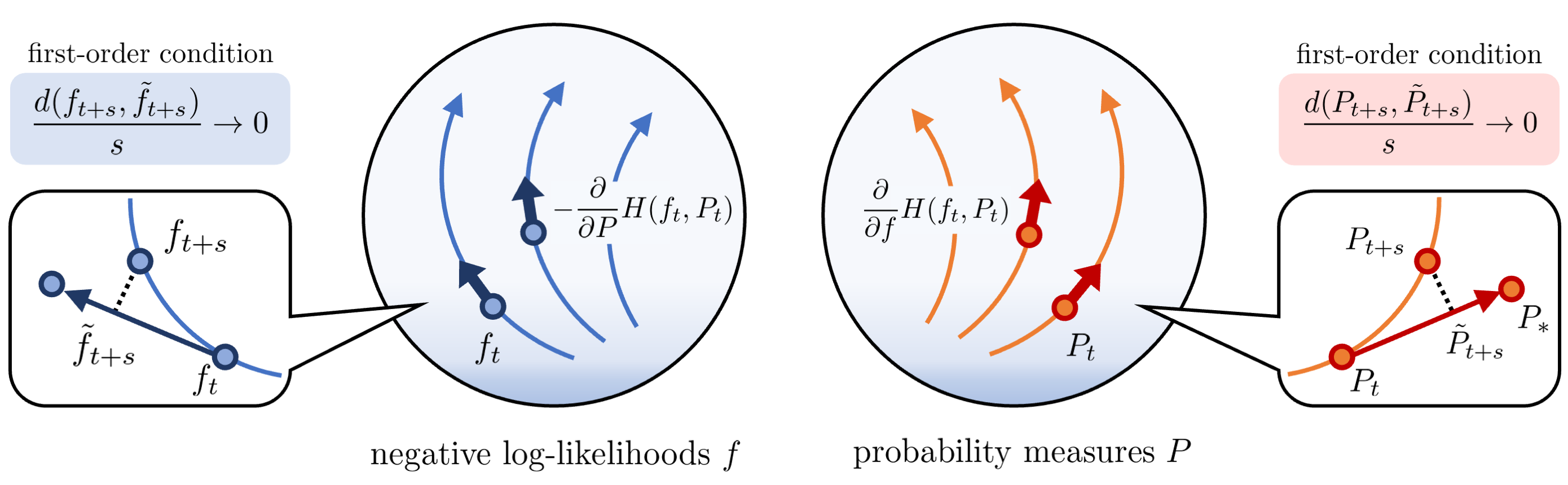}
	\caption{Illustration of the arc Hamiltonian system of the minimum free energy. Any generated flow $t \mapsto (f_t, P_t)$ of the system over $t \in [0, \infty)$ conserves the minimum free energy $H(f_t, P_t)$.} \label{fig:illustration}
\end{figure}

%% file: section_02.tex

\section{Background} \label{sec:background}

We provide useful preliminaries to develop arc Hamiltonian systems.
\Cref{sec:notation} introduces a set of basic notations and terminologies.
\Cref{sec:hamiltonian_system} briefly recaps a standard Hamiltonian system in the Euclidean space $\R^d$, illustrating a few limitations of an infinite-dimensional extension based on a norm topology for our application.
Finally, \Cref{sec:arcfields} recaps the notion of arc fields established by \cite{Calcaterra2000}, which we capitalise on to formulate the evolution of an arc Hamiltonian system in a metric topology.

\subsection{Notations and Terminologies} \label{sec:notation}

We introduce basic notations and terminologies.

\paragraph*{Functional-Analytic}
By a ``curve", we mean a map from a time interval to a metric space, where it can be equivalently thought of as a time-dependent variable in a metric space.
We use subscripts to indicate the time argument of a curve; for example $t \mapsto \mone_t$ denotes a curve.
A curve is said continuous if it is a continuous map.
We denote by $\lim_{s \to 0^+}$ the limit with respect to a scalar $s \in \R$ approaching to $0$ from the right.
For any metric space $(M, d)$, we denote by $B_r(\mone)$ an open ball in $(M, d)$ around a point $\mone \in M$ with a radius $r$.
For any map $\Phi: M \times [0, 1] \to M$ defined for a set $M$ and the time interval $[0, 1]$, we use subscripts to indicate the time argument of $\Phi$, that is, $\Phi_s(x)$ denotes the value at $(x, s) \in M \times [0, 1]$.

\paragraph*{Measure-Theoretic}
Let $\T$ be a separable completely metrizable space (i.e.~a Polish space).
By a ``measure", we mean a non-negative Borel measure on $\T$. 
By a ``signed measure", we mean a signed Borel measure on $\T$.
Any signed measure $P$ admits the Jordan decomposition $P = P^+ - P^-$ by two uniquely associated (non-negative) measures $P^+$ and $P^-$. 
For each signed measure $P$, we define a measure $| P |$---often referred to as the variation of $P$--- by $| P | := P^+ + P^-$.
This means that $\int_{O} \d | P |(\t) = \int_{O} \d P^+(\t) + \int_{O} \d P^-(\t)$ for any Borel set $O \subseteq \T$.
A signed measure $P$ is said finite if and only if the measure $| P |$ is finite.
If a measure $P$ admits the Radon-Nikodym derivative $h$ with respect to a measure $Q$, we indicate the relationship by $d P(\t) = h(\t) d Q(\t)$ for better presentation.
Denote by $\Ms$, $\Mf$, and $\Mp$ a set of all finite signed measures, all finite measures, and all probability measures.

\subsection{Hamiltonian Systems} \label{sec:hamiltonian_system}

Hamiltonian systems have an extensive history of the development; we refer readers to \cite{Goldstein2002} for a detailed introduction of classical mechanics and to \cite{Arnold1978} for the rigorous mathematical foundation.
Hamiltonian systems enable us to describe the evolution of the objects through the first-order evolution equation, by introducing a variable called the momentum.
For illustration, consider an object moving in $\R^d$ whose location at time $t$ is denoted $\monet \in \R^d$. 
The momentum, denoted $\mtwot \in \R^d$, can be associated with a physical quantity calculated by the product of the mass and velocity at time $t$.
A Hamiltonian function $H: \R^d \times \R^d \to \R$ corresponds to the total energy that the system has at each state $(\mone, \mtwo)$.
How the pair of the location and momentum $(\monet, \mtwot)$ evolves over time $t$ is determined through \emph{Hamilton's equation}:
\begin{align}
	\frac{d}{d t} ( \monet, \mtwot ) = ( - \nabla_{2} H(\monet, \mtwot), \nabla_{1} H(\monet, \mtwot) ) \label{eq:hamiltonian_eq_motion}
\end{align}
where $\nabla_{1}$ and $\nabla_{2}$ denote the gradient operator with respect to the first and second argument.
Hamilton's equation \eqref{eq:hamiltonian_eq_motion} is a differential equation $(d / dt) ( \monet, \mtwot ) = F(\monet, \mtwot)$ driven by the \emph{Hamiltonian vector field} $F(\mone, \mtwo) = ( - \nabla_{2} H(\mone, \mtwo), \nabla_{1} H(\mone, \mtwo) )$, which can be solved as an initial value problem.
Under sufficient regularity, it generates a unique solution flow $t \mapsto ( \monet, \mtwot )$ from any initial state $(\mone_0, \mtwo_0)$ in the domain $\R^d \times \R^d$.
In particular, any flow $t \mapsto ( \monet, \mtwot )$ generated by the Hamiltonian vector field conserves the energy, meaning that the value of the Hamiltonian function $H(\monet, \mtwot)$ remains constant along the flow.
\Cref{fig:pendulum} illustrates a Hamiltonian system of a simple pendulum; it can be observed that the values of the Hamiltonian function remain constant along every generated flow.

\begin{figure}[t]
	\hfill
	\includegraphics[width=150pt]{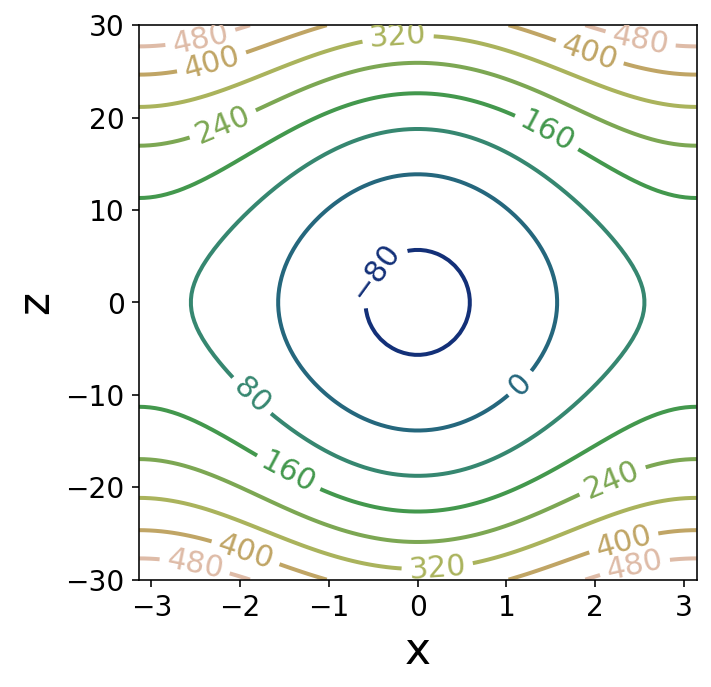}
	\hfill
	\includegraphics[width=150pt]{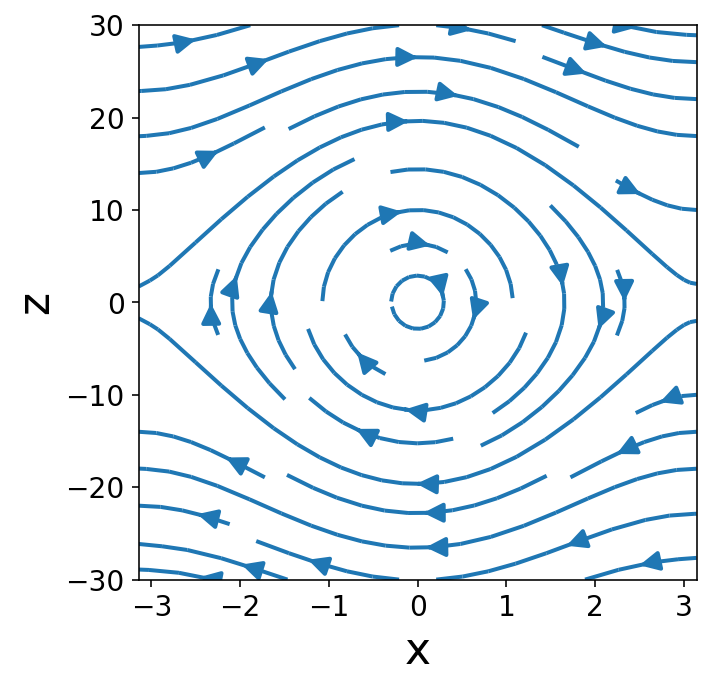}
	\hfill
	\hfill
	\caption{The Hamiltonian system of a simple pendulum whose rod length and mass are both $1$. (Left) The Hamiltonian function $H(x, z) = z^2 / 2 - g^2 \cos(x)$ for each angle $x \in [- \pi, \pi]$ and momentum $z \in \R$ of the pendulum, where $g$ is the standard acceleration of gravity. (Right) A flow by Hamilton's equation $(d / d t) (x_t, z_t) = (z_t, - g^2 \sin(x_t))$ from each initial state.} \label{fig:pendulum}
\end{figure}

\cite{Chernoff1974}---and later \cite{Marsden1983} for application to mechanics of elastic materials---established an infinite-dimensional extension of Hamiltonian systems to Banach spaces and, more generally, to symplectic manifolds built on smooth Banach manifolds.
Here, we do not consider the case of Banach manifolds.
As \cite{Kriegl1997} and \cite{Alexander2022} pointed out, there are known limitations of Banach manifolds.
For example, \cite{Eells1970} proved that an infinite-dimensional smooth Banach manifold is smoothly embedded as an open set in a Banach space used to define the atlas\footnote{Although an atlas of a Banach manifold can be defined using multiple Banach spaces in principle, the Banach spaces turn out to be all isomorphic in most cases to guarantee Fr\'{e}chet-differentiability. This means that it suffices to use one Banach space to define an atlas of a Banach manifold \citep[see][p.22]{Lang1985}.}.
Analysis on smooth Banach manifolds then reduces to that on Banach spaces.

Consider a setting where an object of interest $(\monet, \mtwot)$ belongs to a Banach space $E_1 \times E_2$.
In the extended Hamiltonian system, the evolution of the system is governed by an evolution equation $(d / dt) ( \monet, \mtwot ) = F(\monet, \mtwot)$ driven by a map $F: E_1 \times E_2 \to E_1 \times E_2$ which depends on the gradient of the Hamiltonian function $H(\mone, \mtwo)$ analogous to \eqref{eq:hamiltonian_eq_motion} but taken in a sense of the Fr\'{e}chet derivative under appropriate conditions of $E_1$ and $E_2$.
By the Cauchy-Lipschitz theorem \citep[e.g.][]{Brezis2011}, the evolution equation admits a unique solution from any initial state if the map $F$ is Lipschitz continuous in $E_1 \times E_2$.
However, this formulation based on a norm topology is not suitable for dynamics involving the set of measures $\Mf$.
In a Banach space of signed measures $\Ms$ under e.g.~the total variation norm $\| \cdot \|_\Ms$ \citep{Cohn2013}, it is easy to prove that the subset $\Mf$ has no interior point.
That is, if we use the norm topology on $\Ms$, no open set can be defined inside $\Mf$ to use important topological notions---such as the Fr\'{e}chet derivative or local Lipschitz continuity---of functions defined only on $\Mf$.
One trivial approach to define open sets inside $\Mf$ is to use a metric topology on $\Mf$ instead of the norm topology introduced to the vector space $\Ms$.
For example, the set $\Mf$ will be a metric space $(\Mf, d)$ under the total variation metric $d$ on $\Mf$ defined by restricting the norm distance $\| P - Q \|_\Ms$ to $\Mf \times \Mf$.
This motivates us to introduce a metric on the domain of interest and define the evolution within the metric space.
To achieve this, we capitalise on the notion of arc fields \citep{Calcaterra2000} recapped next.

\subsection{Arc Fields} \label{sec:arcfields}

The notion of arc fields \citep{Calcaterra2000}, that extends vector fields to metric spaces, is used to formulate an analogue of Hamiltonian vector fields and Hamilton's equation in metric spaces.
The underlying idea behind arc fields has a close connection to mutational equations \citep{Aubin1999,Lorenz2010} whose common application includes the evolution of set-valued objects and objects that belongs to complicated regions of vector spaces \citep{Aubin1993}.
To illustrate the idea, consider a vector field $F: E \to E$ in a Banach space $( E, \| \cdot \|_E )$ and a solution curve $t \mapsto \mone_t$ of the corresponding evolution equation $(d / dt) x_t = F(x_t)$.
At each $x \in E$, define a curve $\Phi_s(\mone) = x + s F(x)$ over time $s \in [0, 1]$.
Informally, the evolution equation can be thought of as the first-order approximation of the solution curve $\mone_{t + s} \approx \Phi_s(\mone_t)$ at each $t$ over small $s$, specified by the curve $s \mapsto \Phi_s(\mone_t)$:
\begin{align*}
	\lim_{s \to 0} \left \| \frac{x_{t+s} - x_t}{s} - F(x_t) \right\|_E = \lim_{s \to 0} \frac{ \| x_{t+s} - \Phi_s(\mone_t) \|_E }{| s |} = 0 . 
\end{align*}
A key observation is that the vector-space structure of the set $E$---where the curve $t \mapsto x_t$ belongs---is not essential in this formulation.
In the above condition, consider that the vector space $E$ is replaced with any set, such as $\Mf$, and the norm $\| x_{t+s} - \Phi_s(\mone_t) \|_E$ is replaced with a metric $d(x_{t+s}, \Phi_s(\mone_t))$ on the replaced set $E$.
The above condition can be still well-defined if we have a premise that the curve $s \mapsto \Phi_s(\mone_t)$ stays within the set $E$ over all $s$ in some interval.
This observation leads to the extension of vector fields to metric spaces.

\begin{figure}[t]
	\centering
	\includegraphics[height=125pt]{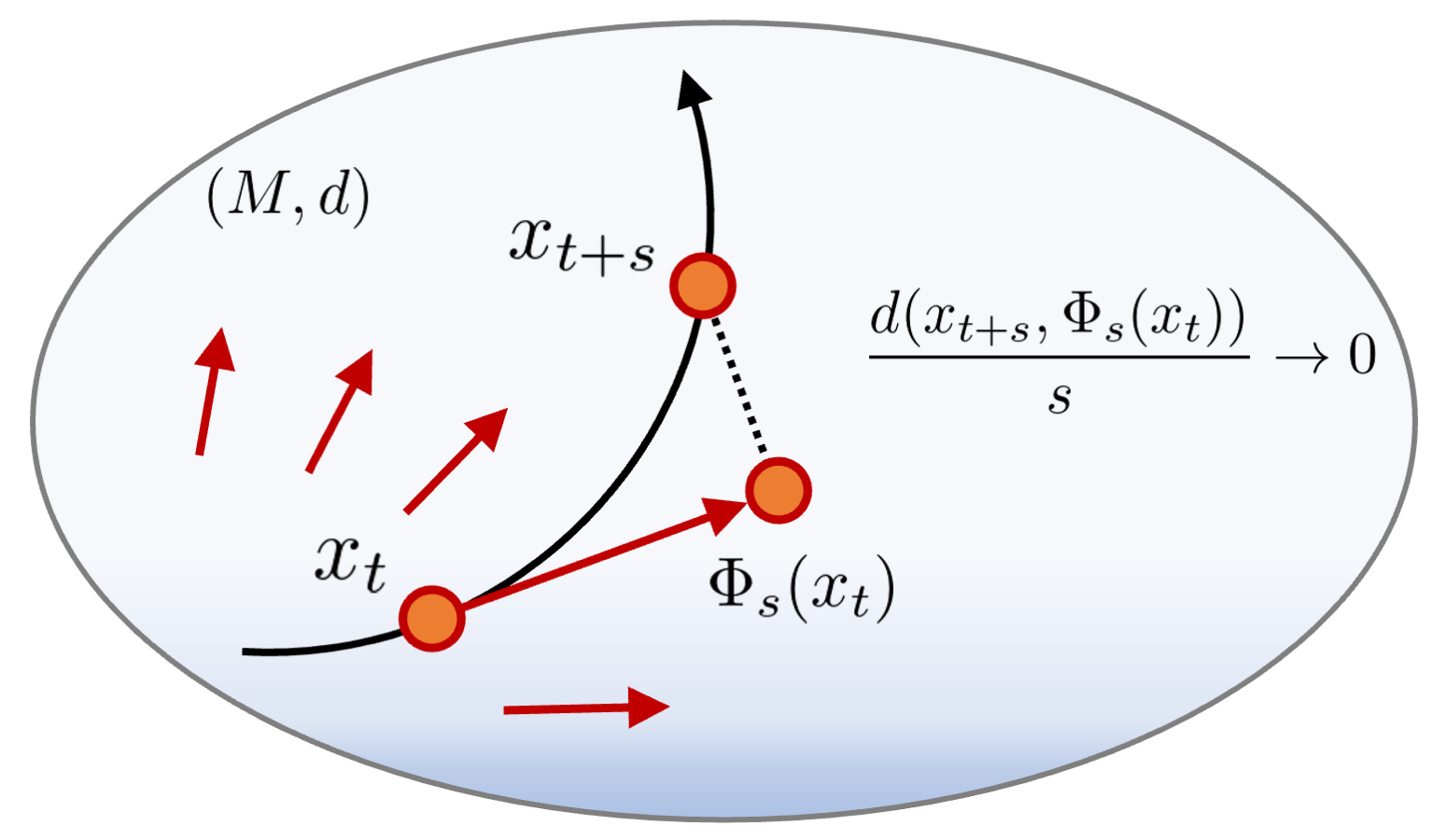}
	\caption{Illustration of the generation of a flow $t \mapsto \mone_t$ in a metric space $(M, d)$ by an arc field $\Phi: M \times [0, 1] \to M$. A local approximation of the flow at time $t$ is specified by the curve $s \mapsto \Phi_s(\mone_t)$ assigned to the point $\mone_t$ by the arc field, under the condition that the approximation error $d(\mone_{t+s}, \Phi_s(\mone_t))$ converges faster than the linear order $s$ at each $\mone_t$.} \label{fig:arcfields}
\end{figure}

Given a metric space $(M, d)$, consider a map $\Phi: M \times [0, 1] \to M$ that satisfies $\Phi_0(x) = x$ for all $x \in M$, where we recall that $\Phi_s(x)$ denotes the value of $\Phi$ at $(x, s) \in M \times [0, 1]$.
To each point $\mone \in M$, this map $\Phi$ is considered to assign a curve $s \mapsto \Phi_s(x)$ over $s \in [0, 1]$ in $(M, d)$ that starts from $\mone$.
\cite{Calcaterra2000} was concerned with the existence of a curve $t \mapsto x_t$ in $(M, d)$ that satisfies the following first-order approximation at each $t$:
\begin{align}
	\lim_{s \to 0^+} \frac{ d( x_{t+s} , \Phi_s( x_t ) ) }{ s } = 0 . \label{eq:emutationeq}
\end{align}
To study the existence, \cite{Calcaterra2000} focused on a class of sufficiently regular maps $\Phi: M \times [0, 1] \to M$, called \emph{arc fields}, that satisfies a continuity condition in \eqref{eq:arc_condition} below.
A curve $t \mapsto x_t$ in $(M, d)$ from some initial point $x_0$ that satisfies the condition \eqref{eq:emutationeq} of an arc field $\Phi$ over some interval $[0, T)$ is then called a \emph{solution} of the arc field $\Phi$.

\begin{definition}[Arc Field] \label{def:cb_definition_21}
	An arc field on a metric space $(M, d)$ is a map $\Phi: M \times [0, 1] \to M$ that satisfies that $\Phi_0(x) = x$ for all $x \in M$ and that
	\begin{align}
		\rho( \mone ) := \sup_{ t, s \in [0, 1] \text{~s.t.~} t \ne s } \frac{ d( \Phi_t(\mone), \Phi_s(\mone) ) }{ | t - s | } < \infty \quad \text{and} \quad \bar{\rho}(\mone, r) := \sup_{ \mthree \in B_r(\mone) } \rho( \mthree ) < \infty \label{eq:arc_condition}
	\end{align}
	at each $\mone \in M$ for $r > 0$ sufficiently small.
\end{definition}

\noindent
See \Cref{fig:arcfields} for illustration of a solution of an arc field $\Phi$.
See \Cref{sec:appendix_0} for a summary of the theoretical results of \cite{Calcaterra2000} relevant to the development of arc Hamiltonian systems.
A flow of an arc Hamiltonian system is formulated as a unique continuous solution of the Hamiltonian arc field $\Phi$ associated with the system.

%% file: section_03.tex

\section{Arc Hamiltonian Systems} \label{sec:methodology}

Extending Hamiltonian systems to infinite-dimensional metric spaces may open up a new horizon for formulating the dynamics of intriguing phenomena.
A focus in this work is on saddle Hamiltonian functions whose domains are equipped with metrics.
This section will establish the general framework of arc Hamiltonian systems.
The first variation (c.f.~\Cref{sec:preliminary}) of a Hamiltonian function determines an arc field, called \emph{Hamiltonian arc field} (c.f.~\Cref{sec:metric_hamiltonian_system}), that generates a flow of the system in the metric space.
It will be shown that the generated flow obeys the conservation of energy.

In \Cref{sec:preliminary}, we present the definition of the first variation based on a convex-analytic formulation.
In \Cref{sec:metric_hamiltonian_system}, we provide the rigorous construction of arc Hamiltonian systems. 
\Cref{sec:existence_flow} shows the existence of a flow generated by a Hamiltonian arc field, followed by \Cref{sec:energy_conservation} establishing the conservation of energy by the generated flow.
Under sufficient regularity, a condition of the existence of the energy-conserving flow reduces to local Lipschitz continuity of the first variation of a Hamiltonian function.
Before moving on to the first subsection, we introduce convex-analytic preliminaries below.

\paragraph*{Convex-Analytic Preliminaries}
Define the extended real number $\Re := \R \cup \{ -\infty, \infty \}$.
It is standard convention in convex analysis to consider extended-valued functions from a vector space $E$ to $\Re$.
We denote by $\Dom(f)$ a subset of $E$ where an extended-valued function $f: E \to \Re$ takes a finite value in $\R$.
An extended-valued function $f: E \to \Re$ is said \emph{convex} if (i) $\Dom(f)$ is a non-empty convex set and (ii) $f( (1 - \lambda) \mone + \lambda \mthree ) \le (1 - \lambda) f( \mone ) + \lambda f(\mthree)$ holds for any $\lambda \in [0,1]$ and $\mone, \mthree \in \Dom(f)$. 
An extended-valued function $f: E \to \Re$ is said \emph{concave} if (i) $-f$ is convex.
Any function $f: D \to \R$ defined on a subset $D$ of $E$ can be identified as an extended-valued function $f: E \to \Re$ by setting $f(\mone) := \infty$ for all $\mone \in E \setminus D$.

\subsection{First Variations} \label{sec:preliminary}

The gradient of functions in a general topological space is characterised as an element of the topological dual.
For example, the topological dual of $\R^d$ is identical to $\R^d$ itself and hence the gradient of functions in $\R^d$ belongs to $\R^d$.
First, we recap definition of \emph{dual pair} \citep[Section~5.14]{Aliprantis2006} that plays a fundamental role in defining a gradient in infinite-dimensional convex analysis.

\begin{definition}[Dual Pair] \label{def:dual_pair}
	A \emph{dual pair}, denoted $\langle E_1, E_2 \rangle$, is a pair of vector spaces $E_1$ and $E_2$ together with a bilinear functional $\langle \cdot, \cdot \rangle: E_1 \times E_2 \to \R$ that satisfies
	\begin{itemize}
		\item if $\langle \mone, \mtwo \rangle = 0$ for each $\mtwo \in E_2$, then $\mone = 0$;
		\item if $\langle \mone, \mtwo \rangle = 0$ for each $\mone \in E_1$, then $\mtwo = 0$.
	\end{itemize}	
	The bilinear functional $\langle \cdot, \cdot \rangle$ is called the duality of the dual pair $\langle E_1, E_2 \rangle$.
\end{definition}

\noindent
A dual pair $\langle E_1, E_2 \rangle$ equips each $E_1$ and $E_2$ with topologies---called the weak and weak$^*$ topologies, respectively---that turn $E_1$ and $E_2$ into the topological dual of each other. 
A large class of vector spaces can be coupled as the mutual topological dual through an appropriate dual pair; see \citep[Section~5.14]{Aliprantis2006} for examples of common dual pairs.
A dual pair we consider in \Cref{sec:application} is one for functions $f$ and measures $P$ in $\Theta$ whose duality is given by $\langle f, P \rangle = \int_\T f(\t) d P(\t)$ \citep[e.g.][]{Muller1997}.

In the rest of \Cref{sec:methodology}, let $\langle E_1, E_2 \rangle$ be an arbitrary dual pair.
It suffices to consider convex or concave functions $f: E_1 \to \Re$ to define a gradient we use for saddle Hamiltonian function.
The directional derivative $D f_\mone(\mthree)$ of $f$ at $\mone$ in direction $\mthree \in E_1$ is defined by
\begin{align}
	D f_\mone(\mthree) := \lim_{\lambda \to 0^+} \frac{ f(\mone + \lambda \mthree) - f(\mone) }{ \lambda } . \label{eq:direc_deriv}
\end{align}
One of the most common notions of gradients in infinite-dimensional convex analysis is the G\^{a}teaux gradient\footnote{The G\^{a}teaux gradient may be called by a different terminology in other literatures. For example, \cite{Rockafellar1974} simply called it the \emph{gradient} and \cite{Ekeland1999} called it the \emph{G\^{a}teaux-differential}.}.
If there exists a unique element $\partial f(x) \in E_2$ s.t.~the limit of \eqref{eq:direc_deriv} satisfies
\begin{align}
	D f_\mone(\mthree) = \langle \mthree, \partial f(x) \rangle \quad \text{for all} \quad \mthree \in E_1 , \label{eq:gateaux_deriv}
\end{align}
the unique element $\partial f(x) \in E_2$ is called the G\^{a}teaux gradient of $f$ at $\mone$.
The linear functional $D f_\mone: E_1 \to \R$ is called the G\^{a}teaux derivative of $f$ at $\mone$. 
A gradient we use is a weaker version of the G\^{a}teaux gradient that requires the equality \eqref{eq:gateaux_deriv} to hold only for directions $\mthree$ s.t.~$\mone + \lambda \mthree$ stays within $\Dom(f)$ over all $\lambda \in [0, 1]$.
This ensures that the value of $f(\mone + \lambda \mthree)$ in the directional derivative \eqref{eq:direc_deriv} is finite for all small $\lambda \ge 0$ to take the limit in $\lambda \to 0^+$.
Such a version is often used and called \emph{first variation} in the context of Wasserstein gradient flows \citep[e.g.][Chapter~7.2]{Santambrogio2015} for functions whose domain is a set of probability measures $\Mp$ in the vector space of the signed measures $\Ms$.
What first variation refers to may differ in each field.
Now, we introduce our definition of first variation.

\begin{definition}[First Variation] \label{def:gateux_derivative}
Given a convex or concave function $f: E_1 \to \Re$, let $C(\mone)$ denote a set $\{ \mthree \in E_1 \mid \mone + \mthree \in \Dom(f) \}$ at each $\mone \in \Dom(f)$.
If there exists a unique element $\partial f(x) \in E_2$ s.t.~the directional derivative of $f$ at $\mone \in \Dom(f)$ in \eqref{eq:direc_deriv} satisfies that
\begin{align}
	D f_\mone(\mthree) = \langle \mthree, \partial f(x) \rangle \quad \text{for all} \quad \mthree \in C(\mone) , \label{eq:c_weak_deriv}
\end{align}
we call the unique element $\partial f(x) \in E_2$ the \emph{first variation} of $f$ at $\mone$.
We refer the function $D f_\mone: C(\mone) \to \R$ to as the \emph{variational derivative} of $f$ at $\mone$.
\end{definition}

Note that, if $\mthree \in C(\mone)$, we have $\mone + \lambda \mthree \in \Dom(f)$ for all $\lambda \in [0, 1]$ because $\mone + \lambda \mthree$ is expressed as a convex combination $(1 - \lambda) \mone + \lambda ( \mone + \mthree )$ in the convex set $\Dom(f)$.
Note also that the set $C(\mone)$ is never empty; given any $u \in \Dom(f)$, a point $u - \mone$ is a trivial element of $C(\mone)$.
If $C(\mone) = E_1$, the first variation coincides with the G\^{a}teaux gradient by definition.

\begin{remark}
	The first variation are immediately defined for convex or concave functions in the other space $E_2$ by interchanging the symbols $E_1$ and $E_2$ throughout \Cref{def:gateux_derivative} and replacing $D f_\mone(\mthree) = \langle \mthree, \partial f(x) \rangle$ with $D f_\mone(\mthree) = \langle \partial f(x), \mthree \rangle$ in \eqref{eq:c_weak_deriv}.
\end{remark}

\subsection{Definition of Arc Hamiltonian Systems} \label{sec:metric_hamiltonian_system}

Hamiltonian systems describe the dynamics of a paired variable $(\monet, \mtwot)$ in some space.
We consider a setting where the variable $(\monet, \mtwot)$ belongs to a subset $M$ in the product space $E_1 \times E_2$.
A Hamiltonian function $H: M \to \R$ on $M$ is regarded as an extended-valued function $H: E_1 \times E_2 \to \Re$ in convex-analytic manner.

The framework of arc Hamiltonian systems is developed for saddle Hamiltonian functions.
An extended-valued function $H: E_1 \times E_2 \to \Re$ is said \emph{saddle} if, at each $(\mone, \mtwo) \in \Dom(H)$,
\begin{itemize}
	\item for fixed $\mtwo$, the function $H(\cdot, \mtwo): E_1 \to \Re$ is concave;
	\item for fixed $\mone$, the function $H(\mone, \cdot): E_2 \to \Re$ is convex.
\end{itemize}
Denote by $\partial_1 H(\mone, \mtwo) \in E_2$ the first variation of the concave function $H(\cdot, \mtwo): E_1 \to \Re$ at $\mone$.
Denote by $\partial_2 H(\mone, \mtwo) \in E_1$ the first variation of the convex function $H(\mone, \cdot): E_2 \to \Re$ at $\mtwo$.
It is convenient to introduce a condition, \emph{compatibility}, of a saddle function $H$ to be used in arc Hamiltonian systems.
In what follows, the duplet $\partial H(\mone, \mtwo) := (- \partial_2 H(\mone, \mtwo), \partial_1 H(\mone, \mtwo))$ is referred to as the \emph{symplectic variation}\footnote{It is possible to define the symplectic variation as $\partial H(\mone, \mtwo) = (\partial_2 H(\mone, \mtwo), - \partial_1 H(\mone, \mtwo))$ alternatively.} of a saddle function $H$.

\begin{definition}[Compatibility] \label{def:compatible_H}
	A saddle function $H: E_1 \times E_2 \to \Re$ is said \emph{compatible} if, at each $(\mone, \mtwo) \in \Dom(H)$, the symplectic variation $\partial H(\mone, \mtwo) = (- \partial_2 H(\mone, \mtwo), \partial_1 H(\mone, \mtwo))$ exists and satisfies $(\mone, \mtwo) + s \, ( - \partial_2 H(\mone, \mtwo), \partial_1 H(\mone, \mtwo)) \in \Dom(H)$ for all $s \in [0, 1]$.
\end{definition}

We then equip the domain $\Dom(H)$ with a metric $d$ to define the Hamiltonian arc field associated with the compatible Hamiltonian function $H$.

\begin{definition}[Hamiltonian Arc Field]
	Given a compatible function $H: E_1 \times E_2 \to \Re$ and a metric $d$ on the domain $M := \Dom(H)$, a map $\Phi: M \times [0, 1] \to M$ defined by
	\begin{align*}
		\Phi_s(\mone, \mtwo) := (\mone, \mtwo) + s ( - \partial_2 H(\mone, \mtwo), \partial_1 H(\mone, \mtwo))
	\end{align*}
	is called the \emph{Hamiltonian arc field} on $(M, d)$ if it is an arc field on $(M, d)$.
\end{definition}

The Hamiltonian arc field $\Phi$ acts as an analogue of a Hamiltonian vector field in the metric space $(M, d)$.
We will see in \Cref{sec:energy_conservation} that a solution of the Hamiltonian arc field $\Phi$---i.e., a curve that satisfies the condition \eqref{eq:emutationeq} of $\Phi$---is an energy-conserving flow in $(M, d)$.
We refer the condition \eqref{eq:emutationeq} of the Hamiltonian arc field $\Phi$ to as the \emph{arc Hamilton equation}.

\begin{definition}[Arc Hamiltonian System] \label{def:metric_hamiltonian_system}
Let $H: E_1 \times E_2 \to \Re$ be compatible and let $d$ be a metric on the domain $M := \Dom(H)$.
We call the tuple $(H, M, d)$ an \emph{arc Hamiltonian system} if the Hamiltonian arc field $\Phi$ on $(M, d)$ exists.
Given an initial state $(\mone, \mtwo) \in M$, if there exists a unique continuous solution $t \mapsto (\monet, \mtwot)$ that satisfies the arc Hamilton equation
\begin{align}
	\lim_{s \to 0^+} \frac{ d( (\mone_{t+s}, \mtwo_{t+s}), \Phi_s( \mone_t, \mtwo_t ) ) }{ s } = 0 \label{eq:H_mutation_equation}
\end{align}
from $(\mone_0, \mtwo_0) = (\mone, \mtwo)$ in some interval $[0, T)$, the solution is called a \emph{flow} of the system from $(\mone, \mtwo)$.
The flow is said \emph{global} if $T = \infty$.
\end{definition}

One of our main contributions is this rigorous formulation of the novel class of Hamiltonian systems in infinite-dimensional metric spaces.
We briefly discuss a numerical discretisation scheme of the arc Hamilton equation \eqref{eq:H_mutation_equation}.
One of the most straightforward approaches is the first-order discretisation scheme.
Given a small $\delta > 0$, a flow generated by the system is approximated by the following recursive formula starting from $t = 0$:
\begin{align}
	(\mone_{t+\delta}, \mtwo_{t+\delta}) \gets (\mone_t, \mtwo_t) + \delta \, ( - \partial_2 H(\mone_t, \mtwo_t), \partial_1 H(\mone_t, \mtwo_t) ) . \label{eq:discretisation}
\end{align}
Next, we turn our attention to the condition under which there exists a global flow generated by a Hamiltonian arc field from any point in the domain $(M, d)$.

\subsection{Existence of Flows} \label{sec:existence_flow}

By definition, a flow of an arc Hamiltonian system is a unique continuous solution of the Hamilton arc field $\Phi$ from a given initial state.
\cite{Calcaterra2000} established a condition of a general arc field $\Phi$, under which a unique solution of the arc field $\Phi$ exists from any initial state.
First, we adapt their result to arc Hamilton fields to provide general conditions of the existence of global flows.
Second, we consider a convenient setting where a metric $d$ on $M$ has an extension to a norm-induced metric $d_*$ on $E_1 \times E_2$. 
In this setting, we can reduce the general condition provided first to local Lipschitz continuity of the map $(\mone, \mtwo) \mapsto \partial H(\mone, \mtwo)$ from $(M, d)$ to $(E_1 \times E_2, d_*)$.

In the rest of \Cref{sec:methodology}, let $(M, d)$ be a complete metric space s.t.~$M \subseteq E_1 \times E_2$.
Let $H: E_1 \times E_2 \to \Re$ be a compatible function s.t.~$\Dom(H) = M$.
Let $\Phi: M \times [0, 1] \to M$ be a map defined by $\Phi_s(\mone, \mtwo) := (\mone, \mtwo) + s \partial H(\mone, \mtwo)$ at each $(\mone, \mtwo) \in M$ and $s \in [0, 1]$. 

Our first theorem shows the existence of global flows of the system.
We adapt the result in \cite{Calcaterra2000} by relaxing their original condition; see \Cref{sec:appendix_0} for a summary of their result and the relaxation of the original condition.
Recall from \Cref{sec:notation} that $B_r(\mone, \mtwo)$ denotes an open ball in $(M, d)$.

\begin{assumption} \label{asmp:flow_existence_1}
	For some $\Lambda: M \times M \times [0, 1] \to \R$ and $\Omega: M \times [0, 1] \times [0, 1] \to [0, \infty)$, at each $(\mone_0, \mtwo_0) \in M$ there exist a radius $r_0 > 0$ and a length $\epsilon_0 > 0$ s.t.
	\begin{enumerate}
		\item $d(\Phi_t(\mone, \mtwo), \Phi_t(\mthree, \mfour) ) \le d((\mone, \mtwo), (\mthree, \mfour)) (1 + t \, \Lambda((\mone, \mtwo), (\mthree, \mfour), t))$, and $\Lambda((\mone, \mtwo), (\mthree, \mfour), t)$ is upper bounded, for all $(\mone, \mtwo), (\mthree, \mfour) \in B_{r_0}(\mone_0, \mtwo_0)$ and all $t \in [0, \epsilon_0)$;
		\item $d(\Phi_{t+s}(\mone, \mtwo), \Phi_s(\Phi_t(\mone, \mtwo))) \le t \, s \, \Omega((\mone, \mtwo), t, s))$, and $\Omega((\mone, \mtwo), t, s))$ is upper bounded, for all $(\mone, \mtwo) \in B_{r_0}(\mone_0, \mtwo_0)$ and all $t,s \in [0, \epsilon_0)$.
	\end{enumerate}
\end{assumption}

\begin{assumption} \label{asmp:flow_existence_2}
	For the following function
	\begin{align*}
		\rho( \mone, \mtwo ) := \sup_{ t, s \in [0, 1] \text{~s.t.~} t \ne s } \frac{ d( \Phi_t(\mone, \mtwo), \Phi_s(\mone, \mtwo) ) }{ | t - s | } ,
	\end{align*}
	at each $(\mone_0, \mtwo_0) \in M$ there exist a radius $r_0 > 0$ and two constants $C_1, C_2$ s.t.~we have $\rho(\mone, \mtwo) \le C_1 d((\mone, \mtwo), (\mone_0, \mtwo_0)) + C_2$ for all $(\mone, \mtwo) \in B_{r_0}(\mone_0, \mtwo_0)$.
\end{assumption}

\begin{theorem}[Existence of Global Flow] \label{thm:existence_of_global_flow}
	Suppose Assumptions \ref{asmp:flow_existence_1}--\ref{asmp:flow_existence_2}. 
	The tuple $(H, M, d)$ is an arc Hamiltonian system.
	Given any initial state $(\mone, \mtwo)$ in $M$, there exists a global flow $t \mapsto (\mone_t, \mtwo_t)$ of the system from $(\mone_0, \mtwo_0) = (\mone, \mtwo)$ in the infinite interval $[0, \infty)$.
\end{theorem}

\begin{proof}[Proof of \Cref{thm:existence_of_global_flow}]
	See \Cref{proof:existence_of_global_flow}.
\end{proof}

We define $E := E_1 \times E_2$ to simplify the notation hereafter in \Cref{sec:methodology}.
The condition in Assumptions \ref{asmp:flow_existence_1}--\ref{asmp:flow_existence_2} is stated in a highly general form.
Nonetheless, our second theorem shows that the condition reduces to local Lipschitz continuity of the map $(\mone, \mtwo) \mapsto \partial H(\mone, \mtwo)$ when the metric $d$ on $M$ admits an extension to a norm-induced metric $d_*$ on $E$.
Clearly, a metric $d$, that is defined by restricting some norm $\| \cdot \|_E$ on $E$ to $M$, immediately admits an extension to a norm-induced metric of $\| \cdot \|_E$ on $E$.
Such metrics are common in spaces of measures; for example see \citep{Muller1997} for integral probability metrics.

\begin{assumption} \label{asmp:flow_existence_3}
	The following conditions hold:
	\begin{enumerate}
		\item there exists a norm-induced metric $d_*$ on $E$, to which the metric $d$ on $M$ can be extended;
		\item the map $(\mone, \mtwo) \mapsto \partial H(\mone, \mtwo)$ is locally Lipschitz continuous from $(M, d)$ to $(E, d_*)$, that is, at each $(\mone_0, \mtwo_0) \in M$ there exist a radius $\rll > 0$ and a constant $K_0 > 0$ s.t.
		\begin{align}
			d_*( \partial H(\mone, \mtwo), \partial H(\mthree, \mfour) ) \le K_0 \, d( (\mone, \mtwo), (\mthree, \mfour) ) \label{eq:local_lipschitz}
		\end{align}
		holds for all $(\mone, \mtwo), (\mthree, \mfour) \in B_{\rll}(\mone_0, \mtwo_0)$.
	\end{enumerate}
\end{assumption}

\begin{theorem} \label{lem:existence_of_flow_simplified}
	\Cref{asmp:flow_existence_3} implies Assumptions \ref{asmp:flow_existence_1}--\ref{asmp:flow_existence_2}.
\end{theorem}

\begin{proof}[Proof of \Cref{lem:existence_of_flow_simplified}]
	See \Cref{proof:existence_of_flow_simplified}.
\end{proof}

\begin{remark}
	While the metric space $(M, d)$ is supposed to be complete, the metric space $(E, d_*)$ does not need to be complete.
	Consider the space of measures $\Mf$ for example. 
	There are many common metrics on $\Mf$ that can be extended to norm-induced metrics on the vector space of signed measures $\Ms$ \citep[e.g.][]{Muller1997}.
	However, the resulting norm space may no longer be complete; see \citep[p.192]{Bogachev2006} for the case of the Kantorovich–Rubinstein metric.
	\Cref{asmp:flow_existence_3} is applicable without completeness of $(E, d_*)$.
\end{remark}

By Theorems \ref{thm:existence_of_global_flow}--\ref{lem:existence_of_flow_simplified}, local Lipschitz continuity of the map $(\mone, \mtwo) \mapsto \partial H(\mone, \mtwo)$ implies the existence of a global flow of the system from any initial state.
It gives a metric-space analogue of the Cauchy-Lipschitz theorem \citep{Brezis2011}.
The analogy was discussed in \cite{Calcaterra2000} for an evolution equation in a Banach space defined by a Lipschitz continuous vector field.
Our result demonstrates the analogy for a metric space $(M, d)$ using only local Lipschitz continuity of the map $(\mone, \mtwo) \mapsto \partial H(\mone, \mtwo)$.

\subsection{Conservation of Energy} \label{sec:energy_conservation}

An important property of regular Hamiltonian systems is that they obey the law of conservation of energy, meaning that the value of a Hamiltonian function remains constant along any flow of the system.
Now we shall establish that arc Hamiltonian systems are endowed with this invariance to obey the law of conservation of energy.
This can be shown under the following continuity conditions of $H$.

\begin{assumption} \label{asmp:conservation_of_energy_1}
	The following conditions hold:
	\begin{enumerate}
		\item the function $H$ is locally Lipschitz continuous in $(M, d)$, that is, at each $(\mone_0, \mtwo_0) \in M$ there exists a radius $\eta_0 > 0$ and a constant $J_0 > 0$ s.t.
		\begin{align}
			| H(\mone, \mtwo) - H(\mthree, \mfour) | \le J_0 \, d( (\mone, \mtwo), (\mthree, \mfour) ) \label{eq:H_local_lipschitz}
		\end{align}
		holds for all $(\mone, \mtwo), (\mthree, \mfour) \in B_{\eta_0}(\mone_0, \mtwo_0)$.
		\item convergence $(\mone_n, \mtwo_n) \to (\mone, \mtwo)$ in $(M, d)$ implies that (i) $\langle \partial_1 H(\mone_n, \mtwo_n), \mfour \rangle \to \langle \partial_1 H(\mone, \mtwo), \mfour \rangle$ for each $\mfour \in E_2$ and (ii) $\langle \mthree, \partial_2 H(\mone_n, \mtwo_n) \rangle \to \langle \mthree, \partial_2 H(\mone, \mtwo) \rangle$ for each $\mthree \in E_1$.
	\end{enumerate}
\end{assumption}

\begin{theorem}[Conservation of Energy] \label{prop:conservation_of_energy}
	Suppose the tuple $(H, M, d)$ is an arc Hamiltonian system.
	Under \Cref{asmp:conservation_of_energy_1}, if a flow $t \mapsto ( \monet, \mtwot )$ of the system exists in some interval $[0, T)$, the value of the Hamiltonian function $H( \monet, \mtwot )$ along the flow is constant in $[0, T)$.
\end{theorem}

\begin{proof}[Proof of \Cref{prop:conservation_of_energy}]
	See \Cref{proof:conservation_of_energy}.
\end{proof}

\noindent
Note that \Cref{prop:conservation_of_energy} holds for any length $T > 0$ including $T = \infty$.
Combining \Cref{thm:existence_of_global_flow} with \Cref{prop:conservation_of_energy} establishes the fact that the arc Hamiltonian system $(H, M, d)$ generates energy-conserving flows in $[0, \infty)$ everywhere in the metric space $(M, d)$.

Finally, we revisit the convenient setting of \Cref{asmp:flow_existence_3}.
In this setting, \Cref{asmp:conservation_of_energy_1} reduces to local Lipschitz continuity of the map $(\mone, \mtwo) \to \partial H(\mone, \mtwo)$ if we impose an additional regularity on the metric $d_*$ on $E$ as in \Cref{asmp:conservation_of_energy_3}.
The following theorem combines Theorems \ref{thm:existence_of_global_flow}--\ref{prop:conservation_of_energy} under this setting.
It provides a sufficient condition for the existence of energy-conserving flows, which is easy to verify for our application presented in the next section.

\begin{assumption} \label{asmp:conservation_of_energy_3}
	There exists a norm $\| \cdot \|_{E_1}$ on $E_1$ and a norm $\| \cdot \|_{E_2}$ on $E_2$ that satisfy $\langle \mone, \mtwo \rangle \le \| \mone \|_{E_1} \| \mtwo \|_{E_2}$ for any $(\mone, \mtwo) \in E_1 \times E_2$.
	The norm-induced metric $d_*$ on $E_1 \times E_2$ in \Cref{asmp:flow_existence_3} is defined as $d_*((\mone, \mtwo), (\mthree, \mfour)) = \| \mone - \mtwo \|_{E_1} + \| \mthree - \mfour \|_{E_2}$.
\end{assumption}

\begin{theorem} \label{thm:existence_conservation}
	Suppose Assumptions \ref{asmp:flow_existence_3} and \ref{asmp:conservation_of_energy_3}.
	The tuple $(H, M, d)$ is an arc Hamiltonian system.
	Given any initial state $(\mone, \mtwo) \in M$, there exists a global flow $t \mapsto (\mone_t, \mtwo_t)$ of the system from $(\mone_0, \mtwo_0) = (\mone, \mtwo)$ in the indefinite interval $[0, \infty)$, along which the value of the Hamiltonian function $H(\mone_t, \mtwo_t)$ is constant everywhere in $[0, \infty)$.
\end{theorem}

\begin{proof}[Proof of \Cref{thm:existence_conservation}]
	See \Cref{proof:existence_conservation}.
\end{proof}

\noindent
This completes the construction and theoretical development of arc Hamiltonian systems.

%% file: section_04.tex

\section{Hamiltonian Dynamics of Bayesian Inference} \label{sec:application}

We shall present an arc Hamiltonian system of the minimum free energy $H$ defined in a metric space $(M, d)$ of negative log-likelihoods $f$ and probability measures $P$ on $\Theta$.
We derive the symplectic variation $\partial H(f, P) = ( - \partial_2 H(f, P) , \partial_1 H(f, P) )$ of the minimum free energy $H$.
It determines the Hamiltonian arc field on $(M, d)$ that generates flows of negative log-likelihoods $f_t$ and probability measures $P_t$.
We reveal the underlying invariance to conserve the free energy.

Recall that the main components to construct an arc Hamiltonian system $(H, M, d)$ are:
\begin{enumerate}
	\item a dual pair $\langle E_1, E_2 \rangle$ that contains a domain $M$ of interest;
	\item a domain $M$ and metric $d$ to constitute a complete metric space $(M, d)$;
	\item a compatible Hamiltonian function $H$ over the domain $M$.
\end{enumerate}
In \Cref{sec:domain}, we define a dual pair $\langle E_1, E_2 \rangle$ used throughout this section.
In \Cref{sec:metric_spcae_Bayes}, we detail a metric space $(M, d)$ of negative log-likelihoods $f$ and probability measures $P$, in which the minimum free energy $H$ is defined.
In \Cref{sec:cgf}, we show the compatibility of the minimum free energy $H$ and the generation of the energy-conserving flows in the metric space $(M, d)$.
Finally, \Cref{sec:illustration} provides a visual illustration of the generated flow.

\begin{remark}
Recall from \Cref{sec:notation} that the only assumption of the space $\Theta$, on which functions $f$ and measures $P$ are defined, is to be separable and complete metrizable.
Hence, all the results in this section hold regardless of whether $\Theta$ is finite or infinite dimensional.
Note that we assume no specific form of probability models of negative log-likelihoods $f$, where they are simply viewed as functions $f: \Theta \to \R$ on $\Theta$.
Functions $f$ to be contained in the domain $M$ will be specified by their tail-growth condition.
Hence, the `posterior' $P_*$ of the form \eqref{eq:posterior} in our result includes general Gibbs measures or pseudo-posteriors.
\end{remark}

\subsection{Dual Pair} \label{sec:domain}

First, we define a dual pair $\langle E_1, E_2 \rangle$ consisting of a vector space $E_1$ of functions $f: \T \to \R$ and a vector space $E_2$ of signed measures $P$ on $\T$.
A common condition to construct such a dual pair is the \emph{growth-rate condition} of functions $f$ and the \emph{moment condition} of signed measures $\P$.
The following spaces $E_1$ and $E_2$ are often referred to as a weighted function space and a weighted measure space \citep{Kolokoltsov2019}, where a given weight function $w: \Theta \to \R$ specifies the maximum growth-rate of functions.

\begin{definition}[Weighted Spaces] \label{def:weighted_space}
	Given a continuous function $w: \T \to \R$ that is lower bounded as $w(\t) \ge 1$ over $\T$,
	\begin{itemize}
		\item let $E_1$ be a vector space of all measurable functions $f: \Theta \to \R$ s.t.~$\sup_{\t \in \Theta} \frac{ | f(\t) | }{ w(\t) } < \infty$;
		\item let $E_2$ be a vector space of all signed measures $P \in \mathcal{M}$ s.t.~$\int_{\t \in \T} w(\t) \d | P |(\dt) < \infty$.
	\end{itemize}
	Define a bilinear functional $\langle \cdot, \cdot \rangle: E_1 \times E_2 \to \R$ by the integral
	\begin{align}
		\langle f, P \rangle := \int_\Theta f(\t) \d P(\t) , \label{label:duality_fp}
	\end{align}
	under which the pair of the vector spaces $E_1$ and $E_2$ constitutes a dual pair $\langle E_1, E_2 \rangle$.
\end{definition}

\noindent
The proof that the vector spaces $E_1$ and $E_2$ form a dual pair under the duality \eqref{label:duality_fp} can be found in \cite[Lemma~2.1]{Muller1997}.
For example, if setting $w(\theta) = 1 + \| \theta \|^p$ for some $p \ge 0$ when $\Theta = \R^d$, the space $E_1$ represents all measurable functions up to the $p$-th order polynomial growth; the space $E_2$ represents all signed measures endowed with the $p$-th moment condition.
A choice of the weight function $w$ is arbitrary for all the results in this work.

\subsection{Domain of Minimum Free Energy} \label{sec:metric_spcae_Bayes}

In the rest of \Cref{sec:application}, let $\langle E_1, E_2 \rangle$ be the dual pair in \Cref{def:weighted_space}.
The minimum free energy is defined on a metric space $(M, d)$ of functions $f$ in $E_1$ and measures in $E_2$.
We specify the condition of the domain $M$ as follows.

\begin{assumption} \label{asmp:domain_Bayes}
	Let $M$ be a set of all $(f, P) \in E_1 \times E_2$ s.t.
	\begin{itemize}
		\item $f$ is a continuous function that satisfies $f(\t) \ge \log w(\t)$;
		\item $P$ is a probability measure.
	\end{itemize}
\end{assumption}

\noindent
Since $\log w(\t) \ge 0$ due to $w(\t) \ge 1$ c.f.~\Cref{def:weighted_space}, the above condition on $f$ can be understood as (i) $f(\t)$ grows faster than $\log w(\t)$ and (ii) $f(\t)$ is lower bounded by $0$.
These two conditions are not restrictive for most negative log-likelihoods:
\begin{itemize}
	\item The growth of $\log w(\t)$ in the condition (i) is slower than \emph{any small} $q$-th order polynomial ($0 < q \ll 1$) if the weight function $w(\t)$ is of any $p$-th polynomial order ($0 < p < \infty$)\footnote{For illustration, consider only $\theta > 0$ in $\Theta = \R$. If $w(\t) = (1 + \t)^p$, we have $\log w(\t) = p \log(1 + \t)$. It is easy to verify---e.g., by L'H\^{o}pital's rule---that $\log(1 + \t)$ grows slower than any $q$-th order polynomial.}.
	\item It it natural to assume that a negative log-likelihood has some lower bound $U \in \R$ because it outght to be minimisable if the maximum likelihood estimator exists.
	We suppose $U = 0$ in the condition (ii) without loss of generality. For a negative log-likelihood $l(\t)$ whose value may fall below $U = 0$, we can consider a version $f(\t) = l(\t) + C$ shifted by a large constant $C$.
	The minimum free energy is invariant to any constant term of $f(\t)$ because it is defined with a `mean-zero' transform $f(\t) - \int_{\T} f(\t) \d \P(\t)$ that cancels the constant $C$:
	\begin{align*}
		f(\t) - \int_{\T} f(\t) \d P(\t) = ( l(\t) + C ) - \int_{\T} ( l(\t) + C ) \d P(\t) = l(\t) - \int_{\T} l(\t) \d P(\t) .
	\end{align*} 
\end{itemize}
For example, consider a normal location model $\mathcal{N}(\theta, 1)$ in $\Theta = \R$, whose negative log-likelihood is given by $f(\theta) = (x - \theta)^2 / 2 + \log( \sqrt{ 2 \pi } )$ for any datum $x \in \R$.
This negative log-likelihood $f$ is contained in $M$ if we set $w(\t) = 1 + \| \t \|^p$ for any $p \ge 2$.

We specify a metric $d$ on the domain $M$ that determines convergence of a flow evolving in the domain $M$.
We use the product metric of the weighted uniform metric $\varrho$ for the functions $f$ and the weight total variation metric $\gamma$ for the measures $P$ as below.

\begin{assumption} \label{asmp:metric_Bayes}
	For any $f, g \in E_1$ and any $P, Q \in E_2$, define
	\begin{align}
		\varrho(f, g) := \sup_{ \t \in \Theta } \frac{ | f(\t) - g(\t) | }{ w(\t) } \quad \text{and} \quad \gamma(P, Q) := \int_{ \Theta} w(\t) \d | \P - \Q |(\t) . \label{eq:wu_tv_metrics}
	\end{align}
	Let $d$ be a metric on the set $M$ in \Cref{asmp:domain_Bayes} s.t.~$d((f, P), (g, Q)) := \varrho(f, g) + \gamma(P, Q)$.
\end{assumption}

\noindent
The weighted uniform metric $\varrho$ is a natural choice that allows functions $f$ to grow up to the rate of the weight function $w(\t)$.
The weighted total variation metric $\gamma$ is topologically stronger than most of the other common metrics for measures $P$.
It controls the standard total variation metric and the Kantorovich–Rubinstein metric.
It also controls the Wasserstein metric given an appropriate weight function $w$ \citep{Villani2009}.
If a flow of an arc Hamilton system exists under some metric $d$, the flow exists under any other metric weaker than $d$.
This is because the convergence in \eqref{eq:H_mutation_equation} under the metric $d$ immediately implies the convergence under any other weaker metric.
Therefore, the existence of a flow under \Cref{asmp:metric_Bayes} is sufficient for the existence in any other weaker metrics.

Finally, we show that the metric space $(M, d)$ is complete.

\begin{proposition} \label{prop:M_completeness}
	Suppose Assumptions \ref{asmp:domain_Bayes} and \ref{asmp:metric_Bayes}.
	The metric space $(M, d)$ is complete.
\end{proposition}

\begin{proof}[Proof of \Cref{prop:M_completeness}]
	See \Cref{sec:appendix_b1}.
\end{proof}

\noindent
This completes specification of the domain $(M, d)$ of the minimum free energy.

\subsection{System of Minimum Free Energy} \label{sec:cgf}

We shall derive the first variation of the minimum free energy.
We then show the generation of the energy-conserving flows by the corresponding Hamiltonian arc field.
First, we provide a formal definition of the minimum free energy.

\begin{assumption} \label{asmp:Hamiltonian_Bayes}
	For the set $M$ in \Cref{asmp:domain_Bayes}, define $H: E_1 \times E_2 \to \Re$ by
	\begin{align}
		H(f, \P) := - \log\left( \int_\Theta \exp\left( - \left( f(\t) - \int_\T f(\t) \d \P(\t) \right) \right) \d \P(\t) \right)  \label{eq:Hamiltonian_Bayes}
	\end{align}
	for each $(f, \P) \in M$ and by $H(f, \P) := \infty$ for all $(f, \P) \not \in M$.
\end{assumption}

\noindent
All the values of $H$ outside the domain $M$ are set to $\infty$ in convex-analytic manner.
It takes a finite value in $\R$ inside the domain $M$, where the integral inside the logarithm function is positive and finite because $f$ is non-negative and $P$ is a probability measure.
The next proposition confirms that the minimum free energy $H$ is saddle.

\begin{proposition} \label{thm:Hamiltonian_Bayes_saddle}
	Suppose \Cref{asmp:Hamiltonian_Bayes}.
	The function $H$ is saddle.
\end{proposition}

\begin{proof}[Proof of \Cref{thm:Hamiltonian_Bayes_saddle}]
	See \Cref{proof:Hamiltonian_Bayes_saddle}.
\end{proof}

We derive a form of the symplectic variation $\partial H(f, P) = (- \partial_2 H(f, P), \partial_1 H(f, P))$ of the minimum free energy $H$.
In addition, we confirm that $H$ is compatible (c.f.~\Cref{def:compatible_H}).
Recall from \Cref{sec:notation} that, if a measure $Q$ admits the Radon-Nikodym derivative $h$ with respect to a measure $R$, we use a notation $d Q(\t) = h(\t) d R(\t)$ for better presentation.

\begin{theorem} \label{thm:Hamiltonian_Bayes_variation}
	Suppose \Cref{asmp:Hamiltonian_Bayes}.
	At each $(f, P) \in M$, we have 
	\begin{align*}
		(- \partial_2 H(f, P), \partial_1 H(f, P)) = ( f_* + f , P_* - P )
	\end{align*}
	for a function $f_* \in E_1$ and a measure $P_* \in E_2$ defined by
	\begin{align*}
		f_*(\t) := \frac{ \exp( - f(\t) ) }{ \int_\Theta \exp( - f(\t) ) \d P(\t) } \qquad \text{and} \qquad \d P_*(\t) := \frac{ \exp( - f(\t) ) }{ \int_\Theta \exp( - f(\t) ) \d P(\t) } \d P(\t) .
	\end{align*}
\end{theorem}

\begin{proof}[Proof of \Cref{thm:Hamiltonian_Bayes_variation}]
	See \Cref{sec:appendix_c}.
\end{proof}

\begin{proposition} \label{thm:Hamiltonian_Bayes}
	Suppose \Cref{asmp:Hamiltonian_Bayes}. The function $H$ is compatible.
\end{proposition}

\begin{proof}[Proof of \Cref{thm:Hamiltonian_Bayes}]
	See \Cref{proof:Hamiltonian_Bayes}.
\end{proof}

This result establishes the characterisation $P_* - P$ of the difference of the posterior $P_*$ and the prior $P$ as the first variation $\partial_1 H(f, P)$.
Now, we have the domain $M$, the metric $d$, and the compatible Hamiltonian function $H$ to construct the arc Hamiltonian system $(H, M, d)$ of the minimum free energy.
We arrive at a main result on the generation of the flow everywhere in $(M, d)$ that obeys the invariance to conserve the minimum free energy.

\begin{theorem} \label{thm:Hamiltonian_Dyamics_Bayes}
	Suppose Assumptions \ref{asmp:domain_Bayes}--\ref{asmp:Hamiltonian_Bayes}.
	The tuple $(H, M, d)$ is an arc Hamiltonian system.
	Given any initial state $(\mone, \mtwo) \in M$, there exists a global flow $t \mapsto (\mone_t, \mtwo_t)$ of the system from $(\mone_0, \mtwo_0) = (\mone, \mtwo)$ in the indefinite interval $[0, \infty)$, along which the value of the Hamiltonian function $H(\mone_t, \mtwo_t)$ is invariant everywhere in $[0, \infty)$.
\end{theorem}

\begin{proof}[Proof of \Cref{thm:Hamiltonian_Dyamics_Bayes}]
	See \Cref{proof:Hamiltonian_Dyamics_Bayes}.
\end{proof}

This result provides a fruitful insight into infinite-dimensional geometrical properties associated with Bayesian inference.
By \Cref{thm:Hamiltonian_Bayes_variation}, the Hamiltonian arc field $\Phi$ is given by
\begin{align*}
	\Phi_s(f, P) = (f, P) + s (- \partial_2 H(f, P), \partial_1 H(f, P) ) = ( (1 + s) f + s f_*, (1 - s) P + s P_* ) .
\end{align*}
The Hamiltonian arc field consists of the arc field $(1 + s) f + s f_*$ on the space of negative log-likelihood and the arc field $(1 - s) P + s P_*$ on the space of probability measures.
	
\begin{remark} \label{rm:interpretation}
	Our result confers a geometrical significance of the transition from the prior $P$ to the posterior $P_*$ in the mixture geodesic $(1 - s) P + s P_*$ over $s \in [0, 1]$.
	Assignment of this transition to every probability measures $P$ acts as a certain field on the space of probability measures.
	That field is the arc field driven by the first variation $\partial_1 H(f, P)$ of the minimum free energy.
	It forms the Hamiltonian arc field with another arc field $(1 + s) f + s f_*$ driven by the other first variation $- \partial_2 H(f, P)$.
	Our result reveals the underlying invariance of the conservation of the minimum free energy.
	\Cref{fig:illustration} illustrates this dynamics.
\end{remark}

\begin{remark} \label{rm:contrast}
	As aforementioned in \Cref{sec:introduction}, the minimum free energy $H(f, P)$ equals to the free energy $\mathcal{F}(P_*)$ at the minimiser $P_*$ for the mean-zero potential $f - \int_\T f(\t) d P(\t)$.
	Our result formalises the dynamics of the potential $f_t$ and reference probability measure $P_t$ of the Gibbs measure $P_*$ that conserves the value of the free energy $\mathcal{F}(P_*)$ over time $t \in [0, \infty)$.
\end{remark}

\subsection{Illustration of Generated Flows} \label{sec:illustration}

While our results are all theoretical, we can approximately visualise a flow $t \mapsto (f_t, P_t)$ of the system of the minimum free energy.
We first consider two discrete cases where the parameter space $\Theta$ is (i) a set $\Theta = \{ 1 \}$ and (ii) a set $\Theta = \{ 1, 2, 3 \}$, so that the flow and the simplectic variation can be visualised in the Euclidean space.
We then consider a continuous case $\Theta = \R$, where we visualise approximated trajectories of the infinite-dimensional flow $t \mapsto (f_t, P_t)$ from two different initial state $(f_0, P_0)$.

In the first discrete case $\Theta = \{ 1 \}$, functions $f$ and probability measures $P$ on $\Theta$ are identified as scalars in the Euclidean space $\R$. 
Here, probability measures $P$ on $\Theta = \{ 1 \}$ all reduce to $P = 1$; for visualisation purpose we extend the minimum free energy $H(f, P)$ to non-probability measures $P > 0$, so that we can plot the value and the flow of the minimum free energy $H(f, P)$ over a two-dimensional domain $(0, \infty) \times (0, \infty)$.
\Cref{fig:one_dimensional_H} shows the flows driven by the simplectic variation $\partial H(f, P) = (f_* + f, P_* - P)$ and the value of the extended minimum free energy $H(f, P)$ on $(0, \infty) \times (0, \infty)$.
It can be observed that all flows run along the contour line where $H(f, P)$ is constant, and that the flow $t \mapsto (f_t, P_t)$ in the domain where $P$ is the probability measure, i.e.~$P = 1$, stays within the domain.

\begin{figure}[b]
	\includegraphics[width=0.37\textwidth]{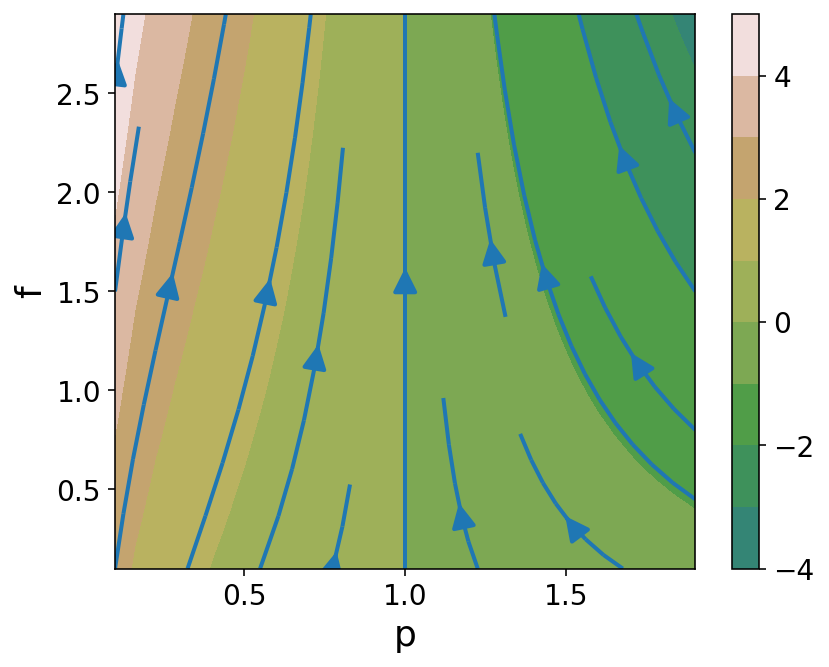}
	\caption{The value (contour colour map) and the generated flows (blue arrow) of the minimum free energy $H(f, P)$ extended to $(0, \infty) \times (0, \infty)$ when $\Theta = \{ 1 \}$. The value and flow at $P = 1$ corresponds to that of the minimum free energy in the original domain $(0, \infty) \times \{ 1 \}$.} \label{fig:one_dimensional_H}
\end{figure}

In the second discrete case $\Theta = \{ 1, 2, 3 \}$, functions $f$ and probability measures $P$ on $\Theta$ are identified as vectors in the Euclidean space $\R^3$.
We visualise the first variation $\partial_1 H(f, P) = P_* - P$ and the other first variation $- \partial_2 H(f, P) = f_* + f$ as three-dimensional vectors.
Let $f_0$ be a function s.t.~$(f_0(1), f_0(2), f_0(3)) = (2.0, 1.0, 0.5)$.
\Cref{fig:three_dimensional_H} (b) visualises the first variation $\partial_1 H(f_0, P)$ at the fixed $f_0$ and each point $P$ in the probability simplex in $\R^3$.
This illustrates what field on a space of probability measures is induced by the first variation $\partial_1 H(f_0, P)$ assigned at every $P$.
Let $P_0$ be a probability measure s.t.~$(P_0(1), P_0(2), P_0(3)) = (1/3, 1/3, 1/3)$.
\Cref{fig:three_dimensional_H} (a) visualises the first variation $- \partial_2 H(f, P_0)$ at the fixed $P_0$ and each point $f$, again, in the probability simplex in $\R^3$.
Here, we take each $f$ from the probability simplex for visualisation, but the first variation $- \partial_2 H(f, P_0)$ can be defined for all non-negative functions $f$.

\begin{figure}[t]
	\subcaptionbox{The first variation $- \partial_2 H(f, P_0)$}{\includegraphics[width=0.45\textwidth]{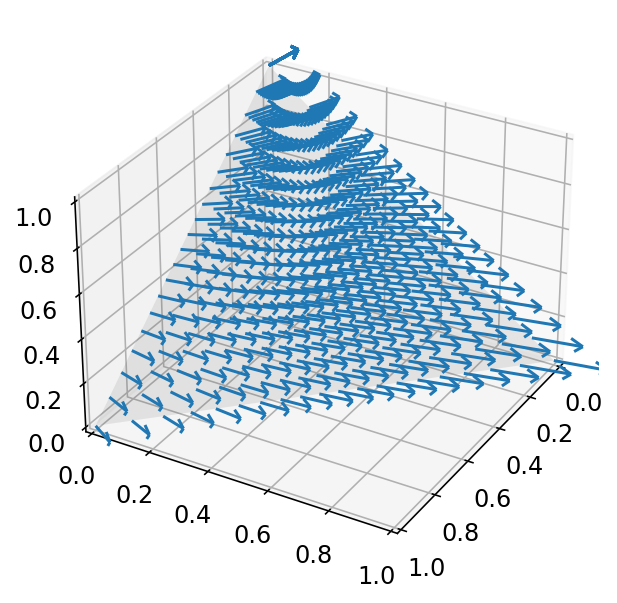}}
	\hfill
	\subcaptionbox{The first variation $\partial_1 H(f_0, P)$}{\includegraphics[width=0.45\textwidth]{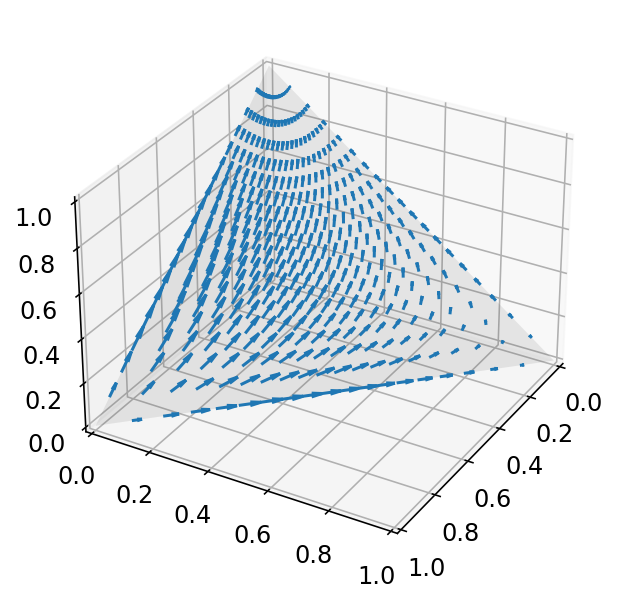}}
	\caption{The first variations (a) and (b) of the minimum free energy visualised as vectors in $\R^3$ when $\Theta = \{ 1, 2, 3 \}$. With a probability $P_0 = (1/3, 1/3, 1/3)$ fixed, the vector (a) is assigned at each $f$ in the simplex s.t.~$f(1) + f(2) + f(3) = 1$. With a function $f_0 = (2.0, 1.0, 0.5)$ fixed, the vector (b) is assigned at each $P$ in the simplex s.t.~$P(1) + P(2) + P(3) = 1$.} \label{fig:three_dimensional_H}
\end{figure}

Next, we illustrate an infinite-dimensional flow in the continuous case $\Theta = \R$.
We approximate a flow $t \mapsto (f_t, P_t)$ from $t = 0$ to $t = 3$ by the discretisation scheme in \eqref{eq:discretisation}:
\begin{align*}
	(f_{t + \delta}, P_{t + \delta}) \gets (f_t, P_t) + \delta \, \partial H(f_t, P_t) 
\end{align*}
where we use $\delta=0.001$.
By \Cref{thm:Hamiltonian_Bayes_variation}, the symplectic variation of the minimum free energy at each $(f, P) = (f_t, P_t)$ is given by $\partial H(f_t, P_t) = (f_* + f_t, P_* - P_t)$ with
\begin{align*}
    f_*(\t) = \frac{\exp(- f_t(\t))}{\int_\T \exp(- f_t(\t)) d P_t(\t)} \qquad \text{and} \qquad d P_* = \frac{\exp(- f_t(\t))}{\int_\T \exp(- f_t(\t)) d P_t(\t)} d P_t(\t).
\end{align*}
At any time $t$ in $[0, 3]$, the integral $\int_\T \exp(f_t(\t)) d P_t(\t)$ in the definition of $f_*$ and $P_*$ is numerically computed by the left-endpoint rule using $2000$ grid points over $[-10, 10] \subset \Theta$.
We identify the measure $P_t$ with the density $p_t$ with respect to the Lebesgue measure in $\Theta = \R$.
We compute the approximated flow $t \mapsto (f_t, p_t)$ at each time $t = 0, 1, 2, 3$, from two different initial states $(f_0, p_0)$.
Consider a normal location model $\mathcal{N}(\theta, 1)$ over data domain $\R$ with mean parameter $\theta \in \Theta$ and scale $1$.
The first initial state $(f_0, p_0)$ is defined by the negative log-likelihood $f_0(\theta) = ( x - \theta )^2 / 2 + \log( 2 \pi ) / 2$ of the normal location model for datum $x = 5.0$ and the density $p_0(\theta) = \exp( - \theta^2 / 2 ) / \sqrt{2 \pi}$.
Consider a Cauchy location model $\mathcal{C}(\theta, 1)$ over data domain $\R$ with mean parameter $\theta \in \Theta$ and scale $1$.
The second initial state $(f_0, p_0)$ is defined by the negative log-likelihood $f_0(\theta) = \log( 1 + (x - \theta)^2 ) + \log( \pi )$ of the Cauchy location model for datum $x = 5.0$ and the same density $p_0(\t) = \exp( - \theta^2 / 2 ) / \sqrt{2 \pi}$ as above.
\Cref{fig:gaussian_location,fig:cauchy_location} show the approximated flow $t \mapsto (f_t, P_t)$ from each initial state.

\begin{figure}[t]
	\includegraphics[width=\textwidth]{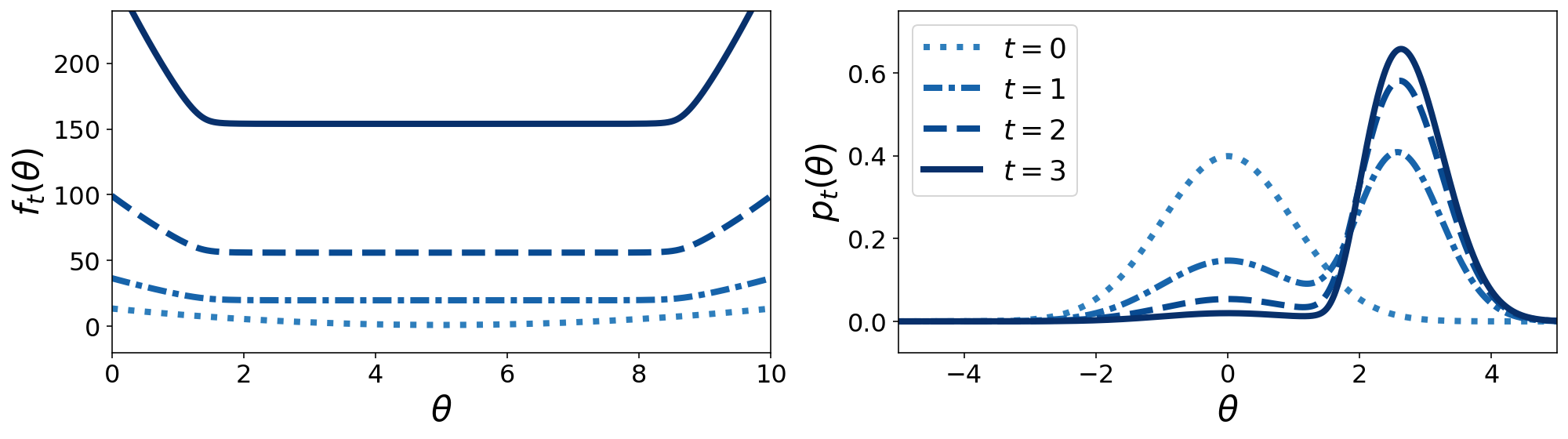}
	\caption{An approximated flow $t \mapsto (f_t, p_t)$ of the system of the minimum free energy from an initial state $(f_0, p_0)$ given by $f_0 = ( 5.0 - \theta )^2 / 2 + \log( 2 \pi ) / 2$ and $p_0 = \exp( - \theta^2 / 2 ) / \sqrt{2 \pi}$. The left panel shows the evolution of the negative log-likelihood $f_t$ and the right panel shows the evolution of the density $p_t$, respectively, at time $t = 0, 1, 2, 3$.\\~\\} \label{fig:gaussian_location}
	
	\includegraphics[width=\textwidth]{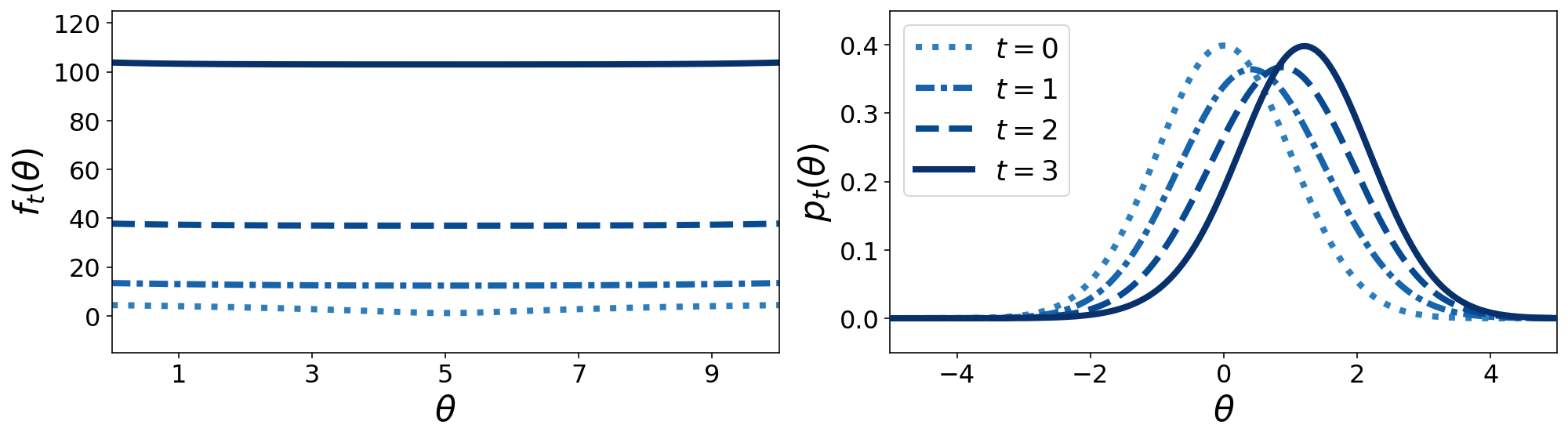}
	\caption{An approximated flow $t \mapsto (f_t, p_t)$ of the system of the minimum free energy from an initial state $(f_0, p_0)$ given by $f_0 = \log( 1 + (5.0 - \theta)^2 ) + \log( \pi )$ and $p_0 = \exp( - \theta^2 / 2 ) / \sqrt{2 \pi}$. The left panel shows the evolution of the negative log-likelihood $f_t$ and the right panel shows the evolution of the density $p_t$, respectively, at time $t = 0, 1, 2, 3$.} \label{fig:cauchy_location}
\end{figure}

%% file: section_05.tex

\section{Conclusion} \label{sec:conclusion}

This paper made two theoretical contributions.
First, we established the framework of arc Hamiltonian systems for saddle Hamiltonian functions whose domains are infinite-dimensional metric spaces.
We derived the general condition, under which a Hamiltonian arc field generates energy-conserving flows everywhere in a metric space.
We showed that the general condition reduces to local Lipschitz continuity of the first variation of the Hamiltonian function when the metric on $M$ has an extension to a norm-induced metric on $E_1 \times E_2$.
Second, we presented the arc Hamiltonian system of the minimum free energy, in which the difference of the posterior and the prior of Bayesian inference is characterised as the first variation of the minimum free energy $H(f, P)$ with respect to $f$.
Pairing this first variation with the other first variation with respect to $P$, we revealed the underlying Hamiltonian arc field behind the transition from the prior and the posterior, as highlighted in \Cref{rm:interpretation}.
The invariance of the conservation of the minimum free energy was established.

In this work, convexity of saddle Hamiltonian functions plays a vital role in the construction of our energy-conserving systems.
Arc Hamiltonian systems satisfy the conservation of energy, using only a topologically weak notion of derivative and gradient, i.e., the variational derivative and the first variation.
In contrast to the Fr\'{e}chet derivative, weak notions of derivative, such as the G\^{a}teaux derivative, no longer satisfy the chain rule in general \citep{Yamamuro1974}.
Nonetheless, convexity acts in place of the chain rule to guarantee the energy-conserving property of arc Hamiltonian systems.
We focused on saddle Hamiltonian functions that are independent of time $t$.
Relaxing time-independence of saddle Hamiltonian functions represents one of natural directions for future work.
An extension of arc fields to time-dependent cases was considered in \cite{Kim2013}.
In this paper, we applied the framework of arc Hamiltonian systems to formalise the infinite-dimensional geometrical property associated with Bayesian inference.
Apart from our application, there may be other intriguing dynamical phenomena or optimisation methodologies that benefit from infinite-dimensional Hamiltonian systems in metric spaces.
The general framework of arc Hamiltonian systems can naturally lead to the development of further applications.

%% file: supplement.tex

\appendix
\begin{frontmatter}
	\title{Supplemental Material for `Hamiltonian Dynamics of Bayesian Inference Formalised by Arc Hamiltonian Systems'}
	\runtitle{Hamiltonian Dynamics of Bayesian Inference}
	\begin{aug}
		\author{\fnms{Takuo}~\snm{Matsubara}}
		\address{School of Mathematics, The University of Edinburgh, \href{mailto:takuo.matsubara@ed.ac.uk}{takuo.matsubara@ed.ac.uk}}
	\end{aug}
	
%
	
\end{frontmatter}
\setcounter{equation}{0}
\setcounter{figure}{0}
\setcounter{table}{0}
\setcounter{page}{1}
\makeatletter
\renewcommand{\theequation}{S\arabic{equation}}
\renewcommand{\thefigure}{S\arabic{figure}}
\renewcommand{\bibnumfmt}[1]{[S#1]}
\renewcommand{\citenumfont}[1]{S#1}
\renewcommand{\appendixname}{}

This supplementary material contains proofs of all the theoretical results presented in the main text.
\Cref{sec:appendix_0} provides a concise summary of the theoretical results in \cite{Calcaterra2000} relevant to the theoretical development of arc Hamiltonian systems.
It also shows that a part of the original conditions can be relaxed.
\Cref{sec:appendix_a} contains proofs of all the theoretical results presented in \Cref{sec:methodology}.
\Cref{sec:appendix_b} contains proofs of all the theoretical results presented in \Cref{sec:application} apart from the derivation of the first variation.
The derivation of the first variation in \Cref{sec:application} is summarised in \Cref{sec:appendix_c}.
Finally, \Cref{sec:appendix_d} presents proofs of three intermediate lemmas used in the proof of \Cref{thm:Hamiltonian_Dyamics_Bayes}.

\input{appendix_01}

\input{appendix_02}

\input{appendix_03}

\input{appendix_04}

\input{appendix_05}

%% file: appendix_01.tex

\section{The Theory of Arc Fields and The Relaxation of The Condition} \label{sec:appendix_0}

Let $(M, d)$ be an arbitrary complete metric space and let $\Phi: M \times [0, 1] \to M$ be an arbitrary map s.t.~$\Phi_0(\mone) = \mone$ throughout \Cref{sec:appendix_0}.
Recall that the map $\Phi$ is called an arc field if it satisfies the condition in \Cref{def:cb_definition_21}.
\cite{Calcaterra2000} were concerned with a curve $t \mapsto x_t$ in $(M, d)$ that satisfies the following condition at each time $t$:
\begin{align}
    \lim_{s \to 0^+} \frac{ d( x_{t+s} , \Phi_s( x_t ) ) }{ s } = 0 . \label{eq:apx_mutationeq}
\end{align}
They established conditions for the map $\Phi$, under which a unique solution $t \mapsto \mone_t$ of \eqref{eq:apx_mutationeq} from $\mone_0 = a$ exists in the infinite interval $[0, \infty)$ for every starting point $a \in M$.

In \Cref{sec:appendix_0_1}, we provide a brief recap of their results relevant to the theoretical development of arc Hamiltonian systems.
We restate their results using our notation, remarking corresponding statements of \cite{Calcaterra2000} in each statement provided in \Cref{sec:appendix_0_1}.
We prefix `CB-' to statements of \cite{Calcaterra2000} refered here---e.g.,~CB-Definition 2.1---for simpler citation.
In \Cref{sec:appendix_0_2}, we show that a part of their original conditions can be relaxed.
\Cref{sec:appendix_0_3} contains the proof of the relaxation.

\subsection{The Theory of Arc Fields} \label{sec:appendix_0_1}

First, they studied under what condition a unique solution $t \mapsto x_t$ of \eqref{eq:apx_mutationeq} exists at least in some finite interval $[0, T)$ at each starting point.
They provided the following two fundamental conditions.

\begin{assumption}[CB-Condition AL] \label{asmp:cb_condition_al}
	For some function $\Lambda: M \times M \times [0, 1] \to \R$, at each $a \in M$ there exist a radius $r_a > 0$ and a length $\epsilon_a > 0$ s.t.
	\begin{align*}
		d(\Phi_t(\mone), \Phi_t(\mthree) ) \le d(\mone, \mthree) (1 + t \, \Lambda(\mone, \mthree, t))
	\end{align*}
	and $\Lambda(\mone, \mthree, t)$ is upper bounded for all $\mone, \mthree \in B_{r_a}(a)$ and $t \in [0, \epsilon_a)$.
\end{assumption}

\begin{assumption}[CB-Condition BL] \label{asmp:cb_condition_bl}
	For some function $\Omega: M \times [0, 1] \times [0, 1] \to [0, \infty)$, at each $a \in M$ there exist a radius $r_a > 0$ and a length $\epsilon_a > 0$ s.t.
	\begin{align*}
		d(\Phi_{t+s}(\mone), \Phi_s(\Phi_t(\mone))) \le t \, s \, \Omega(\mone, t, s))
	\end{align*}
	and $\Omega(\mone, t, s))$ is upper bounded for all $\mone \in B_{r_a}(a)$ and $t,s \in [0, \epsilon_a)$.
\end{assumption}

\noindent
Originally, the above inequality in CB-Condition BL was stated in a form
\begin{align*}
	d(\Phi_{t+s}(\mone), \Phi_s(\Phi_t(\mone))) \le t \, g(t, s) \, \Omega(\mone, t, s)) 
\end{align*}
using some function $g: [0, \epsilon_a) \times [0, \epsilon_a) \to [0, \infty)$.
However, CB-Remark 2.1 highlighted that it suffices to set $g(t, s) = s$ in CB-Condition BL.
We hence use $g(t, s) = s$ as our default choice in \Cref{asmp:cb_condition_bl} by combining CB-Condition BL and CB-Remark 2.1.

By CB-Definition 4.1, an arc field $\Phi$ is said to have \emph{unique solutions for short time} if there exists an interval $[0, T_a)$ at each point $a \in M$, in which a unique solution $t \mapsto \mone_t$ of \eqref{eq:apx_mutationeq} exists from $\mone_0 = a$.
This means that, if two non-unique solutions of \eqref{eq:apx_mutationeq} exist from $\mone_0 = a$ in some interval larger than $[0, T_a)$, they must equal to each other at least in the interval $[0, T_a)$.
Then they showed in CB-Proposition 4.2 that, if an arc field $\Phi$ has unique solutions for short time, a \emph{maximal} solution (i.e.~the longest solution that can exist) at each point is, in fact, unique.
In other words, any solution of \eqref{eq:apx_mutationeq}, shown to exists in some interval, turns out to be a part of the underlying longest unique solution.
CB-Proposition 4.2 formalised this intuition.

\begin{proposition}[CB-Proposition 4.2] \label{prop:cb_proposition_42}
	Assume that an arc field $\Phi$ has unique solutions for short time. At each point $a \in M$, there exists a unique solution $t \mapsto \mone_t$ of \eqref{eq:apx_mutationeq} from $\mone_0 = a$ in some interval $[0, T_a^*)$ s.t.~for any solution $t \mapsto \hat{\mone}_t$ of \eqref{eq:apx_mutationeq} from $\hat{\mone}_0 = a$ whose interval is $[0, T_a)$, we have $T_a^* \ge T_a$ and $t \mapsto \hat{\mone}_t$ equals to $t \mapsto \mone_t$ everywhere in $[0, T_a)$.
\end{proposition}

\noindent
They concisely summarised in CB-Corollary 4.3 that an arc field $\Phi$ has unique solutions for short time under CB-Condition AL and CB-Condition BL.

\begin{corollary}[CB-Corollary 4.3] \label{cor:cb_corollary_43}
	An arc field $\Phi$ that satisfies Assumptions \ref{asmp:cb_condition_al}--\ref{asmp:cb_condition_bl} has unique solutions for short time, and hence has a unique maximal solution of \eqref{eq:apx_mutationeq} starting from each point in $M$.
\end{corollary}

In CB-Corollary 4.3, the interval in which each unique maximal solution exists was implicit.
They next studied under what additional condition each unique maximal solution exists globally in the infinite interval $[0, \infty)$.
For simplicity, we drop the term `maximal' for any unique solution in $[0, \infty)$, because a unique solution in $[0, \infty)$ is trivially the longest solution.
They used a condition, termed linear speed growth, and showed the existence of a unique solution in the infinite interval $[0, \infty)$ in CB-Theorem 4.4.

\begin{definition}[CB-Definition 4.2] \label{def:cb_definition_42}
	An arc field $\Phi$ is said to have \emph{linear speed growth} if there exists a point $\mone \in M$ and some constant $C_1, C_2$ dependent only on $\mone$ s.t.
	\begin{align*}
		\bar{\rho}(\mone, r) \le C_1 \, r + C_2
	\end{align*}
	holds for all $r > 0$, where $\bar{\rho}$ is defined in \Cref{def:cb_definition_21}.
\end{definition}

\begin{theorem}[CB-Theorem 4.4] \label{thm:cb_theorem_44}
	Assume that an arc field $\Phi$ satisfies Assumptions \ref{asmp:cb_condition_al}--\ref{asmp:cb_condition_bl} and has linear speed growth.
	At each $a \in M$, there exists a unique solution $t \mapsto \mone_t$ of \eqref{eq:apx_mutationeq} from $\mone_0 = a$ in $[0, \infty)$.
\end{theorem}

\noindent
Note that each unique solution above is continuous while this was implicit in CB-Theorem 4.4.
Continuity of each solution is made explicit in our statement in \Cref{sec:appendix_0_2}. 
We move on to \Cref{sec:appendix_0_2}, where we weaken CB-Definition 4.2 and reproduce CB-Theorem 4.4.

\subsection{The Relaxation of The Condition} \label{sec:appendix_0_2}

The condition of CB-Definition 4.2 for an arc field $\Phi$ requires that the inequality $\bar{\rho}(\mone, r) \le C_1 \, r + C_2$ holds for all radius $r$ in $(0, \infty)$.
This is in contrast to the inequalities in Assumptions \ref{asmp:cb_condition_al}--\ref{asmp:cb_condition_bl} that need to hold only for all radius $r$ smaller than some constant $r_a$ at each point $a \in M$.
The condition of CB-Definition 4.2 may be restrictive for some application; in the theory of arc Hamiltonian systems, we use local Lipschitz continuity of the first variation of a Hamiltonian function to show the existence of a flow.
We conclude \Cref{sec:appendix_0} by showing that the condition of CB-Definition 4.2 can be relaxed in a manner that an analogous inequality holds only for all radius $r$ smaller than some constant $r_a$ at each point $a \in M$.
The function $\rho$ in \Cref{asmp:cb_definition_42_relaxed} is defined in \Cref{def:cb_definition_21}.

\begin{assumption} \label{asmp:cb_definition_42_relaxed}
	At each $a \in M$, there exist a radius $r_a > 0$ s.t.
	\begin{align*}
		\rho(a) \le C_1 d(\mone, a) + C_2 
	\end{align*}
	holds for all $\mone \in B_{r_a}(a)$, where $C_1, C_2$ are constants dependent only on $a$.
\end{assumption}

\noindent
Here, we note that a map $\Phi$ that satisfies \Cref{asmp:cb_definition_42_relaxed} is in fact an arc field.

\begin{lemma} \label{lem:cb_theorem_44_relaxed_lemma}
	A map $\Phi$ that satisfies \Cref{asmp:cb_definition_42_relaxed} is an arc field on $(M, d)$.
\end{lemma}

\begin{proof}
	Recall from \Cref{def:cb_definition_21} that $\bar{\rho}(\mone, r) = \sup_{ \mthree \in B_r(\mone) } \rho(\mthree)$.
	It follows from \Cref{asmp:cb_definition_42_relaxed} that each $a \in M$ there exist a radius $r_a > 0$ s.t.~it holds for all $0 < r \le r_a$ that
	\begin{align*}
		\bar{\rho}(a, r) = \sup_{ \mthree \in B_r(a) } \rho(\mthree) \le \sup_{ \mthree \in B_r(a) } \left( C_1 d(\mthree, a) + C_2 \right) \le C_1 \, r_a + C_2 
	\end{align*}
	where the last inequality is immediate from $\mthree \in B_r(a)$ and $r \le r_a$.
	Since $\bar{\rho}(a, r) < \infty$ for all $0 < r \le r_a$ at each $a \in M$, we conclude that $\Phi$ satisfies \Cref{def:cb_definition_21} and is an arc field.
\end{proof}

\noindent
Now, we reproduce CB-Theorem 4.4 under \Cref{asmp:cb_definition_42_relaxed} without using the condition of CB-Definition 4.2.
Continuity of each solution is made explicit in the following \Cref{thm:cb_theorem_44_relaxed}.

\begin{theorem} \label{thm:cb_theorem_44_relaxed}
	Suppose that a map $\Phi$ satisfies Assumptions \ref{asmp:cb_condition_al}--\ref{asmp:cb_definition_42_relaxed}. 
	At each $a \in M$, there exists a unique continuous solution $t \mapsto \mone_t$ of \eqref{eq:apx_mutationeq} from $\mone_0 = a$ in $[0, \infty)$.
\end{theorem}

\noindent
The proof of \Cref{thm:cb_theorem_44_relaxed} is provided in \Cref{sec:appendix_0_3}.
\Cref{thm:cb_theorem_44_relaxed} is the main result that directly used in the theoretical development of arc Hamiltonian systems.
This completes the brief recap of the theory of arc fields and the relaxation of the condition.

\subsection{Proof of \Cref{thm:cb_theorem_44_relaxed}} \label{sec:appendix_0_3}

We introduce an intermediate lemma used in the main proof.

\begin{lemma} \label{lem:cb_continuity}
	Suppose that for a point $a \in M$ an arc field $\Phi$ has a unique solution $t \mapsto \mone_t$ from $\mone_0 = a$ in some interval.
	Given any radius $r > 0$, there exists some length $T_a > 0$ s.t.~the unique solution $t \mapsto \mone_t$ from $\mone_0 = a$ satisfies $\mone_t \in B_r(a)$ for all $t \in [0, T_a)$.
\end{lemma}

\begin{proof}[Proof of \Cref{lem:cb_continuity}]
	Denote by $[0, T_a^*)$ the interval in which a unique solution $t \mapsto \mone_t$ from $\mone_0 = a$ is supposed to exist.
	Let $t$ be an arbitrary point in $[0, T_a^*)$ and let $s$ be any variable in $[0, T_a^* - t)$, where we have $0 \le t + s < T_a^*$.
	By the triangle inequality, we have
	\begin{align*}
		d(\mone_{t+s}, \mone_t) & \le d(\mone_{t+s}, \Phi_s(\mone_t)) + d(\Phi_s(\mone_t), \mone_t) \\
		& = d(\mone_{t+s}, \Phi_s(\mone_t)) + d(\Phi_s(\mone_t), \Phi_0(\mone_t)) \\
		& = s \cdot \frac{ d(\mone_{t+s}, \Phi_s(\mone_t)) }{ s } + s \cdot \frac{ d(\Phi_s(\mone_t), \Phi_0(\mone_t)) }{ s } .
	\end{align*}
	In the above inequality, we plug the following trivial inequality
	\begin{align*}
		\frac{ d(\Phi_s(\mone_t), \Phi_0(\mone_t)) }{ s } \le \sup_{ s_1 \ne s_2 \in [0, 1] } \frac{ d( \Phi_{s_1}(\mone_t), \Phi_{s_2}(\mone_t)) }{ | s_1 - s_2 | } = \rho(\mone_t)
	\end{align*}
	where $\rho$ is the function defined in \Cref{def:cb_definition_21}.
	We take the limit $s \to 0^+$ to see that
	\begin{align*}
		\lim_{s \to 0^+} d(\mone_{t+s}, \mone_t) & = \lim_{s \to 0^+} \left( s \cdot \frac{ d(\mone_{t+s}, \Phi_s(\mone_t)) }{ s } + s \cdot \rho(\mone_t) \right) \\
		& = \underbrace{ \lim_{s \to 0^+} s }_{ = 0 } \cdot \underbrace{ \lim_{s \to 0^+} \frac{ d(\mone_{t+s}, \Phi_s(\mone_t)) }{ s } }_{ =: (*) } + \underbrace{ \lim_{s \to 0^+} s }_{ = 0 } \cdot \underbrace{\vphantom{\lim_{s \to 0^+}} \rho(\mone_t) }_{ < \infty} .
	\end{align*}
	We have $(*) = 0$ by the condition \eqref{eq:apx_mutationeq} and have $\rho(\mone_t) < \infty$ by \Cref{def:cb_definition_21}.
	We set $t = 0$ to see that $\lim_{s \to 0^+} d(\mone_s, \mone_0) = \lim_{s \to 0^+} d(\mone_s, a) = 0$.
	By definition of convergence, for any value $r > 0$, we can take some value $T_a > 0$ so that $d(\mone_s, a) < r$ holds for all $s \in [0, T_a)$.
\end{proof}

\noindent
Now, we present the main proof of \Cref{thm:cb_theorem_44_relaxed}.

\begin{proof}[Proof of \Cref{thm:existence_of_global_flow}]
	By \Cref{lem:cb_theorem_44_relaxed_lemma}, the map $\Phi$ is an arc field under \Cref{asmp:cb_definition_42_relaxed}.
	Let $a$ be an arbitrary point in $M$.
	Let $r_a > 0$ be the radius that the inequality in \Cref{asmp:cb_definition_42_relaxed} holds at $a \in M$.
	By \Cref{cor:cb_corollary_43}, the arc field $\Phi$ has unique solutions for short time.
	Hence a unique solution $t \mapsto \mone_t$ from $\mone_0 = a$ exists in some interval.
	By \Cref{lem:cb_continuity}, there exists some $T_a > 0$ s.t.~the solution $t \mapsto \mone_t$ from $\mone_0 = a$ satisfies $\mone_t \in B_{r_a}(a)$ for all $t \in [0, T_a)$.

	Our aim is to show that the unique solution $t \mapsto \mone_t$ in $[0, T_a)$ can be uniquely and indefinitely extended to $[0, \infty)$, by the same argument in the proof of CB-Theorem 4.4.
	Let $\{ t_n \}_{n=1}^{\infty}$ be a Cauchy sequence in $[0, T_a)$ s.t.~$t_n \to T_a$.
	A key part of the proof is to show that 
	\begin{center}
		$(\star)$ the solution sequence $\{ \mone_{t_n} \}_{n=1}^{\infty}$ from $\mone_{t_1} = a$ is Cauchy.
	\end{center}
	If item $(\star)$ holds, the solution $t \mapsto \mone_t$ from $\mone_0 = a$ converges to some point $b \in M$ as $t \to T_a$.
	This implies that, at $t = T_a$, the solution $t \mapsto \mone_t$ connects to the other unique solution $t \mapsto \hat{\mone}_t$ running from $\hat{\mone}_0 = b$ in some interval $[0, T_b)$.
	Since we show item $(\star)$ for arbitrary $a$, the other solution $t \mapsto \hat{\mone}_t$ from $b$ is also Cauchy and converges to some point $c \in M$ as $t \to T_b$.
	By repeating this process, the original solution from $a$ can be continued over $[0, \infty)$.
	By \Cref{prop:cb_proposition_42}, the continued solution in $[0, \infty)$ is unique, which concludes the existence of a unique solution in $[0, \infty)$.
	Continuity of the solution is clarified in the end of this proof.
	
	It suffices to prove item $(\star)$.
	We repeat the argument in the proof of CB-Theorem 4.4 for item $(\star)$, using only \Cref{asmp:cb_definition_42_relaxed}.
	Let $f(t) := d( \mone_t, \mone_0 )$ for any $t \in [0, T_a)$ where $\mone_0 = a$.
	For any $t \in [0, T_a)$ and $h \in [0, T_a - t)$, it follows from the reverse triangle inequality that
	\begin{align*}
		f(t + h) - f(t) = d( \mone_{t+h}, \mone_0 ) - d( \mone_t, \mone_0 ) \le d( \mone_{t+h}, \mone_t ) .
	\end{align*}
	It further follows from the triangle inequality that
	\begin{align*}
		f(t + h) - f(t) & \le d( \mone_{t+h}, \Phi_h(\mone_t) ) + d( \Phi_h(\mone_t), \mone_t ) \\
		& = d( \mone_{t+h}, \Phi_h(\mone_t) ) + h \cdot \frac{ d( \Phi_h(\mone_t), \Phi_0(\mone_t) ) }{ h }  \\
		& \le d( \mone_{t+h}, \Phi_h(\mone_t) ) + h \cdot \rho(\mone_t) .
	\end{align*}
	Dividing both the side by $h$ and taking the limit in $h \to 0^+$, we have
	\begin{align*}
		\lim_{h \to 0^+} \frac{ f(t + h) - f(t) }{ h } & \le \lim_{h \to 0^+} \frac{ d( \mone_{t+h}, \Phi_h(\mone_t) ) }{ h } + \rho(\mone_t) = \rho(\mone_t) .
	\end{align*}
	where $\lim_{h \to 0^+} d( \mone_{t+h}, \Phi_h(\mone_t) ) / h = 0$ due to the condition \eqref{eq:apx_mutationeq}.
	Recall that $\mone_t \in B_{r_a}(\mone_0)$ holds for all $t \in [0, T_a)$, where $\mone_0 = a$.
	Since $r_a$ is the radius under which the inequality in \Cref{asmp:cb_definition_42_relaxed} holds at $a$, we can apply the inequality of $\rho(\mone_t)$ for all $t \in [0, T_a)$, that is,
	\begin{align}
		\lim_{h \to 0^+} \frac{ f(t + h) - f(t) }{ h } & \le C_1 d( \mone_t, \mone_0 ) + C_2 = C_1 f(t) + C_2 . \label{eq:proof_fth_ft_condition}
	\end{align}
	This inequality \eqref{eq:proof_fth_ft_condition} implies a bound $f(t) \le ( C_2 / C_1 ) ( \exp( C_1 t ) - 1 )$.
	To verify this, we set $g(t) := \exp( - C_1 t ) ( f(t) + C_2 / C_1 )$ and calculate its right-derivative
	\begin{align*}
		\lim_{h \to 0^+} \frac{ g(t + h) - g(t) }{ h } & = \exp( - C_1 t ) \left( \lim_{h \to 0^+} \frac{ f(t + h) - f(t) }{ h } - C_1 f(t) - C_2 \right) \le 0 
	\end{align*}
	where the right-hand side is non-positive because of the inequality \eqref{eq:proof_fth_ft_condition}.
	Since the right-derivative is non-positive, $g$ decreases over $[0, T_*)$ and hence $g(t) \le g(0)$, which leads to
	\begin{align*}
		g(t) = \exp( - C_1 t ) \left( f(t) + \frac{ C_2 }{ C_1 } \right) \le \frac{ C_2 }{ C_1 } = g(0)
	\end{align*}
	where we used $f(0) = 0$.
	Rearranging the term confirms that we have the intended bound $f(t) \le ( C_2 / C_1 ) ( \exp( C_1 t ) - 1 )$.
	Since $f(t) = d(\mone_t, \mone_0)$, we arrive at a main bound:
	\begin{align}
		d(\mone_t, \mone_0) \le \frac{C_2}{C_1} ( \exp( C_1 t ) - 1 ) . \label{eq:proof_data0_bound}
	\end{align}
	Let $s$ be an arbitrary point in $[0, T_a)$.
	Define $\mthree_t := \mone_{s+t}$ for any $t \in[0, T_* - s)$, where $\mthree_0 = \mone_s$.
	We aim to derive an upper bound of $d( \mthree_t, \mthree_0 )$ similar to \eqref{eq:proof_data0_bound}.
	The above argument for $f(t) = d( \mone_t, \mone_0 )$ depends only on the condition $\rho(\mone_t) \le C_1 d( \mone_t, \mone_0 ) + C_2$.
	If a corresponding condition $\rho(\mthree_t) \le C_1' d( \mthree_t, \mthree_0 ) + C_2'$ holds for some constants $C_1'$ and $C_2'$, we have
	\begin{align}
		d(\mthree_t, \mthree_0) \le \frac{C_2'}{C_1'} ( \exp( C_1' t ) - 1 ) . \label{eq:proof_dbtb0_bound}
	\end{align}
	In fact, the required condition holds for $C_1' = C_1$ and $C_2' = C_1 d( \mone_s, \mone_0 ) + C_2$ because
	\begin{align*}
		\rho(\mone_{s+t}) \le C_1 \, d(\mone_{s+t}, \mone_0) + C_2 \le C_1 \, d(\mone_{s+t}, \mone_s) + C_1 \, d(\mone_s, \mone_0) + C_2
	\end{align*}
	where the first inequality follows from the inequality in \Cref{asmp:cb_definition_42_relaxed} and the last inequaliy follows from the triangle inequality.
	Applying \eqref{eq:proof_data0_bound} for $d(\mone_s, \mone_0)$ in the constant $C_2'$, we further have $C_2' \le C_2 \exp( C_1 T_a )$.
	We plug these constants $C_1'$ and $C_2'$ in \eqref{eq:proof_dbtb0_bound} to see that
	\begin{align*}
		d(\mone_{s+t}, \mone_s) & \le \frac{C_2 \exp( C_1 T_a )}{C_1} | \exp( C_1 t ) - \exp(0) | .
	\end{align*}
	Since the mean value theorem implies $| \exp( C_1 t ) - \exp( 0 ) | \le \exp( C_1 T_a ) | C_1 t - 0 |$ for any $t \in [0, T_a)$, we have, for arbitrary $s \in [0, T_a)$ and $t \in [0, T_a - s)$, that
	\begin{align*}
		d(\mone_{s+t}, \mone_s) & \le C_2 \exp( 2 C_1 T_a ) \, | t | .
	\end{align*}
	Since the metric $d$ are symmetric, we can rewrite this inequality as
	\begin{align}
		d(\mone_t, \mone_s) & \le C_2 \exp( 2 C_1 T_a ) \, | t - s | \label{eq:proof_dbtb0_lipschitz}
	\end{align}
	for any $s, t \in [0, T_a)$.
	This inequality implies Lipschitz continuity of the solution $t \mapsto \mone_t$ in $[0, T_a)$.
	We conclude that the solution sequence $\{ \mone_{t_n} \}_{n=1}^{\infty}$ is Cauchy as intended, because we have $d(\mone_{t_n}, \mone_{t_m}) \to 0$ whenever $| t_n - t_m | \to 0$ due to \eqref{eq:proof_dbtb0_lipschitz}.
	
	Finally, we make explicit continuity of the continued solution in $[0, \infty)$.
	As aforementioned, a solution $t \mapsto \mone_t$ from $\mone_0 = a$ in $[0, T_a)$ converges to some point $b \in M$ at $t = T_a$ and connects to the other solution $t \mapsto \hat{\mone}_t$ from $\hat{\mone}_0 = b$ in $[0, T_b)$.
	Since $a$ is arbitrary in this proof, both the solution $t \mapsto \mone_t$ from $\mone_0 = a$ and the solution $t \mapsto \hat{\mone}_t$ from $\hat{\mone}_0 = b$ satisfy Lipschitz continuity \eqref{eq:proof_dbtb0_lipschitz}.
	Thus, the connected solution over $[0, T_a + T_b)$ is continuous if it is continuous at $b$.
	It is trivially continuous at $b$ since the two solutions are concatenated at $b$.
	By repeating this process, the solution $t \mapsto \mone_t$ from $a$ continues over $[0, \infty)$ continuously.
\end{proof}

%% file: appendix_02.tex

\section{Proofs for Results in Section 3} \label{sec:appendix_a}

\subsection{Proof of \Cref{thm:existence_of_global_flow}} \label{proof:existence_of_global_flow}

Theorem 4.4 of \cite{Calcaterra2000} established conditions for an arbitrary arc field $\Phi: M \times [0, 1] \to M$, under which a unique solution of the associated condition \eqref{eq:emutationeq} exists in $[0, \infty)$ from any initial point in $M$.
See \Cref{sec:appendix_0_1} for a summary of their results.
We established in \Cref{thm:cb_theorem_44_relaxed} in \Cref{sec:appendix_0_2} that a part of their original assumptions can be relaxed.
\Cref{thm:existence_of_global_flow} is a direct application of \Cref{thm:cb_theorem_44_relaxed} stated under the relaxed assumptions suitable for our theoretical development.

\begin{proof}
	By definition, a global flow of the system from a point $(\mone, \mtwo)$ is a unique continuous solution $t \mapsto (\mone_t, \mtwo_t)$ of the arc Hamilton equation \eqref{eq:H_mutation_equation} from $(\mone_0, \mtwo_0) = (\mone, \mtwo)$ in the infinite interval $[0, \infty)$.
	\Cref{thm:cb_theorem_44_relaxed} shows that a unique continuous solution of the condition \eqref{eq:emutationeq} exists in the infinite interval $[0, \infty)$ from every initial point in $M$ if the associated map $\Phi$ satisfies Assumptions \ref{asmp:cb_condition_al}--\ref{asmp:cb_definition_42_relaxed}.
	For the arc Hamilton equation \eqref{eq:H_mutation_equation}, \Cref{asmp:flow_existence_1} restates Assumptions \ref{asmp:cb_condition_al}--\ref{asmp:cb_condition_bl} combined together and \Cref{asmp:flow_existence_2} restates \Cref{asmp:cb_definition_42_relaxed}.
\end{proof}

\subsection{Proof of \Cref{lem:existence_of_flow_simplified}} \label{proof:existence_of_flow_simplified}

In \Cref{proof:existence_of_flow_simplified}, let $\| \cdot \|_E$ be the norm on $E = E_1 \times E_2$ that induces the metric $d_*$ on $E$ in \Cref{asmp:flow_existence_3} item (1).
Since $d_*$ is an extension of $d$ from $M$ to $E$, we have $d( (\mone, \mtwo), (\mthree, \mfour) ) = d_*( (\mone, \mtwo), (\mthree, \mfour) ) = \| (\mone, \mtwo) - (\mthree, \mfour) \|_E$ everywhere in $M$.
We introduce an intermediate lemma used in the main proof.

\begin{lemma} \label{lem:equality_psi}
	Suppose \Cref{asmp:flow_existence_3}.
	For any $(\mone, \mtwo) \in M$, we have
	\begin{align*}
		d((\mone, \mtwo), \Phi_s(\mone, \mtwo)) = s \| \partial H(\mone, \mtwo) \|_E .
	\end{align*}
	At each $(\mone_0, \mtwo_0) \in M$, it holds for all $(\mone, \mtwo) \in B_{\rll}(\mone_0, \mtwo_0)$ that
	\begin{align*}
		\| \partial H(\mone, \mtwo) \|_E \le K_0 \, d((\mone, \mtwo), (\mone_0, \mtwo_0)) + \| \partial H(\mone_0, \mtwo_0) \|_E 
	\end{align*}
	where $\delta_0$ is the radius and $K_0$ is the constant of local Lipshitz continuity \eqref{eq:local_lipschitz} at $(\mone_0, \mtwo_0)$.
\end{lemma}

\begin{proof}[Proof of \Cref{lem:equality_psi}]
	The first equality is immediate from
	\begin{align*}
		d( (\mone, \mtwo), \Phi_s(\mone, \mtwo) ) & = \| (\mone, \mtwo) - (\mone, \mtwo) - s \, \partial H(\mone, \mtwo) ) \|_E = s \| \partial H(\mone, \mtwo) \|_E .
	\end{align*}
	For any $(\mone, \mtwo) \in B_{\rll}(\mone_0, \mtwo_0)$, we apply the triangle inequality to see that
	\begin{align*}
		\| \partial H(\mone, \mtwo) \|_E & \le \| \partial H(\mone, \mtwo) - \partial H(\mone_0, \mtwo_0) \|_E + \| \partial H(\mone_0, \mtwo_0) \|_E  \\
		& = d_*( \partial H(\mone, \mtwo), \partial H(\mone_0, \mtwo_0) ) + \| \partial H(\mone_0, \mtwo_0) \|_E \\
		& \le K_0 \, d( (\mone, \mtwo), (\mone_0, \mtwo_0) ) + \| \partial H(\mone_0, \mtwo_0) \|_E
	\end{align*}
	where the last inequality follows from local Lipschitz continuity \eqref{eq:local_lipschitz}.
\end{proof}

\noindent
Now, we provide the main proof of \Cref{lem:existence_of_flow_simplified}.

\begin{proof}[Proof of \Cref{lem:existence_of_flow_simplified}]
	First, we specify the functions $\Lambda: M \times M \times [0, 1] \to \R$ and $\Omega: M \times [0, 1] \times [0, 1] \to [0, \infty)$, for which \Cref{asmp:flow_existence_1} is shown to hold:
	\begin{align*}
		\Lambda((\mone, \mtwo), (\mthree, \mfour), t) & = \frac{ d_*( \partial H(\mone, \mtwo) , \partial H(\mthree, \mfour) ) }{ d((\mone, \mtwo), (\mthree, \mfour)) } , \\
		\Omega((\mone, \mtwo), t, s) & = \Lambda((\mone, \mtwo), \Phi_s(\mone, \mtwo), t) \| \partial H(\mone, \mtwo) \|_E . 
	\end{align*}
	Here, we verify $\Lambda$ and $\Omega$ are well-defined at each point.
	If $(\mone, \mtwo) \ne (\mthree, \mfour)$, the value of $\Lambda$ at each point is finite because $d((\mone, \mtwo), (\mthree, \mfour)) \ne 0$.
	Even if $(\mone, \mtwo) = (\mthree, \mfour)$, the value of $\Lambda$ is finite because local Lipschitz continuity \eqref{eq:local_lipschitz} implies that, for all $(\mone, \mtwo), (\mthree, \mfour) \in B_{\rll}(\mone_0, \mtwo_0)$,
	\begin{align}
		\Lambda((\mone, \mtwo), (\mthree, \mfour), t) \le \frac{ K_0 \, d((\mone, \mtwo), (\mthree, \mfour))  }{ d((\mone, \mtwo), (\mthree, \mfour)) } = K_0 \label{eq:inequality_item1_1}
	\end{align}
	at each point $(\mone_0, \mtwo_0) \in M$.
	Therefore, $\Lambda$ is well-defined at each point, so does $\Omega$.
	
	In what follows, let $(\mone_0, \mtwo_0)$ be an arbitrary point in $M$. 
	Let $\rll > 0$ be the radius that local Lipschitz continuity \eqref{eq:local_lipschitz} in \Cref{asmp:flow_existence_3} holds at $(\mone_0, \mtwo_0)$.
	Let $(\mone, \mtwo), (\mthree, \mfour)$ be arbitrary points in an open ball $B_{r_0}(\mone_0, \mtwo_0)$ for a given radius $r_0 \in (0, \infty)$. 
	Let $t,s$ be arbitrary points in an interval $[0, \epsilon_0)$ for a given length $\epsilon_0 \in (0, 1]$.
	First, we show that there exist some radius $r_0 > 0$ and length $\epsilon_0 > 0$ s.t.~\Cref{asmp:flow_existence_1} items (1)--(2) hold.
	
	\vspace{5pt}
	\noindent
	\underline{\textit{\Cref{asmp:flow_existence_1} item (1)}}:
	It follows from the triangle inequality that
	\begin{align*}
		d(\Phi_t(\mone, \mtwo), \Phi_t(\mthree, \mfour) ) & = \| (\mone, \mtwo) + t \partial H(\mone, \mtwo) - (\mthree, \mfour) - t \partial H(\mthree, \mfour) \|_E \\
		& \le \| (\mone, \mtwo) - (\mthree, \mfour) \|_E  + \| t \partial H(\mone, \mtwo) - t \partial H(\mthree, \mfour) \|_E \\
		& = d_*((\mone, \mtwo), (\mthree, \mfour)) + t d_*( \partial H(\mone, \mtwo), \partial H(\mthree, \mfour) ) .
	\end{align*}
	Since $d_*( (\mone, \mtwo), (\mthree, \mfour) ) = d( (\mone, \mtwo), (\mthree, \mfour) )$ everywhere in $M$, we have
	\begin{align*}
		d(\Phi_t(\mone, \mtwo), \Phi_t(\mthree, \mfour) ) & \le d( (\mone, \mtwo), (\mthree, \mfour) ) + t \, d_*( \partial H(\mone, \mtwo) , \partial H(\mthree, \mfour) ) \\
		& = d( (\mone, \mtwo), (\mthree, \mfour) ) \, \left( 1 + t \, \Lambda( (\mone, \mtwo), (\mthree, \mfour) , t ) \right) .
	\end{align*}
	If $(\mone, \mtwo), (\mthree, \mfour)$ are contained in $B_{\rll}(\mone_0, \mtwo_0)$, the term $ \Lambda( (\mone, \mtwo), (\mthree, \mfour) , t )$ is bounded by $K_0$ as established in \eqref{eq:inequality_item1_1}.
	Selecting any $r_0 \le \rll$ and $\epsilon_0 \in (0, 1]$ concludes \Cref{asmp:flow_existence_1} item (1).
	
	\vspace{5pt}
	\noindent
	\underline{\textit{\Cref{asmp:flow_existence_1} item (2)}}:
	Denote $( \tilde{\mone}_s, \tilde{\mtwo}_s ) := \Phi_s(\mone, \mtwo) = (\mone, \mtwo) + s \, \partial H(\mone, \mtwo)$ to see
	\begin{align}
		d(\Phi_{t+s}(\mone, \mtwo), \Phi_t( \Phi_s(\mone, \mtwo) ) ) & = \| (\mone, \mtwo) + (t + s) \, \partial H(\mone, \mtwo) - ( \tilde{\mone}_s, \tilde{\mtwo}_s ) - t \, \partial H(\tilde{\mone}_s, \tilde{\mtwo}_s) \|_E \nonumber \\
		& = \| ( \tilde{\mone}_s, \tilde{\mtwo}_s ) + t \, \partial H(\mone, \mtwo) - ( \tilde{\mone}_s, \tilde{\mtwo}_s ) - t \, \partial H(\tilde{\mone}_s, \tilde{\mtwo}_s) \|_E \nonumber \\
		& = t ~ d_*( \partial H(\mone, \mtwo) , \partial H(\tilde{\mone}_s, \tilde{\mtwo}_s) ) \nonumber \\
		& = t ~ \Lambda( (\mone, \mtwo), (\tilde{\mone}_s, \tilde{\mtwo}_s), t ) ~ d( (\mone, \mtwo), (\tilde{\mone}_s, \tilde{\mtwo}_s) ) . \nonumber
	\end{align}
	Since $d( (\mone, \mtwo), (\tilde{\mone}_s, \tilde{\mtwo}_s) ) = s \| \partial H(\mone, \mtwo) \|_E$ by \Cref{lem:equality_psi}, it follows from definition of $\Omega$ that
	\begin{align*}
		d(\Phi_{t+s}(\mone, \mtwo), \Phi_t( \Phi_s(\mone, \mtwo) ) ) & = t \, s \, \Omega( (\mone, \mtwo), s, t ) .
	\end{align*}
	We can conclude \Cref{asmp:flow_existence_1} item (2) if there exist some $r_0 > 0$ and $\epsilon_0 > 0$ s.t.~$\Omega( (\mone, \mtwo), s, t )$ is bounded for all $(\mone, \mtwo) \in B_{r_0}(\mone_0, \mtwo_0)$ and all $s, t \in [0, \epsilon_0)$.
	Since by definition we have $\Omega((\mone, \mtwo), t, s) = \Lambda((\mone, \mtwo), \Phi_s(\mone, \mtwo), t) \| \partial H(\mone, \mtwo) \|_E$, we will upper bound each component $\Lambda((\mone, \mtwo), \Phi_s(\mone, \mtwo), t)$ and $\| \partial H(\mone, \mtwo) \|_E$.
	If we have $(\mone, \mtwo) \in B_{\rll}(\mone_0, \mtwo_0)$, \Cref{lem:equality_psi} shows that $\| \partial H(\mone, \mtwo) \|_E \le K r_0 + \| \partial H(\mone_0, \mtwo_0) \|_E < \infty$.
	It thus suffices to take any $r_0 \le \rll$ to upper bound the second component.
	It is established in \eqref{eq:inequality_item1_1} that $\Lambda( (\mone, \mtwo), \Phi_s(\mone, \mtwo), t ) \le K_0$ if both $(\mone, \mtwo)$ and $\Phi_s(\mone, \mtwo)$ are contained in $B_{\rll}(\mone_0, \mtwo_0)$.
	In fact, if we have $(\mone, \mtwo) \in B_{\rll/2}(\mone_0, \mtwo_0)$, there exists $\epsilon_0 > 0$ s.t.~$\Phi_s(\mone, \mtwo) \in B_{\rll}(\mone_0, \mtwo_0)$ for all $s \in [0, \epsilon_0)$.
	To verify this, see that \Cref{lem:equality_psi} implies $\lim_{s \to 0^+} d((\mone, \mtwo), \Phi_s(\mone, \mtwo)) = 0$.
	By definition of convergence, for any given $r > 0$ there exists some $\epsilon_r > 0$ s.t.~$d((\mone, \mtwo), \Phi_s(\mone, \mtwo)) < r$ for all $s \in [0, \epsilon_r)$.
	That is, given the radius $\rll/2$, there exists some $\epsilon_* > 0$ s.t.~$\Phi_s(\mone, \mtwo)$ is contained in $B_{\rll/2}(\mone, \mtwo)$ for all $s \in [0, \epsilon_*)$.
	Clearly, the open ball $B_{\rll/2}(\mone, \mtwo)$ around $(\mone, \mtwo)$ is further contained in the open ball $B_{\rll}(\mone_0, \mtwo_0)$ around $(\mone_0, \mtwo_0)$ if we have $(\mone, \mtwo) \in B_{\rll/2}(\mone_0, \mtwo_0)$.
	This means that, if we take any $r_0 \le \rll / 2$, then $(\mone, \mtwo)$ and $\Phi_s(\mone, \mtwo)$ are both contained in $B_{\rll}(\mone_0, \mtwo_0)$ for all $s \in [0, \epsilon_*)$.
	Therefore, selecting any $r_0 \le \rll / 2$ and $\epsilon_0 \le \epsilon_*$ guarantees that $\Omega( (\mone, \mtwo), s, t )$ is bounded.

	\vspace{5pt}
	We complete the proof by showing \Cref{asmp:flow_existence_2}.
	
	\vspace{5pt}
	\noindent
	\underline{\textit{\Cref{asmp:flow_existence_2}}}:
	Since $d(\Phi_t(\mone, \mtwo), \Phi_s(\mone, \mtwo)) = \| \Phi_t(\mone, \mtwo) - \Phi_s(\mone, \mtwo) \|_E$, we have
	\begin{align*}
		\rho(\mone, \mtwo) & = \sup_{t \ne s \in [0, 1]} \frac{ \| (\mone, \mtwo) + t \, \partial H(\mone, \mtwo) - (\mone, \mtwo) - s \, \partial H(\mone, \mtwo) \|_E }{ | t - s| } = \| \partial H(\mone, \mtwo) \|_E . 
	\end{align*}
	By selecting any $r_0 \le \rll$ so that $(\mone, \mtwo) \in B_{\rll}(\mone_0, \mtwo_0)$, it follows from \Cref{lem:equality_psi} that
	\begin{align}
		\rho(\mone, \mtwo) & = \| \partial H(\mone, \mtwo) \|_E \le K_0 \, d((\mone, \mtwo), (\mone_0, \mtwo_0)) + \| \partial H(\mone_0, \mtwo_0) \|_E . \label{eq:proof_rho_ph}
	\end{align}
	Setting $C_1 = K_0$ and $C_2 = \| \partial H(\mone_0, \mtwo_0) \|_E$ concludes \Cref{asmp:flow_existence_2}.
\end{proof}

\subsection{Proof of \Cref{prop:conservation_of_energy}} \label{proof:conservation_of_energy}

If a convex function $f: E_1 \to \Re$ admits the first variation $\partial f(\mone) \in E_2$ at $\mone \in \Dom(f)$, the following inequality holds for any $\mthree \in \Dom(f)$:
\begin{align}
	\langle \mthree - \mone, \partial f(\mone) \rangle = \lim_{\epsilon \to 0^+} \frac{ f( (1 - \epsilon) \mone + \epsilon \mthree ) - f(\mone) }{ \epsilon } \le f(\mthree) - f(\mone) . \label{eq:Gateaux_subgradient}
\end{align} 
The first equality follows from definition of the first variation, where $\mthree - \mone$ is a trivial element of $C(\mone)$.
The last inequality follows from convexity $f( (1 - \epsilon) \mone + \epsilon \mthree ) \le (1 - \epsilon) f( \mone ) + \epsilon f( \mthree )$.
If $f: E_1 \to \Re$ is a concave function, we have the inequality with the less-than sign flipped
\begin{align}
	\langle \mthree - \mone, \partial f(\mone) \rangle \ge f(\mthree) - f(\mone) . \label{eq:Gateaux_subgradient_2}
\end{align} 
Note that the essentially same inequalities hold for a convex or concave function in $E_2$.
We introduce a lemma that applies \eqref{eq:Gateaux_subgradient} and \eqref{eq:Gateaux_subgradient_2} for a saddle function $H: E_1 \times E_2 \to \Re$.

\begin{lemma} \label{lem:saddle_bound}
	Suppose that a saddle function $H: E_1 \times E_2 \to \Re$ has the first variations $\partial_1 H(\mone, \mtwo)$ and $\partial_2 H(\mone, \mtwo)$ everywhere in $\Dom(H)$.
	For any $(\mone, \mtwo), (\mthree, \mfour) \in \Dom(H)$,
	\begin{align*}
		H(\mthree, \mfour) - H(\mone, \mtwo) & \ge \langle \mthree -  \mone, \partial_1 H(\mthree, \mfour) \rangle + \langle \partial_2 H(\mone, \mtwo), \mfour - \mtwo \rangle , \\
		H(\mthree, \mfour) - H(\mone, \mtwo) & \le \langle \mthree - \mone, \partial_1 H(\mone, \mfour) \rangle + \langle \partial_2 H(\mone, \mfour), \mfour - \mtwo \rangle .
	\end{align*}
\end{lemma}
	
\begin{proof}[Proof of \Cref{lem:saddle_bound}]
	Let $(\mone, \mtwo), (\mthree, \mfour)$ be arbitrary points in $\Dom(H)$.
	Denote $H_{\mfour}(\cdot) := H(\cdot, \mfour)$ for fixed $\mfour$ and $H_{\mone}(\cdot) := H(\mone, \cdot)$ for fixed $\mone$, where $H_{\mfour}$ is concave and $H_{\mone}$ is convex.
	By definition, the first variation of $H_{\mfour}$ at $\mone$ is $\partial_1 H(\mone, \mfour)$, and the first variation of $H_{\mone}$ at $\mtwo$ is $\partial_2 H(\mthree, \mtwo)$.
	We apply \eqref{eq:Gateaux_subgradient_2} for both (i) $H_{\mfour}(\mthree) - H_{\mfour}(\mone)$ and (ii) $H_{\mfour}(\mone) - H_{\mfour}(\mthree)$ to see that
	\begin{align*}
		\text{(i)}~& \langle \mthree -  \mone, \partial_1 H(\mone, \mfour) \rangle \ge H_{\mfour}(\mthree) - H_{\mfour}(\mone) , \\
		\text{(ii)}~& \langle \mone - \mthree, \partial_1 H(\mthree, \mfour) \rangle \ge H_{\mfour}(\mone) - H_{\mfour}(\mthree) .
	\end{align*}
	We multiply both sides of the inequality (ii) by $-1$ to have that
	\begin{align}
		\langle \mthree - \mone, \partial_1 H(\mthree, \mfour) \rangle \le H_\mfour(\mthree) - H_\mfour(\mone) \le \langle \mthree -  \mone, \partial_1 H(\mone, \mfour) \rangle . \label{eq:H1_subgradient}
	\end{align}
	Similarly, we apply \eqref{eq:Gateaux_subgradient} for both (iii) $H_{\mone}(\mfour) - H_{\mone}(\mtwo)$ and (iv) $H_{\mone}(\mtwo) - H_{\mone}(\mfour)$ to see that
	\begin{align*}
		\text{(iii)}~& \langle \partial_2 H(\mone, \mtwo), \mfour -  \mtwo \rangle \le H_{\mone}(\mfour) - H_{\mone}(\mtwo) , \\
		\text{(iv)}~& \langle \partial_2 H(\mone, \mfour), \mtwo - \mfour \rangle \le H_{\mone}(\mtwo) - H_{\mone}(\mfour) .
	\end{align*}
	We multiply both sides of the inequality (iv) by $-1$ to have that
	\begin{align}
		\langle \partial_2 H(\mone, \mtwo), \mfour -  \mtwo \rangle \le H_{\mone}(\mfour) - H_{\mone}(\mtwo) \le \langle \partial_2 H(\mone, \mfour), \mfour - \mtwo \rangle  . \label{eq:H2_subgradient}
	\end{align}
	Adding all sides of \eqref{eq:H1_subgradient} and \eqref{eq:H2_subgradient} completes the proof.
\end{proof}

\noindent
Now we present the main proof of \Cref{prop:conservation_of_energy}.

\begin{proof}[Proof of \Cref{prop:conservation_of_energy}]
	It suffices to show that (i) $t \mapsto H( \monet, \mtwot )$ is continuous in $[0, T)$ and (ii) the right derivative $(\d^+ / \d t) H( \monet, \mtwot )$ equals to zero everywhere in $[0, T)$.
	First of all, item (i) immediately holds by combining local Lipschitz continuity \eqref{eq:H_local_lipschitz}
	\begin{align*}
		| H( \mone_{s}, \mtwo_{s} ) - H( \mone_{t}, \mtwo_{t} ) | \le J_0 \, d( (\mone_{s}, \mtwo_{s}), ( \mone_{t}, \mtwo_{t}) ) 
	\end{align*}
	with continuity of the flow, that is, $\lim_{s \to t} d( (\mone_{s}, \mtwo_{s}), ( \mone_{t}, \mtwo_{t}) ) = 0$ for any $t \in (0, T)$ with the edge case $\lim_{s \to 0^+} d( (\mone_{s}, \mtwo_{s}), ( \mone_{0}, \mtwo_{0}) ) = 0$.
	For item (ii), by definition of the right derivative, 
	\begin{align*}
		\frac{\d^+}{\d t} H( \monet, \mtwot ) & = \lim_{s \to 0^+} \frac{ H( \mone_{t+s}, \mtwo_{t+s} ) - H( \monet, \mtwot ) }{s} \\
		& = \underbrace{ \lim_{s \to 0^+} \frac{ H( \mone_{t+s}, \mtwo_{t+s} ) - H( \Phi_s(\mone_t, \mtwo_t) ) }{s} }_{ =: (*_1) } + \underbrace{ \lim_{s \to 0^+} \frac{ H( \Phi_s(\mone_t, \mtwo_t) ) - H( \monet, \mtwot ) }{s} }_{ =: (*_2) } .
	\end{align*}
	In the remainder, let $t$ be arbitrary fixed time in $[0, T)$, and we show that both the term $(*_1)$ and $(*_2)$ equal to zero.
	This in turn concludes that item (ii) holds as intended.
	
	\vspace{5pt}
	\noindent
	\underline{\textit{Term $(*_1)$}}:
	Let $\eta_0 > 0$ be the radius that local Lipschitz continuity \eqref{eq:H_local_lipschitz} in \Cref{asmp:conservation_of_energy_1} holds at the point $(\mone_t, \mtwo_t)$.
	We verify that $(\mone_{t+s}, \mtwo_{t+s})$ and $\Phi_s(\mone_t, \mtwo_t)$ are contained in $B_{\rll}(\mone_t, \mtwo_t)$ for all $s \ge 0$ sufficiently small.
	It is trivial that $(\mone_{t+s}, \mtwo_{t+s}) \in B_{\rll}(\mone_t, \mtwo_t)$ for all $s \ge 0$ sufficiently small by continuity of the flow.
	It also holds that $\Phi_s(\mone_t, \mtwo_t) \in B_{\rll}(\mone_t, \mtwo_t)$ for all $s \ge 0$ sufficiently small, because it follows from the triangle inequality that
	\begin{align*}
		\lim_{s \to 0^+} d( ( \mone_t, \mtwo_t ),  \Phi_s(\mone_t, \mtwo_t) ) & \le \underbrace{ \lim_{s \to 0^+} d( ( \mone_t, \mtwo_t ), (\mone_{t+s}, \mtwo_{t+s}) ) }_{ = 0 } + \underbrace{ \lim_{s \to 0^+} d( ( \mone_{t+s}, \mtwo_{t+s} ),  \Phi_s(\mone_t, \mtwo_t) ) }_{ = 0 }
	\end{align*}
	where the first term is $0$ by continuity of the flow and the condition \eqref{eq:H_mutation_equation} is strong enough to imply that the second term is $0$.
	Thus, by local Lipschitz continuity \eqref{eq:H_local_lipschitz}, we have
	\begin{align*}
		\lim_{s \to 0^+} \frac{ | H( \mone_{t+s}, \mtwo_{t+s} ) - H( \Phi_s(\mone_t, \mtwo_t) ) | }{s} & \le \lim_{s \to 0^+} \frac{ J_0 \, d( ( \mone_{t+s}, \mtwo_{t+s} ),  \Phi_s(\mone_t, \mtwo_t) ) }{s} = 0 
	\end{align*}
	where the last equality holds by the condition \eqref{eq:H_mutation_equation}.
	Thus the term $(*_1)$ equals to $0$.
	
	\vspace{5pt}
	\noindent
	\underline{\textit{Term $(*_2)$}}: 
	For better presentation, denote $(\mone_t, \mtwo_t)$ simply by $(\mone, \mtwo)$ since $t$ is a fixed time.
	Denote $( \mone_s^*, \mtwo_s^* ) := \Phi_s(\mone, \mtwo) = ( \mone, \mtwo) + s ( - \partial_2 H(\mone, \mtwo), \partial_1 H(\mone, \mtwo) )$ to rewrite the term $(*_2)$ as
	\begin{align*}
		(*_2) = \lim_{s \to 0^+} \frac{ H( \Phi_s(\mone, \mtwo) ) - H( \mone, \mtwo ) }{s} = \lim_{s \to 0^+} \frac{ H( \mone_s^*, \mtwo_s^* ) - H( \mone, \mtwo ) }{s} .
	\end{align*}
	We shall derive upper and lower bounds of the term inside the limit, showing that the bounds both converge to $0$ in the limit.
	We begin with the lower bound.
	It follows from \Cref{lem:saddle_bound} that
	\begin{align*}
		\frac{ H( \mone_s^*, \mtwo_s^* ) - H( \mone, \mtwo ) }{s} & \ge \frac{ \langle \mone_s^* - \mone , \partial_1 H(\mone_s^*, \mtwo_s^*) \rangle + \langle \partial_2 H(\mone, \mtwo), \mtwo_s^* - \mtwo \rangle }{ s } .
	\end{align*}
	We then plug $\mone_s^* - \mone = s ( - \partial_2 H(\mone_t, \mtwo_t) )$ and $\mtwo_s^* - \mtwo  = s \, \partial_1 H(\mone_t, \mtwo_t)$ in, to see that
	\begin{align*}
		\frac{ H( \mone_s^*, \mtwo_s^* ) - H( \mone, \mtwo ) }{s} & \ge - \langle \partial_2 H(\mone, \mtwo) , \partial_1 H(\mone_s^*, \mtwo_s^*) \rangle + \langle \partial_2 H(\mone, \mtwo), \partial_1 H(\mone, \mtwo) \rangle =: (*_a) .
	\end{align*}
	Here, it is clear that we have $(\mone_s^*, \mtwo_s^*) \to (\mone, \mtwo)$ in $(M, d)$ in the limit $s \to 0^+$.
	It then follows from \Cref{asmp:conservation_of_energy_1} item (2) that we have $\langle \mthree, \partial_1 H(\mone_s^*, \mtwo_s^*) \rangle \to \langle \mthree, \partial_1 H(\mone, \mtwo) \rangle$ for any $\mthree \in E_1$.
	By applying this convergence with $\mthree = \partial_2 H(\mone, \mtwo)$, we have
	\begin{align*}
		\lim_{s \to 0^+} (*_a) = - \langle \partial_2 H(\mone, \mtwo) , \partial_1 H(\mone, \mtwo) \rangle + \langle \partial_2 H(\mone, \mtwo), \partial_1 H(\mone, \mtwo) \rangle = 0 .
	\end{align*}
	We move on to the upper bound.
	It follows from \Cref{lem:saddle_bound} that
	\begin{align*}
		\frac{ H( \mone_s^*, \mtwo_s^* ) - H( \mone, \mtwo ) }{s} & \le \frac{ \langle \mone_s^* - \mone , \partial_1 H(\mone, \mtwo_s^*) \rangle + \langle \partial_2 H(\mone, \mtwo_s^*), \mtwo_s^* - \mtwo \rangle }{ s } \\
		& = - \langle \partial_2 H(\mone, \mtwo) , \partial_1 H(\mone, \mtwo_s^*) \rangle + \langle \partial_2 H(\mone, \mtwo_s^*), \partial_1 H(\mone, \mtwo) \rangle =: (*_b) .
	\end{align*}
	Here, it is clear that we have $(\mone, \mtwo_s^*) \to (\mone, \mtwo)$ in $(M, d)$ in the limit $s \to 0^+$.
	We apply the essentially same argument as above for the terms $\partial_1 H(\mone, \mtwo_s^*)$ and $\partial_2 H(\mone, \mtwo_s^*)$ to see that
	\begin{align*}
		\lim_{s \to 0^+} (*_b) = - \langle \partial_2 H(\mone, \mtwo) , \partial_1 H(\mone, \mtwo) \rangle + \langle \partial_2 H(\mone, \mtwo), \partial_1 H(\mone, \mtwo) \rangle = 0 .
	\end{align*}
	Therefore, we have $\lim_{s \to 0^+} (*_a) \le (*_2) \le \lim_{s \to 0^+} (*_b)$, which concludes that the term $(*_2)$ equals to $0$ since the upper and lower bounds both converge to $0$.
\end{proof}

\subsection{Proof of \Cref{thm:existence_conservation}} \label{proof:existence_conservation}

\begin{proof}[Proof of \Cref{thm:existence_conservation}]
	It follows from \Cref{lem:existence_of_flow_simplified} that \Cref{asmp:flow_existence_3} implies Assumptions \ref{asmp:flow_existence_1}--\ref{asmp:flow_existence_2}.
	Then, by \Cref{thm:existence_of_global_flow}, there exists a global flow $t \mapsto (\monet, \mtwot)$ of the system from any initial state in $M$.
	We complete the proof by showing that the value of the Hamiltonian function $H( \monet, \mtwot )$ along the flow $t \mapsto (\monet, \mtwot)$ is constant everywhere in $[0, \infty)$.
	
	Our aim is to show that \Cref{prop:conservation_of_energy} holds.
	To this end, we show that \Cref{asmp:conservation_of_energy_1} item (1)--(2) are each implied by \Cref{asmp:flow_existence_3} item (2) under \Cref{asmp:conservation_of_energy_3}.
	Recall that, under \Cref{asmp:conservation_of_energy_3}, the metric $d_*$ is defined as $d_*((\mone, \mtwo), (\mthree, \mfour)) = \| \mone - \mthree \|_{E_1} + \| \mtwo - \mfour \|_{E_2}$ and we have the inequality $| \langle \mone, \mtwo \rangle | \le \| \mone \|_{E_1} \| \mtwo \|_{E_2}$.
	In what follows, let $(\mone_0, \mtwo_0)$ be an arbitrary point in $M$.
	Let $\rll > 0$ be the radius that local Lipschitz continuity \eqref{eq:local_lipschitz} in \Cref{asmp:flow_existence_3} holds at $(\mone_0, \mtwo_0)$.
	Let $(\mone, \mtwo), (\mthree, \mfour)$ be arbitrary points in the open ball $B_{\rll/2}(\mone_0, \mtwo_0)$.
	For better presentation, we show \Cref{asmp:conservation_of_energy_1} item (2) first and item (1) next.
	
	\vspace{5pt}
	\noindent
	\underline{\textit{\Cref{asmp:conservation_of_energy_1} item (2)}}:
	First, we show part (i) in \Cref{asmp:conservation_of_energy_1} item (2).
	By \Cref{asmp:conservation_of_energy_3}.
	\begin{align*}
		\langle \mthree_\star, \partial_1 H(\mthree, \mfour) - \partial_1 H(\mone, \mtwo) \rangle \le \| \mthree_\star \|_{E_1} \| \partial_1 H(\mthree, \mfour) - \partial_1 H(\mone, \mtwo) \|_{E_2}
	\end{align*}
	holds for any $\mthree_\star \in E_1$.
	It is trivial that $\| \partial_1 H(\mthree, \mfour) - \partial_1 H(\mone, \mtwo) \|_{E_2} \le d_*( \partial H(\mthree, \mfour), \partial H(\mone, \mtwo) )$ from definition of $d_*$.
	We apply this bound to see that
	\begin{align*}
		\langle \mthree_\star, \partial_1 H(\mthree, \mfour) - \partial_1 H(\mone, \mtwo) \rangle \le \| \mthree_\star \|_{E_1} d_*( \partial H(\mthree, \mfour), \partial H(\mone, \mtwo) ) .
	\end{align*}
	Then, local Lipschitz continuity in \Cref{asmp:flow_existence_3} item (2) implies that the right-hand side converges to $0$ whenever $(\mthree, \mfour) \to (\mone, \mtwo)$ in $(M, d)$.
	This concludes part (i) of \Cref{asmp:conservation_of_energy_1} item (2).
	The essentially same argument holds for part (ii) of \Cref{asmp:conservation_of_energy_1} item (2).
	
	\vspace{5pt}
	\noindent
	\underline{\textit{\Cref{asmp:conservation_of_energy_1} item (1)}}:
	By \Cref{lem:saddle_bound}, we have $(*_a) \le H(\mthree, \mfour) - H(\mone, \mtwo) \le (*_b)$ where 
	\begin{align*}
		(*_a) & =  \langle \mthree -  \mone, \partial_1 H(\mthree, \mfour) \rangle + \langle \partial_2 H(\mone, \mtwo), \mfour - \mtwo \rangle , \\
		(*_b) &= \langle \mthree - \mone, \partial_1 H(\mone, \mfour) \rangle + \langle \partial_2 H(\mone, \mfour), \mfour - \mtwo \rangle .
	\end{align*}
	By the inequality $| \langle \mone, \mtwo \rangle | \le \| \mone \|_{E_1} \| \mtwo \|_{E_2}$ in \Cref{asmp:conservation_of_energy_3}, we further have
	\begin{align*}
		| (*_a) | & =  \| \mthree -  \mone \|_{E_1} \| \partial_1 H(\mthree, \mfour) \|_{E_2} + \| \partial_2 H(\mone, \mtwo) \|_{E_1} \| \mfour - \mtwo \|_{E_2} , \\
		| (*_b) | &= \| \mthree - \mone \|_{E_1} \| \partial_1 H(\mone, \mfour) \|_{E_2} + \| \partial_2 H(\mone, \mfour) \|_{E_1} \| \mfour - \mtwo \|_{E_2} .
	\end{align*}
	For any $(\mthree_\star, \mfour_\star) \in M$, the following inequality is trivial by definition of $d_*$:
	\begin{align*}
		\| \partial_1 H(\mthree_\star, \mfour_\star) \|_{E_2} \le \| \partial H(\mthree_\star, \mfour_\star) \|_{E} \quad \text{and} \quad \| \partial_2 H(\mthree_\star, \mfour_\star) \|_{E_1} \le \| \partial H(\mthree_\star, \mfour_\star) \|_E .
	\end{align*} 
	If $(\mthree_\star, \mfour_\star) \in B_{\rll}(\mone_0, \mtwo_0)$, it follows from local Lipschitz continuity \eqref{eq:local_lipschitz} around $(\mone_0, \mtwo_0)$ that 
	\begin{align*}
		\| \partial H(\mthree_\star, \mfour_\star) \|_E & \le \| \partial H(\mthree_\star, \mfour_\star) - \partial H(\mone_0, \mtwo_0) \|_E + \| \partial H(\mone_0, \mtwo_0) \|_E \\
		& = d_*(\partial H(\mthree_\star, \mfour_\star), \partial H(\mone_0, \mtwo_0)) + \| \partial H(\mone_0, \mtwo_0) \|_E \\
		& \le K_0 \, d((\mthree_\star, \mfour_\star), (\mone_0, \mtwo_0)) + \| \partial H(\mone_0, \mtwo_0) \|_E \\
		& \le K_0 \, \rll + \| \partial H(\mone_0, \mtwo_0) \|_E = : U_0
	\end{align*}
	where the first last inequality follows from the triangle inequality.
	Note that $U_0$ is a constant dependent only on $(\mone_0, \mtwo_0)$.
	Hence, for any $(\mthree_\star, \mfour_\star) \in B_{\rll}(\mone_0, \mtwo_0)$, we have
	\begin{align}
		\| \partial_1 H(\mthree_\star, \mfour_\star) \|_{E_2} \le U_0 \quad \text{and} \quad \| \partial_2 H(\mthree_\star, \mfour_\star) \|_{E_1} \le U_0 . \label{eq:norm_ub_combined}
	\end{align} 
	We aim to apply this bound \eqref{eq:norm_ub_combined} for all the terms $\| \partial_1 H(\mthree, \mfour) \|_{E_2}$, $\| \partial_2 H(\mone, \mtwo) \|_{E_1}$, $\| \partial_1 H(\mone, \mfour) \|_{E_2}$, and $\| \partial_2 H(\mone, \mfour) \|_{E_1}$ used in the bound $| (*_a) |$ and $| (*_b) |$.
	To do so, we need to verify that each point $(\mone, \mtwo), (\mthree, \mfour), (\mone, \mfour)$---at which the above terms are evaluated---is contained in $B_{\rll}(\mone_0, \mtwo_0)$.
	It is trivial that $(\mone, \mtwo), (\mthree, \mfour) \in B_{\rll}(\mone_0, \mtwo_0)$ since we let $(\mone, \mtwo), (\mthree, \mfour) \in B_{\rll/2}(\mone_0, \mtwo_0)$ at first.
	It can be further verified that $(\mone, \mfour) \in B_{\rll}(\mone_0, \mtwo_0)$ because
	\begin{align*}
		d((\mone, \mfour), (\mone_0, \mtwo_0)) = \| \mone - \mone_0 \|_{E_1} + \| \mfour - \mtwo_0 \|_{E_2} \le \underbrace{ d((\mone, \mtwo), (\mone_0, \mtwo_0)) }_{ < \rll / 2 } + \underbrace{ d((\mthree, \mtwo), (\mone_0, \mtwo_0)) }_{ < \rll /2 } .
	\end{align*} 
	Therefore, by applying \eqref{eq:norm_ub_combined} for all the terms, we upper bound $| (*_a) |$ and $| (*_b) |$ as follows:
	\begin{align*}
		| (*_a) | & = U_0 \, \| \mthree -  \mone \|_{E_1} + U_0 \, \| \mfour - \mtwo \|_{E_2} = U_0 \, d((\mone, \mtwo), (\mthree, \mfour)) , \\
		| (*_b) | & = U_0 \, \| \mthree - \mone \|_{E_1} + U_0 \, \| \mfour - \mtwo \|_{E_2} = U_0 \, d((\mone, \mtwo), (\mthree, \mfour)) .
	\end{align*}
	By taking the absolute value of all the side of $(*_a) \le H(\mthree, \mfour) - H(\mone, \mtwo) \le (*_b)$, we have
	\begin{align*}
		| H(\mthree, \mfour) - H(\mone, \mtwo) | \le \max( | (*_a) |, | (*_b) | ) \le U_0 ~ d((\mone, \mtwo), (\mthree, \mfour)) .
	\end{align*}
	This bound holds for all $(\mone, \mtwo), (\mthree, \mfour) \in B_{\rll/2}(\mone_0, \mtwo_0)$, where $U_0$ depends only on $(\mone_0, \mtwo_0)$.
	Setting $J_0 = U_0$ and $\eta_0 = \rll / 2$ concludes \Cref{asmp:conservation_of_energy_1} item (1).
\end{proof}

%% file: appendix_03.tex

\section{Proofs for Results in Section 4} \label{sec:appendix_b}

\subsection{Proof of \Cref{prop:M_completeness}} \label{sec:appendix_b1}

\begin{proof}[Proof of \Cref{prop:M_completeness}]
	Let $B_1$ be a set of all continuous functions $f: \T \to \R$ that satisfies $\sup_{\t \in \Theta} | f(\t) | / w(\t) < \infty$. 
	Let $B_2$ be a set of all signed measures $P \in \Ms$ that satisfies $\int_{\T} w(\t) \d | P |(\t) < \infty$.
	Given the metrics $\varrho$ and $\gamma$ in \eqref{eq:wu_tv_metrics}, the metric spaces $(B_1, \varrho)$ and $(B_2, \gamma)$ are complete if $w$ is continuous and strictly positive \citep[p.4]{Kolokoltsov2019}.
	The product metric space $(B, d_B)$ s.t.~$B := B_1 \times B_2$ and $d_B((f, P), (g, Q)) := \varrho(f, g) + \gamma(\P, \Q)$ is then complete.
	Here $(M, d)$ is a metric subspace of $(B, d_B)$ because $M \subset B$ and $d$ is the restriction of $d_B$ to $M$ by definition.
	By a standard fact on complete metric spaces \citep[p.74]{Aliprantis2006}, the metric subspace $(M, d)$ is complete if $M$ is closed in $(B, d_B)$.
	
	The set $M$ is closed in $(B, d_B)$ if any convergent sequence of $(B, d_B)$ taken in $M$ has its limit within $M$.
	Let $\{ (f_n, P_n) \}_{n = 1}^{\infty} \subset M$ be an arbitrary convergent sequence of $(B, d_B)$ taken in $M$, whose limit is denoted $(f, P) \in B$.
	Our aim is to show that $(f, P) \in M$.
	Since the limit $(f, P)$ is at least an element of $B \subset E_1 \times E_2$, it suffices to verify that the limit $(f, P)$ satisfies the condition: (i) $f$ satisfies $f(\t) \ge \log w(\t)$ and (ii) $P$ is a probability measure.
	
	\vspace{5pt}
	\noindent
	\underline{\textit{Condition (i)}}: 
	First, we have $f_n(\t) \ge \log w(\t)$ at all $n$ because $\{ (f_n, P_n) \}_{n = 1}^{\infty} \subset M$.
	The uniform convergence in $\varrho$ implies the pointwise convergence $f(\t) = \lim_{n \to \infty} f_n(\t)$.
	At each $\t \in \Theta$, we take the limit of $n$ in the both side of  $f_n(\t) \ge \log w(\t)$ to see that
	\begin{align*}
		f(\t) = \lim_{n \to \infty} f_n(\t) \ge \lim_{n \to \infty} \log w(\t) = \log w(\t) .
	\end{align*}
	The function $f$ of the limit $(f, P)$ thus satisfies the condition (i).
	
	\vspace{5pt}
	\noindent
	\underline{\textit{Condition (ii)}}:
	Note that convergence of $P_n$ in the metric $\gamma$ implies that
	\begin{align*}
		\int_O \d P(\t) = \lim_{n \to 0} \int_{O} \d P_n(\t) ~~ \text{for any measurable set} ~~ O \subset \T
	\end{align*}
	because $\int_{O} \d ( P - P_n )(\t) \le \int_{\T} w(\t) \d | P - P_n |(\t)$ due to $w(\t) \ge 1$.
	The measure $P$ is non-negative because $P_n$ is non-negative for all $n$ and we have $\int_O \d P(\t) = \lim_{n \to 0} \int_{O} \d P_n(\t) \ge 0$.
	Similarly, the measure $P$ satisfies $\int_\T \d P(\t) = 1$ because $\int_\T \d P_n(\t) = 1$ for all $n$ and we have $\int_\T \d P(\t) = \lim_{n \to 0} \int_{\T} \d P_n(\t) = 1$.
	The measure $P$ thus satisfies the condition (ii).
\end{proof}

\subsection{Proof of \Cref{thm:Hamiltonian_Bayes_saddle}} \label{proof:Hamiltonian_Bayes_saddle}

For $M$ in \Cref{asmp:domain_Bayes}, define $G: E_1 \times E_2 \to \R$ by
\begin{align}
	G(f, P) := - \log\left( \int_\Theta \exp\left( - f(\t) \right) \d \P(\t) \right) \label{eq:Function_G}
\end{align}
for all $(f, P) \in M$ and by $G(f, P) := \infty$ for all $(f, P) \not\in M$.
It is easy to verify that $H$ and $G$ satisfies an equality relation $H(f, P) = G(f, P) - \langle f, P \rangle$ for each $(f, P) \in M$ because
\begin{align}
	H(f, \P) & = - \log\left( \int_\Theta \exp( - f(\t) ) \times \exp\left( \int_{\T} f(\t) \d \P(\t) \right) \d \P(\t) \right) \nonumber \\
	& = - \log\left( \int_\Theta \exp( - f(\t) ) \d \P(\t) \times \exp\left( \int_{\T} f(\t) \d \P(\t) \right) \right) \nonumber \\
	& = - \log\left( \int_\Theta \exp\left( - f(\t) \right) \d \P(\t) \right) - \int_{\T} f(\t) \d \P(\t) = G(f, P) - \langle f, P \rangle . \label{eq:H_Hs_equality}
\end{align}
We introduce an intermediate lemma used in the main proof.

\begin{lemma} \label{lem:Function_G_saddle}
	The function $G: E_1 \times E_2 \to \R$ in \eqref{eq:Function_G} is saddle.
\end{lemma}

\begin{proof}[Proof of \Cref{lem:Function_G_saddle}]
	Let $(f_0, P_0) \in M$ be arbitrary.
	Define $M_1(P_0) := \{ f \in E_1 \mid (f, P_0) \in M \}$ and $M_2(f_0) := \{ P \in E_2 \mid (f_0, P) \in M \}$.
	It is clear that $M_1(P_0)$ and $M_2(f_0)$ are convex from definition of $M$.
	We show that (i) $f \mapsto G(f, P_0)$ is concave over $M_1(P_0)$ and (ii) $\P \mapsto G(f_0, \P)$ is convex over $M_2(f_0)$.
	Then $G$ is saddle since $(f_0, P_0)$ is arbitrary.
	
	\vspace{5pt}
	\noindent
	\underline{\textit{Item (i)}}:
	We show that $(*_1) := - G( \alpha f + (1 - \alpha) g, \P_0) \le \alpha ( - G( f, \P_0) ) + (1 - \alpha) ( - G( g, \P_0) )$ holds for any $f, g \in M_1(P_0)$ and any $\alpha \in (0, 1)$.
	Since the edge case $\alpha = 0$ or $\alpha = 1$ is trivial, it suffices to consider $\alpha \in (0, 1)$ without loss of generality.
	First of all, we have
	\begin{align*}
		(*_1) & = \log\left( \int_\T \exp( - \alpha f(\t) - (1 - \alpha ) g(\t) ) \d \P_0(\t) \right) \\
		& = \log\left( \int_\T \exp( - f(\t) )^\alpha \exp( - g(\t) )^{1 - \alpha } \d \P_0(\t) \right) .
	\end{align*}
	It then follows from H\"{o}lder's inequality with $p = 1 / \alpha$ and $q = 1 / (1 - \alpha)$ that
	\begin{align*}
		(*_1) & \le \log\bigg( \left( \int_\T \exp( - f(\t) ) \d \P_0(\t) \right)^{\alpha} \left( \int_\T \exp( - g(\t) ) \d \P_0(\t) \right)^{1 - \alpha} \bigg) \\
		& = \alpha \times \underbrace{ \log\left( \int_\T \exp( - f(\t) ) \d \P_0(\t) \right) }_{ = - G( f, \P_0) } + (1 - \alpha) \times \underbrace{ \log\left( \int_\T \exp( - g(\t) \d \P_0(\t) \right) }_{ = - G( g, \P_0) } 
	\end{align*}
	where we note that $\exp(- f(\t))$ and $\exp(- g(\t))$ are positive everywhere in $\T$.
	
	\vspace{5pt}
	\noindent
	\underline{\textit{Item (ii)}}:
	We show that $(*_2) := G( f_0, \alpha \P + (1 - \alpha) \Q ) \le \alpha G( f_0, \P ) + (1 - \alpha) G( f_0, \Q)$ holds for any $P, Q \in M_2(f_0)$ and any $\alpha \in (0, 1)$.
	Again, since the edge case $\alpha = 0$ or $\alpha = 1$ is trivial, it suffices to consider $\alpha \in (0, 1)$ without loss of generality.
	First of all, we have
	\begin{align*}
		(*_2) & = - \log\left( \alpha \int_\T \exp( - f_0(\t) ) \d \P(\t) + (1 - \alpha)  \int_\T \exp( - f_0(\t) ) \d \Q(\t) \right) .
	\end{align*}
	Because the negative logarithm $- \log(\cdot)$ is convex in $(0, \infty)$, we complete the proof.
\end{proof}

Now we move on to the main proof.

\begin{proof}[Proof of \Cref{thm:Hamiltonian_Bayes_saddle}]
	It follows from \eqref{eq:H_Hs_equality} that $H(f, P) = G(f, P) - \langle f, P \rangle$.
	By \Cref{lem:Function_G_saddle}, the function $G$ is saddle over $M$.
	The bilinear function $(f, P) \mapsto - \langle f, P \rangle$ is immediately saddle over $M$ because linear functions are convex and concave at the same time.
	Since the sum of two saddle functions remains saddle, the function $H$ is saddle.
\end{proof}

\subsection{Proof of \Cref{thm:Hamiltonian_Bayes}} \label{proof:Hamiltonian_Bayes}

\begin{proof}[Proof of \Cref{thm:Hamiltonian_Bayes}]
	It follows from \Cref{thm:Hamiltonian_Bayes_variation} that $( - \partial_2 H(f, P), \partial_1 H(f, P) ) = (f_* + f, P_* - P)$ at each $(f, \P) \in M$.
	The function $H$ is compatible if, at each $(f, \P) \in M$, we have $(f, P) + s ( f_* + f, P_* - P ) \in M$ for all $s \in [0, 1]$ (c.f.~\Cref{def:compatible_H}).
	We verify that $f + s ( f_* + f )$ and $\P + s ( \P_* - P )$ satisfy the condition of $M$ for all $s \in [0, 1]$.
	First, since $f(\t) \ge \log w(\t) \ge 0$ by \Cref{asmp:domain_Bayes} and $f_*(\t) > 0$ by definition, we have
	\begin{align*}
		f(\t) + s ( f_*(\t) + f(\t) ) \ge f(\t) \ge \log w(\t) ,
	\end{align*}
	where the last inequality holds by the condition of $f$ by \Cref{asmp:domain_Bayes} again.
	Second, since $P$ and $\P_*$ are both probability measures, we have
	\begin{align*}
		\int_{\T} \d ( \P + s ( \P_* - \P ) ) (\t) = (1 - s) \int_{\T} \d \P(\t) + s \int_{\T} d \P_*(\t) = (1 - s) \cdot 1 + s \cdot 1 = 1 .
	\end{align*}
	We therefore conclude that $(f + s ( f_* + f ), P + s ( \P_* - P )) \in M$ for all $s \in [0, 1]$.
\end{proof}

\subsection{Proof of \Cref{thm:Hamiltonian_Dyamics_Bayes}} \label{proof:Hamiltonian_Dyamics_Bayes}

We introduce three intermediate lemmas whose proofs are deferred to \Cref{sec:appendix_d} for better presentation.
Denote $Z(f, P) := \int_{\T} \exp(- f(\t)) \d P(\t)$.

\begin{lemma} \label{lem:P_bound}
	Suppose Assumptions \ref{asmp:domain_Bayes}--\ref{asmp:metric_Bayes}.
	At each $(f_0, P_0) \in M$, we have a constant $V_0$ s.t.~$\int_\Theta w(\t) \d | \P |(\t) \le V_0$ holds for all $(f, P) \in B_r(f_0, P_0)$ given any radius $r \in (0, 1]$.
\end{lemma}

\begin{lemma} \label{lem:Z_bound}
	Suppose Assumptions \ref{asmp:domain_Bayes}--\ref{asmp:metric_Bayes}.
	At each $(f_0, P_0) \in M$, we have a constant $W_0$ s.t.~$1 / Z(f, P) \le W_0$ holds for all $(f, P) \in B_r(f_0, P_0)$ given any radius $r \in (0, 1]$. 
\end{lemma}

\begin{lemma} \label{lem:psi_lipschitz}
	Suppose Assumptions \ref{asmp:domain_Bayes}--\ref{asmp:metric_Bayes}.
	At each $(f_0, P_0) \in M$, we have a constant $L_0$ s.t.
	\begin{align*}
		\sup_{ \t \in \Theta } \left| \frac{ \exp(- f(\t)) }{ Z(f, P) } - \frac{ \exp(- g(\t)) }{ Z(g, Q) } \right| \le L_0 \, d((f, P), (g,Q)) 
	\end{align*}
	holds for all $(f, P), (g, Q) \in B_r(f_0, P_0)$ given any radius $r \in (0, 1]$. 
\end{lemma}

\noindent
The proofs are each contained in \Cref{sec:appendix_d_1,sec:appendix_d_2,sec:appendix_d_3}.
We move on to the main proof.

\begin{proof}[Proof of \Cref{thm:Hamiltonian_Dyamics_Bayes}]
	The metric space $(M, d)$ is complete by \Cref{prop:M_completeness}.
	\Cref{thm:Hamiltonian_Dyamics_Bayes} follows from \Cref{thm:existence_conservation} and we hence verify Assumptions \ref{asmp:flow_existence_3} and \ref{asmp:conservation_of_energy_3}.
	We begin with verifying \Cref{asmp:flow_existence_3} item (1) and \Cref{asmp:conservation_of_energy_3}.
	Define a norm $\| f \|_{E_1} := \sup_{\t \in \Theta} | f(\t) | / w(\t)$ on $E_1$ and a norm $\| P \|_{E_2} := \int_{\T} w(\t) \d | P |(\t)$ on $E_2$.
	It is clear that $\varrho$ and $\gamma$ in \eqref{eq:wu_tv_metrics} are norm-induced metrics $\varrho(f, g) = \| f - g \|_{E_1}$ and $\gamma(P, Q) = \| P - Q \|_{E_2}$ by definition. 
	Thus, the metric $d$ on $M$ defined in \Cref{asmp:metric_Bayes} has an extension to a norm-induced metric $d_*( (f, P), (g, Q) ) := \| f - g \|_{E_1} + \| P - Q \|_{E_2}$ on $E_1 \times E_2$.
	This confirms \Cref{asmp:flow_existence_3} item (1).
	Furthermore, by definition of $\| \cdot \|_{E_1}$ and $\| \cdot \|_{E_2}$, we have
	\begin{align*}
		| \langle f, \P \rangle | = \left| \int_{\T} f(\t) \d \P(\t) \right| \le \sup_{ \t \in \Theta } \frac{ | f(\t) | }{ w(\t) } \int_{\T} w(\t) \d | \P |(\t) = \| f \|_{E_1} \| P \|_{E_2}
	\end{align*}
	for any $(f, P) \in E_1 \times E_2$, which confirms \Cref{asmp:conservation_of_energy_3}.
	
	We complete the proof by verifying \Cref{asmp:flow_existence_3} items (2).
	By \Cref{thm:Hamiltonian_Bayes_variation}, at each $(f, P) \in M$ we have $(- \partial_2 H(f, \P), \partial_1 H(f, \P) ) = (f_* + f, P_* - P)$ where
	\begin{align*}
		f_*(\t) := \frac{ \exp(-f(\t)) }{ Z(f, P) } \qquad \text{and} \qquad \d \P_* := \frac{ \exp(-f(\t)) }{ Z(f, P)} \d \P(\t) .
	\end{align*}	
	In what follows, let $(f_0, \P_0)$ be an arbitrary point in $M$. 
	Let $r$ be an arbitrary radius in $(0, 1]$.
	Let $(f, \P), (g, \Q)$ be arbitrary points in the open ball $B_r(f_0, \P_0)$.
	Our aim is to find a constant $K_0$ dependent only on $(f_0, P_0)$ s.t.~we have $d_*( \partial H(f, \P), \partial H(g, \Q) ) \le K_0 d((f, \P), (g, \Q))$ for $(f, \P), (g, \Q) \in B_r(f_0, \P_0)$.
	By the triangle inequality, we have
	\begin{align}
		d_*( \partial H(f, \P), \partial H(g, \Q) ) & = \| (f_*, P_*) + (f, - \P) - ( (g_*, Q_*) + (g, - \Q) ) \|_{E} \nonumber \\
		& \le \| (f_*, P_*) - (g_*, Q_*) \|_{E} + \| (f, \P) - ( g, \Q) \|_{E} \nonumber \\
		& = d_*( (f_*, P_*), (g_*, Q_*) ) + d( (f, \P), (g, \Q) ) . \label{eq:locally_lipschitz_P}
	\end{align}
	It hence suffices to find a constant $K_0^*$ s.t.~we have $d_*( (f_*, P_*), (g_*, Q_*)  \le K_0^* d((f, \P), (g, \Q))$ for $(f, \P), (g, \Q) \in B_r(f_0, \P_0)$, by which the proof is completed with $K_0 = K_0^* + 1$.
	
	We upper bound $d_*( (f_*, P_*), (g_*, Q_*) )$ in the remainder.
	By definition of $d_*$, we have
	\begin{align*}
		d_*( (f_*, P_*), (g_*, Q_*) ) & = \underbrace{ \sup_{\t \in \Theta} \frac{| f_*(\t) - g_*(\t) |}{w(\t)} }_{ = \varrho(f_*, g_*) } + \underbrace{ \int_\T w(\t) \d | \P_* - \Q_* | (\t) }_{ = \gamma(P_*, Q_*) } .
	\end{align*}
	Due to $w(\t) \ge 1$ (c.f.~\Cref{def:weighted_space}) and \Cref{lem:psi_lipschitz}, the first term $\varrho(f_*, g_*)$ is bounded by
	\begin{align}
		\varrho(f_*, g_*) & \le \sup_{\t \in \T} \bigg| \underbrace{ \frac{ \exp(- f(\t)) }{ Z(f, P) } - \frac{ \exp(- g(\t)) }{ Z(g, Q) } }_{ =: h_1(\t) } \bigg| \le L_0 \, d((f, P), (g, Q)) . \label{eq:proof_varrho_b1}
	\end{align}
	For the second term $\gamma(P_*, Q_*)$, we first decompose the term $\P_* - \Q_*$ as follows:
	\begin{align*}
		\d (\P_* - \Q_*) (\t) & = \d \P_*(\t) - \frac{ \exp(- f(\t)) }{ Z(f, P) } \d \Q(\t) + \frac{ \exp(- f(\t)) }{ Z(f, P) } \d \Q(\t) - \d \Q_*(\t) \\
		& = \underbrace{ \frac{ \exp(- f(\t)) }{ Z(f, P) } }_{ =: h_2(\t) } \d ( \P - \Q )(\t) + \bigg( \underbrace{ \frac{ \exp(- f(\t)) }{ Z(f, P) } - \frac{ \exp(- g(\t)) }{ Z(g, Q) } }_{ = h_1(\t) } \bigg) \d \Q(\t) .
	\end{align*}
	By the triangle inequality, the second term $\gamma(P_*, Q_*)$ is then bounded as
	\begin{align*}
		\gamma(P_*, Q_*) & \le \underbrace{ \sup_{ \t \in \Theta } | h_2(\t) | }_{=: (*_a)} \underbrace{ \int_\T w(\t) \d | \P - \Q |(\t) }_{= \gamma(\P, \Q)} + \sup_{\t \in \T} | h_1(\t) | \underbrace{ \int_\T w(\t) \d | \Q | (\t) }_{=: (*_b)} .
	\end{align*}
	We bound the term $(*_a) = \sup_{ \t \in \Theta } | \exp(- f(\t)) | / Z(f, P)$.
	For the numerator, we have $\sup_{\t \in \Theta} | \exp(- f(\t)) | \le 1$ because $f(\t) \ge \log w(\t) \ge 0$ by the condition of $M$.
	For the denominator, \Cref{lem:Z_bound} shows that $1 / Z(f, P) \le W_0$.
	This implies that $(*_a) \le W_0$.
	The term $(*_b)$ is bounded as $(*_b) \le V_0$ by \Cref{lem:P_bound}.
	Finally, it follows from \Cref{lem:psi_lipschitz} that $\sup_{\t \in \T} | h_1(\t) | \le L_0 d((f, P), (g, Q))$.
	By applying these upper bounds, we have
	\begin{align}
		\gamma(P_*, Q_*) & \le W_0 \, \gamma(\P, \Q) + V_0 L_0 \, d((f, P), (g, Q)) \nonumber \\
		& \le \left( W_0 + V_0 L_0 \right) \, d((f, P), (g, Q)) \label{eq:proof_gamma_b2}
	\end{align}
	where we applied a trivial inequality $\gamma(P, Q) \le d((f, P), (g, Q))$ for the last inequality, which follows from definition of $d$.
	Combining each bound \eqref{eq:proof_varrho_b1} and \eqref{eq:proof_gamma_b2}, we arrive at
	\begin{align}
		d_*( \partial H(f, \P), \partial H(g, \Q) ) = \varrho(f_*, g_*) + \gamma(P_*, Q_*) \le K_0^* \, d((f, P), (g,Q)) \label{eq:local_lipschitz_Dynamics_Bayes}
	\end{align}
	where we set $K_0^* = L_0 + W_0 + V_0 L_0$ which is dependent only on $(f_0, P_0)$.
	Since this argument holds for the arbitrary $(f_0, P_0) \in M$, we conclude \Cref{asmp:flow_existence_3} item (2). 
\end{proof}

%% file: appendix_04.tex

\section{Derivation of The First Variations in Section 4} \label{sec:appendix_c}

\Cref{sec:appendix_c} derives the first variations of the minimum free energy in \Cref{thm:Hamiltonian_Bayes_variation}.
For the function $H$, at each $(f, P) \in \Dom(H)$ denote by $D_1 H_{f, \P}$ the variational derivative of the function $g \mapsto H(g, \P)$ at $g = f$.
The first variation $\partial_1 H(f, \P)$ is a unique element in $E_2$ s.t.
\begin{multline}
	D_1 H_{f, \P}(g) = \lim_{\epsilon \to 0^+} \frac{ H( f + \epsilon g, \P ) - H(f, \P) }{ \epsilon } = \langle g, \partial_1 H(f, \P) \rangle \\
	\text{holds for all $g \in E_1$ s.t.~$f + g \in \Dom(H(\cdot, \P))$}. \label{eq:proof_D1H_C1}
\end{multline}
Similarly, at each $(f, P) \in \Dom(H)$ denote by $D_2 H_{f, \P}$ the variational derivative of the function $Q \mapsto H(f, Q)$ at $Q = P$.
The first variation $\partial_2 H(f, \P)$ is a unique element in $E_1$ s.t.
\begin{multline}
	D_2 H_{f, \P}(Q) = \lim_{\epsilon \to 0^+} \frac{ H( f, \P + \epsilon Q ) - H(f, \P) }{ \epsilon } = \langle \partial_2 H(f, \P), Q \rangle \\
	\text{holds for all $Q \in E_2$ s.t.~$\P + Q \in \Dom(H(f, \cdot))$}. \label{eq:proof_D2H_C2}
\end{multline}
The first variations $\partial_1 H(f, \P)$ and $\partial_2 H(f, \P)$ can be found by deriving an explicit form of the variational derivative $D_1 H_{f, \P}$ and $D_2 H_{f, \P}$ that satisfies \eqref{eq:proof_D1H_C1} and \eqref{eq:proof_D2H_C2}.

First, we introduce two intermediate lemmas: the first one in \Cref{sec:appendix_c_1} and the second one in \Cref{sec:appendix_c_2}.
We then provide the main proof of \Cref{thm:Hamiltonian_Bayes_variation} in \Cref{sec:appendix_c_3}.

\subsection{The First Lemma for \Cref{thm:Hamiltonian_Bayes_variation}} \label{sec:appendix_c_1}

We introduce the first intermediate lemma used in the main proof.
For the set $M$ in \Cref{asmp:domain_Bayes}, define $\Hc: E_1 \times E_2 \to \Re$ by
\begin{align}
	\Hc(f, \P) := \int_\Theta \exp( - f(\t) ) \d \P(\t) \label{eq:Hc_function}
\end{align}
for each $(f, \P) \in M$ and by $\Hc(f, \P) := \infty$ for all $(f, \P) \not \in M$.
It is easy to see that $f \mapsto \Hc(f, \P)$ and $P \mapsto \Hc(f, P)$ are each convex, by convexity of $\exp(- \cdot)$ and linearity of integral.
Denote by $D_1 \Hc_{f, \P}(g)$ the variational derivative of the function $g \mapsto \Hc(g, \P)$ at $g = f$.
Denote by $D_2 \Hc_{f, \P}(Q)$ the variational derivative of the function $Q \mapsto \Hc(f, Q)$ at $Q = P$.

\begin{lemma} \label{lem:Hc_derivative}
	For the function $\Hc: E_1 \times E_2 \to \Re$ in \eqref{eq:Hc_function}, we have
	\begin{align*}
		& D_1 \Hc_{f, \P}(g) = - \int_\T g(\t) \exp( - f(\t) ) \d P(\t) ~~~\text{for all}~~~ g \in E_1 ~\text{s.t.}~ f + g \in \Dom(\Hc(\cdot, \P)) , \\
		& D_2 \Hc_{f, \P}(Q) = \int_\T \exp( - f(\t) ) \d Q(\t) ~~~\text{for all}~~~ Q \in E_2 ~\text{s.t.}~ P + Q \in \Dom(\Hc(f, \cdot)) .
	\end{align*}
\end{lemma}

\begin{proof}[Proof of \Cref{lem:Hc_derivative}]
	We begin with deriving $D_1 \Hc_{f, \P}(g)$.
	Suppose $\epsilon \to 0^+$ is a sequence in $[0, 1]$ without loss of generality.
	By definition of the variational derivative,
	\begin{align}
		D_1 \Hc_{f, \P}(g) & = \lim_{\epsilon \to 0^+} \frac{ \Hc( f + \epsilon g, P ) - \Hc( f, P ) }{ \epsilon } \label{eq:Hc_derivative_def} \\
		& = \lim_{\epsilon \to 0^+} \int_{\T} \underbrace{ \frac{ \exp( - ( f(\t) + \epsilon g(\t) ) ) - \exp( - f(\t) ) }{ \epsilon } }_{ =: h_\epsilon(\t) } \d P(\t) . \nonumber
	\end{align}
	Our aim is to apply the dominated convergence theorem for the sequence $h_\epsilon$ indexed by $\epsilon$.
	To do so, we need to show that (i) $h_\epsilon(\t)$ is a pointwise convergent sequence and (ii) $| h_\epsilon(\t) |$ is bounded by some $P$-integrable function for all $\epsilon \in [0, 1]$.
	If item (i) and (ii) hold, we have
	\begin{align*}
		D_1 \Hc_{f, \P}(g) = \lim_{\epsilon \to 0^+} \int_{\T} h_\epsilon(\t) \d \P(\t) = \int_{\T} \lim_{\epsilon \to 0^+} h_\epsilon(\t) \d \P(\t)
	\end{align*}
	by the dominated convergence theorem.
	We show item (i) and (ii) below.
	
	\vspace{5pt}
	\noindent
	\underline{\textit{Item (i)}}: 
	At each fixed $\theta$, we set $\delta := \epsilon g(\t)$ to see that
	\begin{align}
		\lim_{\epsilon \to 0^+} h_\epsilon(\t) & = \left( \lim_{\delta \to 0} \frac{ \exp( - ( f(\t) + \delta ) ) - \exp( - f(\t) ) }{ \delta } \right) g(\t) = - \exp( - f(\t) ) g(\t) \label{eq:h_epsilon_limit}
	\end{align}
	where the last equality holds because the limit of $\delta$ is nothing but the derivative of $\exp( - x )$ at $x = f(\t)$.
	Therefore, $h_\epsilon(\t)$ is a pointwise convergent sequence to $- \exp( - f(\t) ) g(\t)$.
	
	\vspace{5pt}
	\noindent
	\underline{\textit{Item (ii)}}: 
	First of all, we have $f + \epsilon \, g \in \Dom(\Hc(\cdot, \P))$ for all $\epsilon \in [0, 1]$.
	This is because we can rewrite $f + \epsilon \, g$ as a convex combination $(1 - \epsilon) f + \epsilon ( f + g )$ in the convex domain $\Dom(\Hc(\cdot, \P))$, where we recall that $g$ is an element s.t.~$f + g \in \Dom(\Hc(\cdot, \P))$.
	By the condition of the domain of $\Hc$ in \Cref{asmp:domain_Bayes}, we have $f(\t) + \epsilon g(\t) \ge \log w(\t) \ge 0$ for all $\t \in \T$.
	By the mean value theorem, the following bound holds for any scalar $0 \le a, b < \infty$:
	\begin{align*}
		| \exp(- a) - \exp(- b) | \le | a - b | .
	\end{align*}
	At each $\t \in \T$, we apply this bound for $a = f(\t) + \epsilon g(\t)$ and $b = f(\t)$ to see that
	\begin{align*}
		| h_\epsilon(\t) | = \frac{| \exp( - ( f(\t) + \epsilon g(\t) ) ) - \exp( - f(\t) ) |}{ \epsilon } \le | g(\t) | .
	\end{align*}
	The function $| g(\cdot) |$ is integrable with respect to $P$ by definition of $E_1$ and $E_2$.
	We hence established that $| h_\epsilon(\t) |$ is bounded by some $P$-integrable function for all $\epsilon \in [0, 1]$.
	
	\vspace{5pt}
	Therefore, by the dominated convergence theorem and plugging \eqref{eq:h_epsilon_limit} in, we arrive at
	\begin{align*}
		D_1 \Hc_{f, \P}(g) = \int_{\T} \lim_{\epsilon \to 0^+} h_\epsilon(\t) \d \P(\t) = - \int_{\T} g(\t) \exp( - f(\t) ) \d \P(\t) .
	\end{align*}	
	We complete the proof by deriving $D_2 \Hc_{f, \P}(Q)$.
	By definition of the variational derivative and linearity of integral, the variational derivative $D_2 \Hc_{f, \P}(Q)$ immediately satisfies
	\begin{align*}
		\lim_{\epsilon \to 0^+} \frac{ \int_\Theta \exp( - f(\t) ) \d (\P + \epsilon Q)(\t) - \int_\Theta \exp( - f(\t) ) \d \P(\t) }{ \epsilon } = \int_\Theta \exp( - f(\t) ) \d Q(\t) 
	\end{align*}
	which completes the proof.
\end{proof}

\subsection{The Second Lemma for \Cref{thm:Hamiltonian_Bayes_variation}} \label{sec:appendix_c_2}

We introduce the second intermediate lemma used in the main proof.
For the set $M$ in \Cref{asmp:domain_Bayes}, define $G: E_1 \times E_2 \to \Re$ by
\begin{align}
	\Hb(f, \P) := - \log\left( \Hc(f, P) \right) = - \log\left( \int_\Theta \exp( - f(\t) ) \d \P(\t) \right) \label{eq:Hb_function}
\end{align}
for each $(f, \P) \in M$ and by $\Hb(f, \P) := \infty$ for all $(f, \P) \not \in M$.
By \Cref{lem:Function_G_saddle}, $\Hb$ is saddle.
Denote by $D_1 \Hb_{f, \P}(g)$ the variational derivative of the function $g \mapsto \Hb(g, \P)$ at $g = f$.
Denote by $D_2 \Hb_{f, \P}(Q)$ the variational derivative of the function $Q \mapsto \Hb(f, Q)$ at $Q = P$.

\begin{lemma} \label{lem:Hb_derivative}
	For the function $\Hb: E_1 \times E_2 \to \Re$ in \eqref{eq:Hb_function}, we have
	\begin{align*}
		& D_1 \Hb_{f, \P}(g) = \langle g, P_* \rangle ~~~\text{for all}~~~ g \in E_1 ~\text{s.t.}~ f + g \in \Dom(\Hb(\cdot, \P)) , \\
		& D_2 \Hb_{f, \P}(Q) = \langle - f_*, Q \rangle ~~~\text{for all}~~~ Q \in E_2 ~\text{s.t.}~  P + Q \in \Dom(\Hb(f, \cdot))
	\end{align*}
	where $f_* \in E_1$ and $P_* \in E_2$ are the function and measure defined in \Cref{thm:Hamiltonian_Bayes_variation}.
\end{lemma}

\begin{proof}[Proof of \Cref{lem:Hb_derivative}]
	First, we derive $D_1 \Hb_{f, \P}(g)$.
	By definition of the variational derivative and by rearranging the term, we have
	\begin{align*}
		D_1 \Hb_{f, \P}(g) & = - \lim_{\epsilon \to 0^+} \frac{ \log( \Hc(f + \epsilon g, P) ) - \log( \Hc(f, P) ) }{ \epsilon } \\
		& = - \lim_{\epsilon \to 0^+} \frac{ \log( \Hc(f + \epsilon g, P) ) - \log( \Hc(f, P) ) }{ \Hc(f + \epsilon g, P) - \Hc(f, P) } \cdot \frac{ \Hc(f + \epsilon g, P) - \Hc(f, P) }{\epsilon} 
	\end{align*}
	Set $\delta := \Hc(f + \epsilon g, P) - \Hc(f, P)$.
	Notice that $\lim_{\epsilon \to 0^+} ( \delta / \epsilon )$ corresponds to of the variational derivative $D_1 \Hc_{f, P}(g)$ as in \eqref{eq:Hc_derivative_def}.
	We have $\delta \to 0$ in the limit $\epsilon \to 0^+$ because $\lim_{\epsilon \to 0^+} ( \delta / \epsilon )$ exists from \Cref{lem:Hc_derivative}.
	Therefore, the above limit can be further written as
	\begin{align*}
		D_1 \Hb_{f, \P}(g) & = - \underbrace{ \lim_{\delta \to 0} \frac{ \log( \Hc(f, P) + \delta ) - \log( \Hc(f, P) ) }{ \delta } }_{(*_a)} \cdot \underbrace{ \lim_{\epsilon \to 0^+} \frac{ \Hc(f + \epsilon g, P) - \Hc(f, P) }{\epsilon} }_{ = D_1 \Hc_{f, P}(g) } .
	\end{align*}
	The limit $(*_a)$ is nothing but the derivative of $\log(x)$ at $x = \Hc(f, P)$, meaning that we have $(*_a) = 1 / \Hc(f, P)$.
	It then follows from \Cref{lem:Hc_derivative} that we have
	\begin{align*}
		D_1 \Hb_{f, \P}(g) & = - \frac{1}{\Hc(f, P)} \cdot D_1 \Hc_{f, P}(g) = \int_{\T} g(\t) \frac{ \exp(- f(\t)) }{ \int_{\T} \exp(- f(\t)) \d P(\t) } \d P(\t) .
	\end{align*}
	By definition of the duality $\langle \cdot, \cdot \rangle$, this establishes that $D_1 \Hb_{f, \P}(g) = \langle g, P_* \rangle$.
	Here it is clear that $P_* \in E_2$ since $\exp(- f(\t))$ is bounded as $\exp(- f(\t)) \le 1$ by the condition of $M$.
	
	Second, we derive $D_2 \Hb_{f, \P}(Q)$.
	By definition of the variational derivative, we have
	\begin{align*}
		D_2 \Hb_{f, \P}(Q) & = - \lim_{\epsilon \to 0^+} \frac{ \log( \Hc(f, P + \epsilon Q) ) - \log( \Hc(f, P) ) }{ \epsilon } \\
		& = - \lim_{\epsilon \to 0^+} \frac{ \log( \Hc(f, P + \epsilon Q) ) - \log( \Hc(f, P) ) }{ \Hc(f, P + \epsilon Q) - \Hc(f, P) } \cdot \frac{ \Hc(f, P + \epsilon Q) - \Hc(f, P) }{ \epsilon } .
	\end{align*}
	By abuse of notation, set $\delta := \Hc(f, P + \epsilon Q) - \Hc(f, P)$.
	Notice again that $\lim_{\epsilon \to 0^+} ( \delta / \epsilon )$ corresponds to of the variational derivative $D_2 \Hc_{f, P}(Q)$.
	We have $\delta \to 0$ as $\epsilon \to 0^+$ because $\lim_{\epsilon \to 0^+} ( \delta / \epsilon )$ exists from \Cref{lem:Hc_derivative}.
	Therefore, the above limit can be written as
	\begin{align*}
		D_2 \Hb_{f, \P}(Q) & = - \underbrace{ \lim_{\delta \to 0} \frac{ \log( \Hc(f, P) + \delta ) - \log( \Hc(f, P) ) }{ \delta } }_{ =:(*_b) } \cdot \underbrace{ \lim_{\epsilon \to 0^+} \frac{ \Hc(f, P + \epsilon Q) - \Hc(f, P) }{ \epsilon } }_{ = D_2 \Hc_{f, P}(Q) } .
	\end{align*}
	The term $(*_b)$ equals to the term $(*_a)$ above, meaning that we have $(*_b) = 1 / \Hc(f, P)$.
	Since $D_2 \Hc_{f, P}(Q) = \int_{\T} \exp(- f(\t)) \d Q(\t)$ by \Cref{lem:Hc_derivative}, we have
	\begin{align*}
		D_2 \Hb_{f, \P}(Q) & = \frac{- 1}{\Hc(f, P)} \cdot \int_{\T} \exp(- f(\t)) \d Q(\t) = \int_{\T} \frac{ - \exp(- f(\t)) }{ \int_{\T} \exp(- f(\t)) \d P(\t) } \d Q(\t) .
	\end{align*}
	By definition of the duality $\langle \cdot, \cdot \rangle$, this establishes that $D_2 \Hb_{f, \P}(Q) = \langle - f_*, Q \rangle$.
	Here it is clear that $f_* \in E_1$ since $\exp(- f(\t))$ is bounded as $\exp(- f(\t)) \le 1$ as aforementioned.
\end{proof}

\subsection{The Main Proof of \Cref{thm:Hamiltonian_Bayes_variation}} \label{sec:appendix_c_3}

Finally, we provide the main proof of \Cref{thm:Hamiltonian_Bayes_variation}.
The proof uses the second lemma in \Cref{sec:appendix_c_2}, where the first lemma in \Cref{sec:appendix_c_1} was implicitly used as an intermediate result for the second lemma in \Cref{sec:appendix_c_2}.

\begin{proof}[Proof of \Cref{thm:Hamiltonian_Bayes_variation}]
	As shown in \eqref{eq:H_Hs_equality}, we have $H(f, P) = \Hb(f, P) - \langle f, P \rangle$ for the function $\Hb$ in \eqref{eq:Hb_function}.
	First, we derive the variational derivative $D_1 H_{f, \P}(g)$.
	We have
	\begin{align*}
		D_1 H_{f, \P}(g) & = \lim_{\epsilon \to 0^+} \frac{ G( f + \epsilon g, \P ) - G(f, \P) }{ \epsilon } - \lim_{\epsilon \to 0^+} \frac{ \langle f + \epsilon g, \P \rangle - \langle f, \P \rangle }{ \epsilon } .
	\end{align*}
	Notice that the second limit immediately equals to $\langle g, P \rangle$ because the function $f \mapsto \langle f, P \rangle$ is linear due to the bilinearlity of $\langle \cdot, \cdot \rangle$.
	The first limit is nothing but the variational derivative $D_1 \Hb_{f, \P}(g)$ of the function $\Hb$ in \eqref{eq:Hb_function}.
	It hence follows from \Cref{lem:Hb_derivative} that
	\begin{align*}
		D_1 H_{f, \P}(g) & = \langle g, P_* \rangle - \langle g, P \rangle = \langle g, P_* - P \rangle .
	\end{align*}
	This establishes the first variation $\partial_1 H(f, \P) = P_* - P$ by definition of the first variation.
	
	Second, we derive the variational derivative $D_2 H_{f, \P}(g)$.
	We have
	\begin{align*}
		D_2 H_{f, \P}(Q) & = \lim_{\epsilon \to 0^+} \frac{ G( f, \P + \epsilon Q ) - G(f, \P) }{ \epsilon } - \lim_{\epsilon \to 0^+} \frac{ \langle f, \P + \epsilon Q \rangle - \langle f, \P \rangle }{ \epsilon } .
	\end{align*}
	Again, notice that the second limit immediately equals to $\langle f, Q \rangle$ because the function $P \mapsto \langle f, P \rangle$ is linear due to the bilinearlity of $\langle \cdot, \cdot \rangle$.
	The first limit is nothing but the variational derivative $D_2 \Hb_{f, \P}(g)$ of the function $\Hb$ in \eqref{eq:Hb_function}.
	It hence follows from \Cref{lem:Hb_derivative} that
	\begin{align*}
		D_2 H_{f, \P}(Q) & = \langle - f_*, Q \rangle - \langle f, Q \rangle = \langle - f_* - f, Q \rangle .
	\end{align*}
	This establishes the first variation $\partial_2 H(f, \P) = - f_* - f$ by definition of the first variation.
	It is equivalent to that $- \partial_2 H(f, \P) = f_* + f$, which completes the proof.
\end{proof}

%% file: appendix_05.tex

\section{Intermediate Lemmas for Theorem 4} \label{sec:appendix_d}

This section contains proofs of the three intermediate lemmas, Lemmas \ref{lem:P_bound}--\ref{lem:psi_lipschitz}, used in \Cref{proof:Hamiltonian_Dyamics_Bayes}.
Note that  $Z(f, \P) := \int_{\T} \exp(- f(\t)) \d P(\t)$ throughout \Cref{sec:appendix_d}.

\subsection{The Proof of \Cref{lem:P_bound}} \label{sec:appendix_d_1}

\begin{proof}[Proof of \Cref{lem:P_bound}]
	Define a norm $\| P \|_{E_2} := \int_\Theta w(\t) \d | \P |(\t)$ on $E_2$.
	By the reverse triangle inequality, we have $\| P \|_{E_2} - \| P_0 \|_{E_2} \le \| P - P_0 \|_{E_2} = \gamma(P, P_0)$.
	We apply a trivial inequality $\gamma(P, P_0) \le d((f, P), (f_0, P_0))$, which follows from definition of $d$, to see that
	\begin{align*}
		\int_\Theta w(\t) \d | \P |(\t) - \int_\Theta w(\t) \d | \P_0 |(\t) = \| P \|_{E_2} - \| P_0 \|_{E_2} \le d((f, P), (f_0, P_0)) < r
	\end{align*}
	where the last inequality follows from $(f, P) \in B_r(f_0, P_0)$.
	This implies $\int_\Theta w(\t) \d | \P |(\t) \le \int_\Theta w(\t) \d | \P_0 |(\t) + r$.
	Since $r \le 1$, setting $V_0 = \int_\Theta w(\t) \d | \P_0 |(\t) + 1$ completes the proof.
\end{proof}

\subsection{The Proof of \Cref{lem:Z_bound}} \label{sec:appendix_d_2}

\begin{proof}[Proof of \Cref{lem:Z_bound}]
	Since $x \mapsto \exp(- x)$ is convex, Jensen's inequality implies that
	\begin{align*}
		Z(f, P) = \int_{\T} \exp( - f(\t) ) \d \P(\t) \ge \exp\left( - \int_{\T} f(\t) \d \P(\t) \right) 
	\end{align*}
	This lower bound is equivalent to an upper bound $1 / Z(f, P) \le \exp( \int_{\T} f(\t) \d \P(\t) )$. 
	Define a norm $\| f \|_{E_1} := \sup_{ \t \in \Theta } | f(\t) | / w(\t)$ and a norm $\| P \|_{E_2} := \int_{\T} w(\t) \d | P |(\t)$.
	We upper bound the term $(*) := - \int_{\T} f(\t) \d \P(\t)$ inside the exponential function as follows:
	\begin{align*}
		| (*) | & \le \sup_{ \t \in \Theta } \frac{ | f(\t) | }{ w(\t) } \int_{\T} w(\t) \d | P |(\t) = \| f \|_{E_1} \| P \|_{E_2} .
	\end{align*}
	It follows from the triangle inequality that
	\begin{align*}
		| (*) | & \le ( \| f - f_0 \|_{E_1} + \| f_0 \|_{E_1} ) ( \| P - P_0 \|_{E_2} + \| P_0 \|_{E_2}) .
	\end{align*}
	We have trivial bounds $\| f - f_0 \|_{E_1} \le d((f, P), (f_0, P_0))$ and $\| P - P_0 \|_{E_2} \le d((f, P), (f_0, P_0))$ by definition of $d$.
	Since $(f, P) \in B_r(f_0, P_0)$ and $r \le 1$, we further have $\| f - f_0 \|_{E_1} \le 1$ and $\| P - P_0 \|_{E_2} \le 1$.
	Plugging this in the upper bound of $1 / Z(f, P)$, we establish that
	\begin{align*}
		\frac{1}{Z(f, P)} \le \exp\left( (1 + \| f_0 \|_{E_1}) (1 + \| P_0 \|_{E_2}) \right) .
	\end{align*}
    Setting the right-hand side to $W_0$, which depends only on $(f_0, P_0)$, completes the proof.
\end{proof}

\subsection{The Proof of \Cref{lem:psi_lipschitz}} \label{sec:appendix_d_3}

\begin{proof}[Proof of \Cref{lem:psi_lipschitz}]
	Let $(f, P), (g, Q)$ be arbitrary points in $B_r(f_0, P_0)$.
	Denote $(*) := \sup_{\t \in \Theta} | \exp(- f(\t)) / Z(f, \P) - \exp(- g(\t)) / Z(g, \Q) |$.
	By the triangle inequality,
	\begin{align*}
		(*) & = \sup_{\t \in \Theta} \left| \frac{ \exp(- f(\t)) }{ Z(f, \P) } - \frac{ \exp(- g(\t)) }{ Z(f, \P) } + \frac{ \exp(- g(\t)) }{ Z(f, \P) } - \frac{ \exp(- g(\t)) }{ Z(g, \Q) } \right| \\
		& \le \frac{ 1 }{ Z(f, \P) } \sup_{\t \in \Theta} | \exp(- f(\t)) - \exp(- g(\t)) | + \sup_{\t \in \Theta} | \exp(- g(\t)) | \left| \frac{ 1 }{ Z(f, \P) } - \frac{ 1 }{ Z(g, \Q) } \right| \\
		& = \frac{ 1 }{ Z(f, \P) } \underbrace{ \sup_{\t \in \Theta} | \exp(- f(\t)) - \exp(- g(\t)) | }_{=:(*_1)} + \frac{ \sup_{\t \in \Theta} | \exp(g(\t)) | }{ Z(f, P) Z(g, Q) } \underbrace{ \vphantom{\sup_{\t \in \Theta}} \left| Z(f, \P) - Z(g, \Q) \right| }_{=:(*_2)} .
	\end{align*}
	Denote $(*_a) := \sup_{\t \in \Theta} | \exp(- g(\t)) |$.
	We have $(*_a) \le 1$ since $g(\t) \ge \log w(\t) \ge 0$ by the condition of $M$.
	\Cref{lem:Z_bound} shows that $1 / Z(f, P) \le W_0$ for all $(f, P) \in B_r(f_0, P_0)$.
	Therefore
	\begin{align}
		(*) & \le W_0 \, (*_1) + W_0^2 \, (*_2) . \label{eq:proof_psi_bound}
	\end{align}
	We upper bound each term $(*_1)$ and $(*_2)$ in the remainder.
	
	\vspace{5pt}
	\noindent
	\underline{\textit{Term $(*_1)$}}:
	The mean value theorem implies $| \exp(a) - \exp(b) | \le \max( \exp(a), \exp(b) ) | a - b |$ for any scalars $a, b \in \R$.
	At each $\t \in \T$, we apply it for $a = - f(\t)$ and $b = - g(\t)$ to see
	\begin{align*}
		| \exp(- f(\t)) - \exp(- g(\t)) | \le \max( \exp(- f(\t)), \exp(- g(\t)) ) | f(\t) - g(\t) | .
	\end{align*}
	We have $\exp( - f(\t) ) \le \exp( - \log w(\t) ) = 1 / w(\t)$ by the condition of $M$, and the same holds for $\exp( - g(\t) )$.
	We apply $\max( \exp(- f(\t)), \exp(- g(\t)) ) \le 1 / w(\t)$ to see that
	\begin{align*}
		| \exp(- f(\t)) - \exp(- g(\t)) | \le \frac{ | f(\t) - g(\t) | }{ w(\t) } .
	\end{align*}
	It then follows from taking the supremum of both the sides that
	\begin{align*}
		(*_1) = \sup_{ \t \in \Theta } | \exp(- f(\t)) - \exp(- g(\t)) | \le \sup_{ \t \in \Theta } \frac{ | f(\t) - g(\t) | }{ w(\t) } \le d((f, P), (g, Q)) 
	\end{align*}
	where the last inequality is trivial from definition of $d$.
	
	\vspace{5pt}
	\noindent
	\underline{\textit{Term $(*_2)$}}:
	By the triangle inequality, we decompose the term $(*_2)$ as follows:
	\begin{align*}
		(*_2) & = \left| Z(f, P) - \int_\Theta \exp(- g(\t)) \d P (\t) + \int_\Theta \exp(- g(\t)) \d P (\t) - Z(g, Q) \right| \\
		& \le \underbrace{ \sup_{ \t \in \Theta } | \exp( - f(\t) ) - \exp( - g(\t) ) | }_{= (*_1)} \underbrace{ \int_\Theta \d | \P |(\t) }_{ =: (*_b) } + \underbrace{ \sup_{ \t \in \Theta } | \exp( - g(\t) ) | }_{ =: (*_a) } \underbrace{ \int_\Theta \d | \P - \Q |(\t) }_{= (*_c)} 
	\end{align*}
	We have $(*_1) \le d((f, P), (g, Q))$ by the above argument.
	We have $(*_a) \le 1$ as aforementioned.
	Since $w(\t) \ge 1$, we have $(*_b) \le \int_{\T} w(\t) d | P |(\t)$ and $(*_c) \le \int_{\T} w(\t) d | P - Q |(\t)$.
	We further apply \Cref{lem:P_bound} to see that $(*_b) \le V_0$.
	By these inequalities, we have
	\begin{align*}
		(*_2) & \le V_0 \, d((f, P), (g, Q)) + \int_{\T} w(\t) d | P - Q |(\t) \le (V_0 + 1) \, d((f, P), (g, Q)) 
	\end{align*}
	where the last inequality is trivial from definition of $d$.
	
	\vspace{5pt}
	By plugging the upper bounds of each term $(*_1)$ and $(*_2)$ in \eqref{eq:proof_psi_bound}, we have
	\begin{align*}
		(*) & \le W_0 \, d((f, P), (g, Q)) + W_0^2 \, (V_0 + 1) \, d((f, P), (g, Q)) .
	\end{align*}
	Setting $L_0 = W_0 + W_0^2 (V_0 + 1)$ completes the proof.
\end{proof}